\documentclass[11pt]{article}
\usepackage[margin=1in]{geometry}
\usepackage{amsmath,amssymb,amsthm}
\usepackage{bbm}
\usepackage{stmaryrd}
\usepackage{graphicx}
\usepackage{paralist}
\usepackage{latexsym}
\usepackage{color}
\usepackage{varioref}
\usepackage{enumerate}
\usepackage{pdfsync}
\usepackage{enumitem}
\usepackage{hyperref}
\usepackage{cancel}

\renewenvironment{thebibliography}[1]{%
\begin{oldthebibliography}{#1}%
\setlength{\baselineskip}{.8em}
\linespread{.8}
\small
\setlength{\parskip}{0ex}%
\setlength{\itemsep}{.1em}%
}%
{%
\end{oldthebibliography}%
}

\definecolor{darkred}{rgb}{0.8,0,0}

\makeatletter \@addtoreset{equation}{section}

\def \benumlab#1{\begin{enumerate}[label={\rm (#1{\arabic{*}})}, ref={\rm #1{\arabic{*}}}]}
\def \enumlab{\end{enumerate}}
\def \benumlabi#1{\begin{enumerate}[label={\rm {(#1\roman{*})}}, ref={\rm{#1\roman{*}}}]}
\def \enumlabi{\end{enumerate}}

\def\proofof#1{\noindent\emph{Proof of~#1.}}

\def\strict{(\mbox{strict})}

\def\on{\quad\mbox{on}\quad}
\def\et{\quad\mbox{and}\quad}

  \def\xib{\xi\hspace{-1.7mm}\xi}
  \def\xibu{\xi\hspace{-1.4mm}\xi}
  \def\abs#1{\left|#1\right|}

\def \Liminf{\displaystyle\liminf}
\def \Limsup{\displaystyle\limsup}

\def \E{\mathbb{E}}
\def \F{\mathbb{F}}

\def \M{\mathbb{M}}
\def \N{\mathbb{N}}
\def \R{\mathbb{R}}

\def\P{\mathbb{P}}
\def\Q{\mathbb{Q}}
\def\S{\mathbb{S}}

\def\Cc{{\cal C}}

\def\Fc{{\cal F}}

\def\Hc{{\cal H}}
\def\Ic{{\cal I}}

\def\Lc{{\cal L}}

\def\Oc{{\cal O}}

\def\Qc{{\cal Q}}
\def\Rc{{\cal R}}

\def\Yc{{\cal Y}}

\def\ep{\hbox{ }\hfill$\Box$}
\def\reff#1{{\rm(\ref{#1})}}

\def\be{\begin{eqnarray}}
\def\ee{\end{eqnarray}}
\def\bal{\begin{aligned}}
\def\eal{\end{aligned}}
\def\beq{\begin{equation}}
\def\eeq{\end{equation}}
\def\beqq{\begin{equation*}}
\def\eeqq{\end{equation*}}
\def\b*{\begin{eqnarray*}}
\def\e*{\end{eqnarray*}}

\def\x{\times}

\def\={\;=\;}

\def\pourtout{\mbox{ for all } }

\def\.{\;.}

\def\eps{\varepsilon}
\def\veps{\varepsilon}
\def\vp{\varphi}
\def\1{\mathbf{1}}
\def\Esp#1{\mathbb{E}\left[#1\right]}

\def\Tr#1{\mbox{\rm Tr}\left[#1\right]}

\def\uf{\mathfrak{u}}
\def\nf{\mathfrak{n}}
\def\vf{\mathfrak{v}}

\def\a{a\hspace{-0.65mm}\iota}

\def\lpt#1{\\[#1pt]}

\def\Pf{\mathfrak P}

\def\Df{\mathfrak D}
\def\Dfi{\mathfrak \Df_{<}}
\def\Dfb{\mathfrak \partial_T\Df}

\def\vc{\vartheta}

\theoremstyle{plain}
\newtheorem{theorem}{Theorem}[section]
\newtheorem{proposition}[theorem]{Proposition}
\newtheorem{corollary}[theorem]{Corollary}
\newtheorem{assumption}[theorem]{Assumption}
\newtheorem{definition}[theorem]{Definition}
\newtheorem{lemma}[theorem]{Lemma}
\theoremstyle{definition}
\newtheorem{remark}[theorem]{Remark}

\labelformat{assumptionCU}{(${\underline{{\rm\bf C}}}$)}
\labelformat{assumptionCV}{(${{{\rm\bf C}}}$)}

\newtheorem*{assumptionA}{{\textbf Assumption A}}
\newtheorem*{assumptionB}{{\textbf Assumption B}}

\begin{document}

\title{Trading with Small Price Impact\footnote{We thank Peter Bank, Alex M.\ G.\ Cox, Paolo Guasoni, Jan Kallsen, Shen Li, Ren Liu, Dylan Possamai, Vilmos Prokaj, Mathieu Rosenbaum, Peter Tankov, and Marko Weber for fruitful discussions. We are also very grateful to two anonymous reviewers for their extremely careful reading and numerous very constructive remarks.}}
\author{
     Ludovic Moreau
     \thanks{ETH Z\"urich, Departement f\"ur Mathematik, R\"amistrasse 101, CH-8092, Z\"urich, Switzerland, email \texttt{ludovic.moreau@math.ethz.ch}. Partially supported by the Swiss National Science Foundation under the grant SNF $200021\_152555$ and by the ETH Foundation.
     }
     \and Johannes Muhle-Karbe
     \thanks{ETH Z\"urich, Departement f\"ur Mathematik, R\"amistrasse 101, CH-8092, Z\"urich, Switzerland, and Swiss Finance Institute, email \texttt{johannes.muhle-karbe@math.ethz.ch}. Partially supported by the ETH Foundation.
     }
\and H.~Mete Soner
     \thanks{ETH Z\"urich, Departement f\"ur Mathematik, R\"amistrasse 101, CH-8092, Z\"urich, Switzerland, and Swiss Finance Institute, email \texttt{mete.soner@math.ethz.ch}. Partially supported by the Swiss National Science Foundation under the grant SNF $200021\_152555$. 
  }
}
\date{\today}

\maketitle \vspace{-1em}

\begin{abstract}
An investor trades a safe and several risky assets with linear price impact to maximize expected utility from terminal wealth. In the limit for small impact costs, we explicitly determine the optimal policy and welfare, in a general Markovian setting allowing for stochastic market, cost, and preference parameters. These results shed light on the general structure of the problem at hand, and also unveil close connections to optimal execution problems and to other market frictions such as proportional and fixed transaction costs.
\end{abstract}




\section{Introduction}

Even in the most liquid financial markets, only small quantities can be traded quickly without adversely affecting market prices. For large investors, it is therefore crucial to balance the gains generated by trading against the corresponding price impact costs.

This problem has received a lot of attention in the optimal execution literature, which studies how to efficiently split up a single exogenously given order (cf., e.g., \cite{bertsimas.lo.98,almgren.chriss.01,huberman.stanzl.05,obizhaeva.wang.13} as well as many more recent studies). In contrast, less is known about dynamic portfolio choice with price impact, i.e., the problem of how to endogenously determine the optimal order flow from market dynamics and investors' preferences. Here, previous work has focused on price impact linear in the order size, in concrete models with specific market dynamics and preferences \cite{garleanu.pedersen.13a,garleanu.pedersen.13b,almgren.li.11,dufresne.al.12,guasoni.weber.12,guasoni.weber.13}; see Section~\ref{sec:lit} for a detailed discussion. In the present study, we also focus on linear price impact. However, we allow for arbitrary preferences, as well as for general Markovian dynamics of market prices and impact parameters. Despite this generality, we obtain explicit formulas for the optimal policy and welfare, \emph{asymptotically} for small price impacts.

These results shed new light on the general structure of the problem at hand, and also reveal deep connections to other market frictions. As in previous studies \cite{garleanu.pedersen.13a,garleanu.pedersen.13b, almgren.li.11, guasoni.weber.12,guasoni.weber.13}, it turns out to be optimal to always trade from the current position $\theta^\Lambda_t$ towards the frictionless target $\theta^0_t$ at a finite rate $\dot{\theta}^\Lambda_t$. For a single risky asset,\footnote{The results readily extend to multiple risky assets, cf.\ Theorems~\ref{theo: main result} and \ref{theo: main result 2}. For ease of exposition, we focus on a single risky asset in this introduction.} traded with small linear price impact $\Lambda_t$, this asymptotically optimal trading rate is given explicitly by:
\begin{equation}\label{eq:rate}
\dot{\theta}^\Lambda_t= \sqrt{\frac{(\sigma^S_t)^2}{2 \Lambda_t R_t}}(\theta^0_t-\theta^\Lambda_t).
\end{equation}
Here, $\sigma^S_t$ is the risky asset's volatility and $R_t$ is the frictionless investor's ``indirect risk-tolerance process'', i.e., the risk tolerance of the frictionless value function. Thus, the current position $\theta^\Lambda_t$ is pushed back more aggressively to the frictionless target $\theta^0_t$ if i) the current deviation $\theta^\Lambda_t-\theta_t$ is large, ii) market volatility $\sigma^S_t$ is high, iii) trading costs $\Lambda_t$ are low, or iv) the investor's risk tolerance $R_t$ is low. For constant market, cost, and preference parameters, this reduces to the formulas obtained by \cite{ garleanu.pedersen.13b,almgren.li.11,guasoni.weber.12}. In the general setting considered here, these quantities are updated continuously with the current volatility, price impact, and (indirect) risk tolerance. Hence, the optimal policy is ``myopic'' in the sense that it trades towards the current frictionless maximizer (rather a projected future optimum) with a speed determined by current market and preference parameters.\footnote{Hedging against the future evolution of the frictionless target is studied by Garleanu and Pedersen \cite{garleanu.pedersen.13a,garleanu.pedersen.13b}.}

This observation is in analogy to results for small proportional \cite{martin.12,soner.touzi.13,kallsen.muhlekarbe.13a,kallsen.muhlekarbe.13,kallsen.li.13} and fixed transaction costs \cite{korn.98,altarovici.al.13}, where ``myopic'' policies are also optimal asymptotically. With these frictions, the risky fraction is \emph{always} kept between two trading boundaries around the frictionless target position. In contrast, with price impact, it is no longer optimal to remain uniformly close. Instead, the optimal deviation follows a diffusion process with fluctuations driven by the frictionless optimizer and mean reversion induced by the control \eqref{eq:rate}. Hence, the ``fine'' structure of the optimal policy crucially depends on the specific market friction under consideration. Yet, the ``coarse'' structure is the same in each case, in that the \emph{average} squared deviation from the frictionless target is kept below some threshold, determined by the same inputs.\footnote{This is the (leading-order) stationary variance obtained when considering a small time interval around $t$, and then i) changing time to stretch it to the entire half-line, and ii) normalizing the deviation by the dynamic threshold accordingly. See \cite{kallsen.muhlekarbe.13,kallsen.li.13} for more details.} Indeed, with small linear price impact $\Lambda_t$, this threshold is given by:
$$\sqrt{2}\left(\frac{R_t \Lambda_t}{(\sigma^S_t)^2}\right)^{1/2}\left(\sigma^{\theta^0}_t\right)^2,$$
where $\sigma_t^{\theta^0}=\sqrt{d\langle\theta^0\rangle_t/dt}$ is the volatility of the frictionless target strategy.\footnote{If $\theta^0_t=\Delta(t,S_t)$ is a delta-hedge in a complete Markovian setting then this is the ``Cash-Gamma'', i.e., the second derivative of the option price with respect to the underlying, multiplied by the squared value of the latter.} For small proportional transaction costs $\Lambda_t$, the analogous bound reads as follows:\footnote{This bound is derived by noticing that the deviations from the frictionless target are approximately uniform in this case \cite{janecek.shreve.04,rogers.04,goodman.ostrov.10,kallsen.muhlekarbe.13a,kallsen.muhlekarbe.13,kallsen.li.13}, so that the corresponding average squared deviation equals one third of the halfwidth of the no-trade region determined in \cite{martin.12,soner.touzi.13,kallsen.muhlekarbe.13a,kallsen.muhlekarbe.13,kallsen.li.13}.}
$$\frac{1}{\sqrt[3]{12}}\left(\frac{R_t \Lambda_t}{(\sigma^S_t)^2}\right)^{2/3}\left(\sigma^{\theta^0}_t\right)^{4/3}.$$
Similarly, for small fixed trading costs $\Lambda_t$, the corresponding threshold is given by:\footnote{To see this, note that the approximate probability density of the deviation is a ``hat function'' in this case, so that the corresponding average squared deviation is given by one sixth of the halfwidth of the no-trade region determined by \cite{korn.98,altarovici.al.13}.}
$$\frac{1}{\sqrt{3}}\left(\frac{R_t \Lambda_t}{(\sigma^S_t)^2}\right)^{1/2}\sigma^{\theta^0}_t.$$
Hence, there is a different universal constant in each case, and the powers to which the input parameters are raised also depend on the specific friction at hand. The inputs $R_t$, $\Lambda_t$, $\sigma^S_t$, and $\sigma^{\theta^0}_t$, however, are the same in each model. As a result, the corresponding comparative statics are universal: the frictionless target is tracked tightly, on average, if price risk is high relative to risk tolerance, if trading costs are low, or if the frictionless target strategy is relatively inactive and can therefore be implemented with few adjustments.

The optimal trading rate \eqref{eq:rate} also reveals a close connection to the optimal execution literature. Indeed, for small price impacts, \eqref{eq:rate} locally corresponds to the optimal execution strategy of Almgren and Chriss \cite{almgren.chriss.01} as well as Schied and Sch\"oneborn \cite{schied.schoeneborn.09}, with the order to be executed given by the deviation from the frictionless target.\footnote{This correspondence remains true with several risky assets, where optimal liquidation has been studied by \cite{schied.al.10,schoeneborn.11}.} Hence, dynamic portfolio choice with small price impacts can be interpreted as ``optimally liquidating towards the frictionless target'', where the latter as well as market, impact, and preference parameters all are updated continuously.

The performance of the optimal policy and in turn the welfare loss due to finite market depth can also be quantified. At the leading order, the certainty equivalent loss due to small price impact, i.e., the cash equivalent of trading without frictions, is given by:
\begin{equation}\label{eq:cevloss}
\mathbb{E}_\Q\left[\int_0^T \sqrt{\frac{(\sigma^S_t)^2 \Lambda_t}{2R_t}} \left(\sigma^{\theta^0}_t\right)^2 dt\right].
\end{equation}
As a result, price impact has a substantial welfare effect if i) market risk measured by the volatility $\sigma^S_t$ is high compared to the investor's risk tolerance $R_t$, ii) the trading costs $\Lambda_t$ are large, or iii) the frictionless target strategy is highly active with large volatility $\sigma^{\theta^0}_t$. As all of these quantities generally are time-dependent and random, they have to be averaged suitably, across both time and states. Here, averaging across states is carried out with respect to the frictionless investor's ``marginal pricing measure'' $\Q$,\footnote{That is, the dual martingale measure linked to the primal optimizer by the usual first-order condition. Expectations under this measure correspond to utility indifference prices for infinitesimally small claims \cite{davis.97,karatzas.kou.96,kramkov.sirbu.07}, whence the name ``marginal pricing measure''.} i.e., the effect of the small friction is priced like a marginal path-dependent option.

For frictionless models that can be solved in closed form, Representation~\eqref{eq:cevloss} readily yields explicit formulas. In general, this expression allows to shed further light on the connections between price impact and other market frictions. Indeed, close analogues of Formula \eqref{eq:cevloss} for the certainty equivalent loss due to small price impact remain true for different trading costs. Only the universal constant and the powers of the inputs have to be changed, like for the average squared deviation from the frictionless target. For example, with small proportional transaction costs $\Lambda_t$, the analogue of \eqref{eq:cevloss} reads as follows \cite{soner.touzi.13,kallsen.muhlekarbe.13a,kallsen.muhlekarbe.13}:
$$\mathbb{E}_\Q\left[\int_0^T \sqrt[3]{\frac{9(\sigma^S_t)^2 \Lambda_t}{32R_t}} \left(\sigma^{\theta^0}_t\right)^{4/3} dt\right].$$
Hence, the monotonicity in the model inputs $\sigma^S_t$, $\Lambda_t$, $R_t$, and $\sigma^{\theta^0}$ remains unchanged, and the corresponding comparative statics are the same for each small friction.

For investors with constant absolute risk tolerance, i.e., with exponential utilities, our results readily allow to incorporate random endowments by a change of measure. This in turn allows us to obtain utility-indifference prices and hedging strategies. As volatilities are invariant under equivalent measure changes, it follows that the trading rate \eqref{eq:rate} is truly universal, in that it applies both to optimal investment and to hedging; only the frictionless inputs need to be changed accordingly. Formula \eqref{eq:cevloss} for the corresponding welfare loss in turn leads to utility-based derivative prices in the spirit of Hodges and Neuberger \cite{hodges.neuberger.89} as well as Davis, Panas and Zariphopoulou \cite{davis.al.93}.\footnote{For related asymptotics with small proportional costs, cf.\ \cite{WhWi97,bichuch.13,kallsen.muhlekarbe.13a,bouchard.al.13,PoRo13}.}

We use  dynamic programming and matched asymptotics to prove the results discussed above.
To outline this methodology, let $v^0$ be the frictionless value function of the
initial data $\zeta$.\footnote{As is well known, the frictionless value function depends on
 time $t$, the current values $s$ and $y$ of the risky assets and state
 variables, and the investor's wealth $x$.
 These are collected in $\zeta=(t,s,y,x)$.}
Also let  $v^\lambda$  be its
counterpart  for small linear price impact
 $\Lambda_t=\lambda \Lambda(\cdot)$.\footnote{Here,
 $\lambda \sim 0$ is the small parameter for the asymptotic
 expansion, and $\Lambda(\cdot)$ is a given deterministic function
 of time, the current values of asset prices and state variables,
 and the investor's wealth.}
 Due to the friction, $v^\lambda$
depends not only on $\zeta$  but
also on the number $\vartheta$ of
 shares the investor currently holds.
 Then, the main technical objective is to understand
 the limit behavior of 
$$
\bar u^\lambda(\zeta,\vartheta):=
\frac{v^0(\zeta)-v^\lambda(\zeta,\vartheta)}{\lambda^{1/2}}\ge0, \quad \mbox{as } \lambda \downarrow 0.
$$
The viscosity approach developed
by Evans \cite{Evans92} to problems in homogenization
is suitable for this analysis.
Indeed, it provides a technique to derive  the
equation satisfied by the relaxed semilimits
$\bar u^\ast$ and $\bar u_\ast$ of 
 $\bar u^\lambda$ as $\lambda \downarrow 0$.
Then, by a comparison result, one
concludes that these limits are equal to each other.
In particular, this proves the local uniform
convergence of $\bar u^\lambda$.

In this approach, it is crucial that
the limit functions depend only
on the ``original'' variable~$\zeta$.
However, in our context, the relaxed semilimits
$\bar u^\ast$ and $\bar u_\ast$ depend also on
the  $\vartheta$-variable and we need
to identify this dependence separately.
Indeed, we first show that
$\bar u^\ast$ and $\bar u_\ast$,
are sub- and supersolutions, respectively,
of an Eikonal-type equation as
studied in \cite{kruvzkov.75,ishii.87}:
$$
(D_\vartheta\bar u)^2=\nf,
$$
where $\nf$ is a smooth nonnegative function, quadratic in the $\vartheta$-variable.
In general, there is no comparison principle for the above equation.
However, using a transformation technique,
we prove a comparison result for nonnegative solutions.
This implies the existence of a smooth quadratic function $\varpi$
of the difference between the actual position
$\vartheta$ and the frictionless optimal position $\theta^0(\zeta)$ such that
the there is no$\vartheta$-dependence for the relaxed semilimits of
$$
  \bar u^\lambda(\zeta,\vartheta)-\varpi(\zeta,\vartheta)
$$
We then proceed by analyzing these limits
using the viscosity technique outlined above.

Similar asymptotic results have been recently obtained
for utility maximization with proportional transaction costs in \cite{soner.touzi.13},
for several risky assets in \cite{possamai.al.12}, for random endowments
in \cite{bouchard.al.13}, and for models with fixed transaction costs
in \cite{altarovici.al.13}. In these models, the
semilimits can be shown to be independent of the $\vartheta$-variable
due to the gradient constraint in the dynamic programming equation, because a single trade from the actual position
to the frictionless target is negligible at the leading order.
In contrast, such bulk trades are impossible in our framework as they incur infinite price impact.
This necessitates the novel analysis through the Eikonal equation.

The remainder of this article is organized as follows. The model is set up in Section 2. Afterwards, we state the dynamic programming equations without and with frictions, before turning to the corrector equations governing their asymptotic relationship for small price impacts. For better readability, we first derive the corrector equations heuristically in a simple setting, and then state their general versions. The subsequent Section 4 contains our main results, an asymptotic expansion of the value function for small price impacts and a corresponding almost optimal trading policy. These results, their implications, and connections to the literature are discussed in Section~5, and proved in Section 6. Afterwards, in Section \ref{sec: condition for Assumption B}, we provide a set of sufficient conditions for our technical assumptions, which are standard for verification results (cf., e.g., \cite[Theorem~4.1]{touzi13}). Finally, in Section 8, we show how to verify the conditions of Section 7 in a concrete model.

\paragraph{Notation}
Throughout, $\M^{d\x m}$ denotes the space of $d\x m$ matrices, and $\S^{d}$ the subspace of symmetric $d \times d$ matrices.
For $k\ge1$, $x\in\R^k$ and $r>0$, we write $B_r(x)$ for the open ball of radius $r$ centered at $x$;
$\bar B_r(x)$ and $\partial B_r(x)$ denote its closure and boundary, respectively.

For a smooth function $\vp:(t,x_1,\ldots,x_k)\rightarrow\R$, we write $\partial_t\vp, \partial_{x_i}\vp$ for the corresponding partial derivatives.
The second-order derivatives are denoted by $\partial_{x_ix_j}\vp$ etc. We write $D\vp$ and $D^2\vp$ for the gradient vector and Hessian matrix of $\vp$ with respect to the spatial components, respectively. For any subset $I\subset\{1,\cdots,k\}$, $D_{(x_i)_{i\in I}}$ and $D^2_{(x_i)_{i\in I}}$ refer to the gradient and Hessian with respect to $(x_i)_{i\in I}$.

$C^ i$ denotes the $i$-times continuously differentiable functions, $C^i_b$ is the subspace with bounded derivatives, and $C^{1,2}$ refers to the functions once resp.\ twice continuously differentiable in the time resp.\ space variables.

Finally, for any locally bounded function $v$, the corresponding lower- and upper-semicontinuous envelopes are denoted by $v_*$, $v^*$.

\section{Model}

\subsection{Unaffected Prices}

    Let $(\Omega,\Fc,\P)$ be a complete probability space supporting a $q$-dimensional Brownian motion $W$.
    Fix a finite time horizon $T>0$, and let $\F:=(\Fc_t)_{t\in[0,T]}$ be the augmented filtration generated by $W$.

    We consider a financial market with $d+1$ assets. The first one is safe, and its price is assumed to be normalized to one. The other $d$ assets are risky,
    with unaffected best quotes $S:=(S^1,\ldots, S^d)$ following
    \begin{equation}\label{eq:dynS}
    dS_r=\mu_S(r,S_r,Y_r)dr+\sigma_S(r,S_r,Y_r)dW_r,
    \qquad
    S_t=s,
    \end{equation}
    for a state variable $Y$ taking values in an open subset $\Yc$ of $\R^m$, with dynamics
    \begin{equation}\label{eq:dynY}
    dY_r=\mu_Y(r,Y_r)dr+\sigma_Y(r,Y_r)dW_r,
    \qquad Y_t=y.
    \end{equation}
    The mappings $(\mu_S,\sigma_S):[0,T]\x(0,\infty)^d\x\Yc\longmapsto \R^d\x\M^{d\x q}$ and $(\mu_Y,\sigma_Y):[0,T]\x\Yc\longmapsto \R^m\x\M^{m\x q}$ are continuous and
    Lipschitz-continuous in $(s,y)$. Moreover, $\sigma_S$ belongs to $C^{1,2}$ and satisfies the following local ellipticity condition:
    for any compact subset $B\subset[0,T]\x(0,\infty)^d\x\Yc$,
    there is a constant $\gamma_B>0$ such that:
    \beq\label{eq: ellipticity}
        \abs{{\rm \mathbf{x}}^\top \sigma_S}^2={\rm \mathbf{x}}^\top \sigma_S\sigma_S^\top {\rm \mathbf{x}}\ge\gamma_B\abs{{\rm \mathbf{x}}}^2,
        \quad\pourtout {\rm \mathbf{x}}\in\R^d \mbox{ on } B.
    \eeq
    As a result, for any initial data $(t,s,y)\in[0,T]\x(0,\infty)^d\x\Yc$, there is a unique strong solution of the SDEs (\ref{eq:dynS}-\ref{eq:dynY}), that we denote by $(S^{t,s,y},Y^{t,y})$.

    \begin{remark}
    The condition $\sigma_S \in C^{1,2}$ allows to produce a smooth solution of the First Corrector Equation \eqref{eq: 1st corrector 1D} in Lemma \ref{lem: explicit resolution 1st corrector}.
        This assumption could be weakened using a mollification argument as in \cite{possamai.al.12}.
\end{remark}

\subsection{Linear Price Impact}

The unaffected best quotes $S$ from \eqref{eq:dynS} represent the idealized prices at which minimal amounts can be traded slowly without adversely affecting market prices. In contrast, if $\Delta\theta$ shares are traded over a time interval $\Delta t$, then this order is filled at an average price per share of
$$S_t+\Lambda_t \frac{\Delta\theta}{\Delta t}$$
instead of $S_t$. This price impact is purely ``transient'', in that prices immediately return to their unaffected value after each trade is filled.\footnote{For models also taking into account persistent price impact, cf., e.g., \cite{bertsimas.lo.98,almgren.chriss.01,huberman.stanzl.05,gatheral.10,obizhaeva.wang.13,alfonsi.al.10,roch.soner.13,garleanu.pedersen.13b} and the references therein.}  Moreover, impact is linear in the trading rate $\Delta\theta/\Delta t$.
This is described by the process $\Lambda_t=\lambda\Lambda(t,S_t,Y_t,X_t)$, where $\lambda>0$ is a small parameter and $\Lambda(t,S_t,Y_t,X_t)$ is a $C^{1,2}$-function of time $t$, current prices $S_t$, the state variable $Y_t$, and the investor's current (paper) wealth $X_t$, taking values in the symmetric, positive definite $d\times d$ matrices.\footnote{As pointed out by Garleanu and Pedersen \cite{garleanu.pedersen.13a}, symmetry of $\Lambda$ can be assumed without loss of generality because otherwise the symmetrized version $(\Lambda+\Lambda^\top)/2$ leads to the same trading costs. Positive definiteness means that each transaction has a positive cost. The wealth dependence of the price impact parameter permits the incorporation of feedback effects of the investor's actions on market liquidity. For example, price impact inversely proportional to the investor's current wealth corresponds to the representative investor model of Guasoni and Weber~\cite{guasoni.weber.12,guasoni.weber.13}, where impact is constant relative to the total market capitalization.}  For $\lambda=0$, the usual frictionless model obtains, where arbitrary quantities $\Delta \theta$ can be traded over any time interval $\Delta t$ at the same price $S_t$, for a total execution price of $\Delta\theta S_t$. With a nontrivial $\lambda>0$, trading prices become less favorable in that each order $\Delta \theta$ incurs an additional cost which is quadratic\footnote{Quadratic trading costs can also be motivated by a block-shaped limit order book \cite{obizhaeva.wang.13} or a microstructure model based on the inventory risk accumulated by market makers \cite{garleanu.pedersen.13b}. The empirical literature consistently finds convex trading costs (e.g., \cite{engle.al.08,lillo.al.03}). Some studies actually report quadratic costs \cite{breen.al.02,kyle.obizhaeva.11}, whereas others point towards sublinear price impact with trading costs between linear and quadratic (e.g., \cite{almgren.al.05,toth.al.11}).}  in quantities traded, and inversely proportional to the trade's execution time:
$$\frac{\Delta\theta^\top}{\Delta t} \Lambda_t \frac{\Delta\theta}{\Delta t}\Delta t.$$
These considerations motivate the following continuous-time model.\footnote{Convergence of the respective optimizers is proved in a related model by Garleanu and Pedersen \cite{garleanu.pedersen.13b}.} For any absolutely continuous trading strategy
\begin{equation}\label{eq:strat}
d\theta_r=\dot{\theta}_r dr\;,
\quad
\theta_t=\vartheta,
\end{equation}
 the corresponding (paper) wealth has dynamics
\begin{equation}\label{eq:wealth}
dX_r=\theta_r dS_r-\lambda \dot{\theta}^\top_r \Lambda(r,S_r,Y_r,X_r) \dot{\theta}_r dr\;,
\qquad X_t=x.
\end{equation}
To wit, the usual frictionless dynamics are adjusted for trading costs quadratic in the trading rate $\dot{\theta}$. For notational simplicity, we write
    $$\zeta:=(t,s,y,x)\in\Df,$$where
$$\Df:=\Dfi\cup\Dfb$$
with
    $$
        \Dfi:=[0,T)\x(0,\infty)^d\x\R^m\x \mathbb{R}
        \quad\mbox{and}\quad
        \Dfb:=\{T\}\x(0,\infty)^d\x\R^m\x \mathbb{R}.
    $$
    
With this notation, the set of controls $\Theta^\lambda_0$ consists of the $\F$-progressively measurable trading rates $\dot\theta$ for which the system (\ref{eq:strat}-\ref{eq:wealth}) admits a unique strong solution $(\theta^{t,\vartheta},X^{\zeta,\vartheta,\dot\theta,\lambda})$ for all initial data $(\zeta,\vartheta)\in \Df\x\R^d$.

\begin{remark}
To ease notation and because the time-derivative plays a special role, for any smooth function $\vp:\Df\rightarrow\R$, (resp. $\vp:\Df\x\R^d\rightarrow\R$) we write $D\vp$ (resp. $D_\zeta\vp$) for the gradient of $\vp$ with respect to its \emph{spatial} components $(s,y,x)$. Derivatives with respect to time $t$ are denoted by $\partial_t\vp$ throughout.
\end{remark}

\subsection{Preferences and Liquidation}\label{sec:pref}

In the above market with linear price impact, an investor trades to maximize expected utility from terminal wealth at some finite planning horizon $T>0$. Her utility function $U: \mathbb{R} \to \mathbb{R} \cup \{-\infty\}$ is nondecreasing, as well as smooth and strictly concave on the interior of its effective domain.

As the investment horizon is finite, liquidation at the terminal time $T$ has to be taken into account. For small proportional or fixed trading costs, a single bulk trade is negligible at the leading order, so that this issue disappears asymptotically. With price impact, however, liquidation becomes a nontrivial (and potentially costly) issue. Because we focus here on the dynamic trading before $T$, we separate the liquidation problem as follows. We suppose that the model parameters are simply frozen at time $T$ and the investor's terminal position $\theta_T$ is liquidated quickly towards the frictionless target $\theta^0_T=\theta^0(T,S_T,Y_T,X^{\theta}_T)$ using the deterministic mean-variance optimal strategy from Sch\"oneborn~\cite{schoeneborn.11}, with constant risk-tolerance $R_T=-U'(X^{\theta}_T)/U''(X^{\theta}_T)$. This leads to risk-adjusted liquidation costs \cite[Equation (11)]{schoeneborn.11} of $\lambda^{1/2}\mathfrak{P}(T,S_T,Y_T,X^{\theta}_T,\theta_T)$, where \cite[Theorem 4.1]{schoeneborn.11}:
$$\mathfrak{P}(\zeta,\vartheta):=(\vartheta-\theta^0(\zeta))^\top\frac{\Lambda^{1/2} (\Lambda^{-1/2}\sigma_S\sigma_S^\top \Lambda^{-1/2})^{1/2} \Lambda^{1/2}}{(2R)^{1/2}}(\zeta)(\vartheta-\theta^0(\zeta)).$$
As these liquidation costs are small for small price impacts ($\Lambda \sim 0$), we in turn define the investor's frictional value function as suggested by Taylor's theorem:
    \beq\label{eq: definition friction value function}
        v^\lambda(\zeta,\vartheta) := \sup_{\dot{\theta}\in\dot{\Theta}_{\zeta,\vartheta}^\lambda}\Esp{U\left(X^{\zeta,\vartheta,\dot\theta,\lambda}_T\right)-U'(X^{\zeta,\vartheta,\dot\theta,\veps}_T)\lambda^{1/2}\mathfrak{P}(T,S^\zeta_T,Y^\zeta_T,X^{\zeta,\vartheta,\dot\theta,\veps}_T,\theta_T^{\zeta,\vartheta})},
    \eeq
    for initial data $(\zeta,\vartheta)\in \Df\x\R$. Here, $\dot\theta$ runs through the set $\dot{\Theta}_{\zeta,\vartheta}^\lambda$ of \emph{admissible} controls. These have to satisfy
     \beq\label{eq: admiss for friction control}
        U(X^{\zeta,\vartheta,\dot\theta,\veps}_T)-U'(X^{\zeta,\vartheta,\dot\theta,\veps}_T)\lambda^{1/2}\mathfrak{P}(T,S^\zeta_T,Y^\zeta_T,X^{\zeta,\vartheta,\dot\theta,\veps}_T,\theta_T^{\zeta,\vartheta})    \in L^1.
    \eeq
    Moreover, one needs to be able to approximate the corresponding wealth processes using simple strategies as in Biagini and \v{C}ern\'y~\cite{biagini.cerny.11}. The first condition is evidently needed to make the terminal utility well defined. The second assumption is an economically meaningful class of strategies that is small enough to exclude doubling strategies,\footnote{With superlinear frictions, doubling strategies need not be ruled out a priori to make the frictional problem well posed \cite{guasoni.rasonyi.14}. However, even if doubling strategies are not scalable at will, their availability may still cause the value function to become discontinuous at the terminal time $T$, ruling out classical verification theorems as in Section~\ref{sec: condition for Assumption B}. Therefore, we do not allow doubling strategies here.} but large enough to contain the optimizer under weak assumptions; see \cite{biagini.cerny.11} for more details. For utilities which are only finite on the positive half-line, the approximation property is replaced by requiring the wealth process to be positive on $[0,T]$.

\begin{remark}
The liquidation penalty $\mathfrak{P}$ disappears in the following two important special cases:
\begin{enumerate}
\item For infinite-horizon problems as in \cite{garleanu.pedersen.13a,garleanu.pedersen.13b,guasoni.weber.12,guasoni.weber.13}, liquidation is not an issue. Indeed, as the horizon grows, the cost of the terminal liquidation program remains the same, whereas the accumulated benefits from trading grow indefinitely.
\item Suppose that the initial allocation is close to the frictionless target. Then, for strategies that always trade quickly towards the latter, the deviation always remains small in expectation. Hence, the liquidation penalty is of higher order in this case, and can be neglected asymptotically.
\end{enumerate}
For finite-horizon problems and arbitrary initial endowments, however, liquidation has to be taken into account explicitly, see \cite{almgren.li.11}.
\end{remark}

    \begin{remark}\label{rem:2.4}
    Instead of requiring liquidation to the frictionless optimizer in \eqref{eq: definition friction value function}, one could also impose liquidation to a full cash position, or no liquidation penalty at all. Both of these alternatives are economically meaningful, but complicate the problem substantially. The reason is that unlike for proportional or fixed costs, one cannot set up or liquidate a given portfolio with a single block trade and trading costs negligible at the leading asymptotic order. Therefore, with no liquidation penalty, investors with a short horizon will only trade very little if their initial position is far from the frictionless target to save trading costs. In contrast, with a longer horizon, they will trade much more aggressively to reap the gains from an optimal position in the long run. Requiring full liquidation leads to similar inhomogeneities. Indeed, as the horizon nears, the investor's focus then gradually shifts from rebalancing to maintain an optimal risk-return tradeoff to a liquidation program. In contrast, liquidation towards the frictionless target leads to a ``stationary'' version of the (asymptotic) problem, where the effects of setting up and liquidating the portfolio are disregarded, to be dealt with as separate optimal execution problems.
    \end{remark}

\section{Dynamic Programming and Corrector Equations}\label{sec: mathematical results}

In this section, we state the dynamic programming equations solved by the frictionless and frictional value functions, respectively. For small price impacts, their difference is described by the solution of the so-called ``corrector equations''. To provide some intuition, we first derive these heuristically for a single risky asset and state variable. Afterwards, we state the general multidimensional versions.

\subsection{The Frictionless Case}\label{sec: fl case}

    Without price impact, the diffusions $(S^{t,s,y},Y^{t,y})$ are still defined as the strong solutions of the SDEs
    (\ref{eq:dynS}-\ref{eq:dynY})
    but, without trading costs, the wealth dynamics \eqref{eq:wealth} reduce to
    $$dX^{\zeta,\theta}_r=\theta_r dS_r, \quad X^{\zeta,\theta}_t=x.$$
    Here, the -- now no longer necessarily absolutely continuous -- control  $\theta$ denotes the numbers of risky shares held in the portfolio.
    The control set consists of the $\F$-progressively measurable processes taking values in $\R^d$
    such that the above SDE admits a unique strong solution $X^{\zeta,\theta}$.
    As above, we restrict ourselves to the subset $\Theta^0_\zeta$ of \emph{admissible} controls for which $U(X^{\zeta,\theta}_T)\in L^1$,
    and for which the corresponding wealth processes can be approximated by simple strategies as in~\cite{biagini.cerny.11}.
    The frictionless value function is then defined as follows:
    \beq\label{eq: definition friction less value function with L2 process}
        v^0(\zeta) := \sup_{\theta\in\Theta^0_\zeta}\Esp{U\left( X^{\zeta,\theta}_T \right)}.
    \eeq

    Standard arguments (compare, e.g., \cite{fleming.soner.06}) show that the frictionless value function $v^0$ solves the Dynamic Programming Equation (henceforth DPE) for the problem at hand:

    \begin{proposition}\label{theo: friction less value function}
      Assume that $v^0$ is locally bounded.
      Then it is a (discontinuous) viscosity solution of
    \beq\label{eq: friction less solution}
        \left\{\bal
            \inf_{\vartheta\in\R^d}\left\{ -\Lc^\vartheta v^0 \right\}=0,
                &\quad\mbox{on } \Dfi,\\
            v^0(T,\zeta)=U(x),
                &\quad\mbox{on } \Dfb,
        \eal\right.
    \eeq
      where, for $\psi\in C^{1,2}$ and $(\zeta,\vartheta)\in \Df\x\R^d$:
        \begin{equation*}\label{eq: dynkin def}\bal
            \Lc^\vartheta\psi(\zeta,\vartheta) := \;&
                \left\{\partial_t \psi + \mu_\vartheta\cdot D_\zeta\psi + \frac12\Tr{\sigma_\vartheta\sigma^\top_\vartheta D_\zeta^2\psi}\right\}(\zeta,\vartheta),
        \eal\end{equation*}
        with
    $$
        \mu_\vartheta(\zeta):=\begin{pmatrix} \mu_S\\ \mu_Y\\\vartheta\cdot \mu_S \end{pmatrix}(\zeta)
        \et
        \sigma_\vartheta(\zeta):=\begin{pmatrix} \sigma_S\\ \sigma_Y\\\vartheta^\top \sigma_S \end{pmatrix}(\zeta).
    $$
    \end{proposition}

    \begin{remark}\label{rem: simplification HJB frictionless}
      Suppose that $v^0$ is smooth with $\partial_{xx}v^0<0$. Then, as $\sigma_S$ satisfies the ellipticity condition
      \reff{eq: ellipticity},
    it follows that $v^0$ is a classical solution of
     \beq\label{eq: PDE friction less simplified with pf}\bal
                \Lc^{\theta^0} v^0(\zeta)=0,
     \eal\eeq
     for all $\zeta\in \Dfi$ or, equivalently,
     \beq\label{eq: friction less HJB simplified}
                \left\{\partial_t v^0 + \mu_0 Dv^0 + \frac12\Tr{\bar\sigma_0\bar\sigma^\top_0 D_{(s,y)}^2v^0}\right\}(\zeta)
                =\left(\frac12 (\theta^0)^\top \sigma_S \sigma_S^\top \theta^0 \partial_{xx} v^0 \right)(\zeta),
     \eeq
     where the optimal investment strategy $\theta^0(\zeta)$ satisfies
     \beq\label{eq: optimum control frictionless}
        -(\partial_{xx}v^0 \sigma_S\sigma_S^\top\theta^0)(\zeta) := \mu_S\partial_xv^0
                        +\sigma_S\bar\sigma^\top_0 D_{(s,y)}(\partial_xv^0)(\zeta),
     \eeq
     with
     $$
        \bar\sigma_0
        :=\begin{pmatrix}
          \sigma_S
          \\ \sigma_Y
        \end{pmatrix}.
     $$
     Indeed, given sufficient regularity of the coefficients of the SDEs, standard verification arguments (compare, e.g., \cite{touzi13}) show that
     the Markovian feedback policy
     $$
        \theta^0_u:=\theta^0\left(u,S^{t,s,y}_u,\hat X^{t,s,y,x,\theta^0}_u,Y^{t,y}_u \right), \quad u \in [t,T],
     $$
     is optimal for \reff{eq: definition friction less value function with L2 process} in this case. Note that -- with an abuse of notation -- we use the same symbol to denote both the feedback description of a strategy \emph{and} its evolution as a stochastic process. 
    \end{remark}

\subsection{The Dynamic Programming Equation with Price Impact}

Given that the frictionless value function $v^0$ is locally bounded, its frictional counterpart $v^\lambda$ is evidently locally bounded from above because any absolutely continuous control in $\dot\Theta^\lambda_{\zeta,\vartheta}$ can be reproduced by a control in $\Theta^0_\zeta$,
      the utility function $U$ is nondecreasing, and the penalty function $\Pf$ is nonnegative. We assume in addition that $v^\lambda$ is also locally bounded from below, i.e., there exists at least one strategy that closes out any initial position with finite utility.\footnote{For any initial wealth, this is evidently possible with a single bulk trade for sufficiently small proportional or fixed costs. With linear price impact, only absolutely continuous trading strategies can be implemented. Therefore, one has to restrict to long-only portfolios for utilities defined on the positive half-line, and impose sufficient integrability on the asset dynamics even for utilities defined on the whole real line, see Section~\ref{sec:final} for more details.} 

Next, we turn to the corresponding DPE with linear price impact. Without state constraints, i.e.\ for utilities that are finite on the whole real line, the latter can be derived from the weak dynamic programming principle of Bouchard and Touzi \cite{bouchard.touzi.11}. It is expected that this remains true if wealth is required to remain positive for utilities finite only on $\mathbb{R}_+$, see \cite{bouchard.nutz.12}. Making this rigorous in the presence of frictions is more delicate, though, see \cite{altarovici.al.13,soner.vukelja.14} for some specific examples. Therefore, we simply state the DPE as an assumption in the general setting considered here:

    \begin{assumption} \label{prop:dpefric}
      The frictional value function $v^\lambda$
       is locally bounded and a (discontinuous) viscosity solution of
        \beq\label{eq: viscosity solution for friction case}
            \left\{\bal
              -\Lc^\vartheta v^\lambda - \Hc^\lambda v^\lambda=0,& \quad\mbox{on } \Dfi\x\R^d,\\
              v^\lambda= U-U'\lambda^{1/2}\Pf, &\quad\mbox{on }\partial_T\Df\x\R^d,
            \eal\right.
        \eeq
      where, for $\psi\in C^{1,2}$ and $(\zeta,\vartheta)\in \Df\x\R^d$:
        \beq\label{eq: non linear operator}
            \Hc^\lambda\psi(\zeta,\vartheta):=
                \sup_{\dot{\vartheta}\in\R^d}\left\{ \dot\vartheta\cdot D_\vartheta\psi - \lambda \dot{\vartheta}^\top \Lambda \dot{\vartheta}\partial_x\psi\right\}(\zeta,\vartheta),
        \eeq
        and the liquidation penalty $\Pf$ is defined as in Section \ref{sec:pref}.
\end{assumption}

        \begin{remark}\label{rem: simplification HJB with friction}
      The PDE \eqref{eq: viscosity solution for friction case} generally has to be understood in terms of the semicontinuous envelopes $\Hc^{\lambda,\ast},\Hc^\lambda_\ast$ of 
      $$
        \Hc^\lambda:(\zeta,q_x,q_\vc)\in\Df\x\R\x\R^d\longmapsto
            \sup_{\dot{\vartheta}\in\R^d}\left\{ \dot\vartheta\cdot q_\vc - \lambda \dot{\vartheta}^\top \Lambda(\zeta) \dot{\vartheta}q_x\right\}.
      $$
      (We use the shorthand notation $\Hc^\lambda\psi(\zeta,\vartheta):=\Hc^\lambda(\zeta,\partial_x\psi(\zeta,\vc),D_\vc\psi(\zeta,\vc))$.)
      
      However, we have $\Hc^{\lambda,\ast}=\Hc^\lambda_\ast=\Hc^\lambda$ on $\Df\x(0,\infty)\x\R^d$
      so that this relaxation is superfluous for smooth test function $\psi$ satisfying $\partial_x\psi>0$ on $\Df\x\R^d$.
      Moreover, in this case, positive-definiteness of $\Lambda$ gives that the first line in \reff{eq: viscosity solution for friction case}
     can be rewritten as
        \beq\label{eq: PDE for v eps simplified}
            -\left(\Lc^\vartheta\psi
                +\frac{(D_\vartheta\psi)^\top \Lambda^{-1}D_\vartheta\psi}{4\lambda\partial_x\psi}\right)(\zeta,\vartheta)=0,
                \quad\pourtout (\zeta,\vartheta)\in \Dfi\x \R^d,
        \eeq
        where we have used the pointwise optimizer in \reff{eq: non linear operator}:
        \begin{equation}\label{eq:strategy2}
            \dot{\vartheta}^\lambda(\zeta,\vartheta):=\frac{\Lambda^{-1}D_\vartheta\psi}{2\lambda\partial_x\psi}(\zeta,\vartheta).
        \end{equation}
        
    \end{remark}

    \subsection{Heuristic Expansion for a Single Risky Asset}\label{sec:heuristics}

    Our goal is to show that, for all $(\zeta,\vartheta)\in \Df\x\R^d$, the frictional value function has the asymptotic expansion
            \beq\label{eq: formal taylor expansion of v eps}
            v^\lambda(\zeta,\vartheta)=v^0(\zeta)-\lambda^{1/2} u(\zeta) -\lambda \varpi\circ\xib_\lambda(\zeta,\vartheta)+o(\lambda^{1/2}).
        \eeq
        Here, we write
        $$
            \varpi\circ\xib_\lambda(\zeta,\vartheta):=\varpi(\zeta,\xib_\lambda(\zeta,\vartheta))
        $$
         for $\varpi:(\zeta,\xi)\in \Df\x\R^d\longmapsto \varpi(\zeta,\xi)$, and the ``fast'' variable
        \beq\label{eq: fast variable}
            \xib_\lambda(\zeta,\vartheta):=\frac{\vartheta-\theta^0(\zeta)}{\lambda^{1/4}}
        \eeq
        measures the deviation of the actual position from the frictionless target \reff{eq: optimum control frictionless}, rescaled to be of order one as $\lambda \to 0$.

        \begin{remark}
        The asymptotic scalings for the value function and the optimal policy are motivated by the corresponding results of Guasoni and Weber \cite{guasoni.weber.12}.
        \end{remark}

        To motivate the corrector equations describing the asymptotics
        (cf.\ Section~\ref{ssec:corrector}),
        let us first informally derive them for a single risky asset ($d=1$) and a single state variable ($m=1$).\footnote{The corresponding calculations for several assets and state variables are analogous, but more tedious.} Both processes are driven by a two-dimensional Brownian motion ($q=2$), with volatilities
            $$
                \sigma_S:=\begin{pmatrix} \sigma_{S,1}&0 \end{pmatrix}
                \et
                \sigma_Y:=\begin{pmatrix} \sigma_{Y,1}&\sigma_{Y,2} \end{pmatrix},
            $$
            so that price and state shocks are correlated for $\sigma_{Y,1} \neq 0$.
            In this simple framework, the price impact matrix $\Lambda$ is simply a positive, smooth, scalar function on $\Df$.
            Suppose that $v^0$ and $v^\lambda$ are classical solutions of \reff{eq: friction less solution} and
            \reff{eq: viscosity solution for friction case}, respectively, satisfying $\partial_x v^0\wedge (-\partial_{xx} v^0)\wedge \partial_x v^\lambda>0$.
            Assume furthermore that the functions $\theta^0, u, \varpi$ and $\xib_\lambda$ belong to $C^{1,2}$, and introduce the local quadratic variation of the frictionless optimizer:
            \beq\label{eq: def delta 1 dim}
                c_{\theta^0}(\zeta):=\frac{d\langle \theta^0\rangle}{dt}(\zeta)= \left( \sigma_S\partial_s\theta^0 + \sigma_{SY}\partial_y\theta^0+\sigma_S\theta^0\partial_x\theta^0 \right)^2(\zeta)+\left( \sigma_Y\partial_y\theta^0 \right)^2(\zeta) \geq 0.
            \eeq
            Notice that $\theta^0$ refers to the evolution of the optimal frictionless strategy as a stochastic process, where the appropriate state variables are plugged into its feedback description. 

    \subsubsection{The Corrector Equations}\label{sec: informal Lc u}

            Inserting the ansatz (\ref{eq: formal taylor expansion of v eps}-\ref{eq: fast variable})
            into the frictional DPE \reff{eq: PDE for v eps simplified} leads to
            \begin{align*}
                0 &=
                    -\Lc^{\theta^0}v^0-\lambda^{1/4}\xib_\lambda
                        \left(
                            \mu_S\partial_xv^0 + \sigma_{S,1}^2\partial_{sx}v^0+\sigma_S\sigma_{Y,1}\partial_{xy}v^0+\sigma_{S,1}^2\theta^0\partial_{xx}v^0
                        \right)\\
                & \qquad
                    -\lambda^{1/2}\left(
                        -\Lc^{\theta^0} u
                        +\frac12\sigma^2_{S,1}\partial_{xx}v^0\xib_\lambda^2
                        -\frac12 c_{\theta^0}\partial_{\xi\xi}\varpi
                        +\frac{(\partial_\xi\varpi)^2}{4 \Lambda \partial_xv^0}
                        \right)+o(\lambda^{1/2}).
            \end{align*}
            Here, the first line vanishes by the frictionless DPE \eqref{eq: PDE friction less simplified with pf} and the first-order condition \eqref{eq: optimum control frictionless} for the frictionless optimizer.
                Observe that the map $u$ in \reff{eq: formal taylor expansion of v eps}
                is independent of $\vc$, hence $\Lc^{\theta^0} u$ is a function of $\zeta$ only as well.
                As a consequence, the remaining terms at the order $\lambda^{1/2}$ in the previous equation should not depend on $\vc$ either. Therefore, we first look for a function $a:\Df\rightarrow\R$ such that the pair $(\varpi,a)$ is solution,
            for fixed $(t,s,x,y)\in \Dfi$, of the \emph{first corrector equation}
            \beq\label{eq: 1st corrector 1D}
                        \frac12\sigma^2_{S,1}\xi^2\partial_{xx}v^0
                        -\frac12c_{\theta^0} \partial_{\xi\xi}\varpi
                        +\frac{\Lambda^{-1}(\partial_\xi\varpi)^2}{4\partial_xv^0}
                        +a=0,
            \eeq
            and then identify $u$ as the solution on $\Dfi$ of the \emph{second corrector equation}
            \begin{equation}\label{eq:corrector2}
                -\Lc^{\theta^0} u - a=0.
            \end{equation}
            Now, insert the ansatz \reff{eq: formal taylor expansion of v eps} into the terminal condition \reff{eq: viscosity solution for friction case}
            for the frictional value function $v^\lambda$
            and use the terminal condition \reff{eq: friction less solution} for its frictionless counterpart $v^0$. This shows that the corresponding terminal condition for $u$ is given by
            \beq\label{eq: formal terminal condition for u}
                \lambda^{1/2}u+\lambda\varpi\circ\xib_\lambda=U'\lambda^{1/2}\Pf,
                \quad \mbox{on $\Dfb$}.
            \eeq

            Let $R:=-\partial_xv^0/\partial_{xx}v^0$ denote the risk tolerance of the frictionless value function. As $R>0$ because we assumed $-\partial_{xx} v^0 \wedge \partial_x v^0 >0$, the First Corrector Equation \reff{eq: 1st corrector 1D} is readily rewritten as
            $$
                        \frac{\sigma^2_{S,1}}{2\Lambda R}\xi^2
                        +\frac{c_{\theta^0}}{2\Lambda\partial_xv^0}\partial_{\xi\xi}\varpi
                        -\left(\frac{\partial_\xi\varpi}{2 \Lambda \partial_xv^0}\right)^2
                        -\frac a{\Lambda\partial_xv^0}=0.
            $$
            Evidently, there should be no penalty for deviating when the actual position coincides with the frictionless target. Hence, we impose
            the additional constraint $\varpi(\cdot,0)=0$, obtaining the explicit solution
            $(\varpi,a)$ with
            $$\varpi(\zeta,\xi)=k_2(\zeta)\xi^2,$$
            as well as
            \beq\label{eq: k_2 and a in 1D}
                    k_2 = \pm (\Lambda\partial_xv^0)\sqrt{\sigma_{S,1}^2/(2\Lambda R)},\qquad a = c_{\theta^0} k_2.
            \eeq
             Via \reff{eq:strategy2}, \eqref{eq: formal taylor expansion of v eps}, and \reff{eq: fast variable}, this identifies the optimal trading rate for small price impact ($\lambda \sim 0$) as
            $$
                \dot\theta^\lambda(\zeta,\vartheta) \sim -\frac{\lambda^{3/4}\partial_\xi\varpi(\zeta,\xib(\zeta,\vartheta))}{2\lambda \Lambda(\zeta)\partial_xv^0(\zeta)}
                    =-\left(\pm \sqrt{\frac{\sigma_{S,1}(t,s,y)^2}{2\lambda\Lambda(\zeta)R(\zeta)}}(\vartheta-\theta^0(\zeta))\right).
            $$
            As one should evidently always trade towards the frictionless position $\theta^0$ rather than away from it,
            the positive sign for $k_2$ is the correct one in \reff{eq: k_2 and a in 1D}.
            Hence, asymptotically for small $\lambda$, the optimal policy prescribes trading towards the target portfolio at rate $\sqrt{\sigma^2_{S,1}/(2\lambda\Lambda R)}$, in line with \eqref{eq:rate}.

            Observe furthermore that the explicit form of $k_2$ gives
            $\lambda\varpi\circ\xib_\lambda=\lambda^{1/2}\varpi\circ\xib_1=U'\lambda^{1/2}\Pf$ on $\Dfb$,
            so that the terminal condition for $u$ in \reff{eq: formal terminal condition for u} reads as
            \begin{equation}\label{eq:terminalu}
                u=0,\quad \mbox{on $\Dfb$}.
            \end{equation}

    \subsection{Corrector Equations in the General Multidimensional Case}\label{ssec:corrector}

    Let us now state the general multidimensional counterparts of the Corrector Equations (\ref{eq: 1st corrector 1D}-\ref{eq:corrector2}, \ref{eq:terminalu}). To this end, we first introduce the $d$-dimensional counterpart of the local quadratic variation $c_{\theta^0}$ defined in \reff{eq: def delta 1 dim}:
        \begin{equation}\label{eq: def delta multi d}
            c_{\theta^0}(\zeta):=\frac{d\langle\theta^0\rangle_t}{dt}(\zeta)=(D_\zeta\theta^0)^\top\sigma_{\theta^0}\sigma_{\theta^0}^\top D_\zeta\theta^0.
        \end{equation}
        With this notation, the corrector equations in the general multivariate case read as follows:

        \begin{definition}{\bf (Corrector Equations)}
          For a given point $\zeta\in\Df$,
          the \emph{first corrector equation} for the unknown pair $(a(\zeta),\varpi(\zeta,\cdot))\in\R\x C^{2}(\R)$ is
          \beq\label{eq: first corrector equation}
            \left\{
                    \frac12\abs{\xi^\top \sigma_S}^2\partial_{xx}v^0
                    -\frac12\Tr{c_{\theta^0} D^2_{\xi\xi}\varpi(\cdot,\xi)}
                    +\frac{(D_\xi \varpi)^\top \Lambda^{-1}D_\xi \varpi}{4\partial_xv^0}(\cdot,\xi)
                    +a
            \right\}(\zeta)=0,
          \eeq
          together with the normalization $\varpi(\zeta,0)=0$.

          The \emph{second corrector equation} uses the constant term $a(\zeta)$ from the first corrector, and is a simple linear equation for the function $u:\Df\rightarrow\R$:
          \beq\label{eq: 2nd corrector equation}
            \left\{\bal
                -\Lc^{\theta^0} u= a,&\quad\mbox{on } \Dfi,\\
                u=0,&\quad\mbox{on }\Dfb.
            \eal\right.
          \eeq
          We say that the pair $(u,\varpi)$ is \emph{a solution of the corrector equations}.
        \end{definition}

        For a single risky asset ($d=1$) and a single state variable ($m=1$), one readily verifies that these definitions coincide with the equations derived heuristically in Section \ref{sec:heuristics} above.

\section{Main Results}

Our main results are an asymptotic expansion of the value function $v^\lambda$ for small price impact $\Lambda_t=\lambda \Lambda(\cdot) \sim 0$, and an ``almost optimal'' trading policy that achieves the optimal performance at the leading order. To formulate these results, set
\begin{equation}\label{eq:baru}
\bar u^\lambda(\zeta,\vartheta):= \frac{v^0(\zeta)-v^\lambda(\zeta,\vartheta)}{\lambda^{1/2}} \ge 0.
\end{equation}
 Then, the leading-order behavior of this difference can be analyzed under our Standing Assumption~\ref{prop:dpefric} that the frictional value function is a viscosity solution of the corresponding DPE and the following abstract conditions:\footnote{Convenient sufficient conditions for their validity are provided in Section \ref{sec: condition for Assumption B}, and verified in a specific setting in Section \ref{sec:final}. As in related results for proportional and fixed costs \cite{soner.touzi.13,altarovici.al.13}, these ``verification theorems'' are based on the availability of classical smooth solutions.}

 \begin{assumptionA}
          \benumlab{A}
            \item\label{ass: v smooth}
                (Regularity of the frictionless problem) The frictionless value function $v^0$ and optimal investment strategy $\theta^0$ belong to $C^{1,2}$. Moreover, $\partial_xv^0\wedge(-\partial_{xx}v^0)>0$.
            \item\label{ass: u locally bounded from above}
                (Locally uniform bound) For any $(\zeta_o,\vartheta_o)\in \Df\x\R^d$, there exist $r_o, \lambda_o>0$ such that
                $$
                    \sup
                        \left\{
                            \bar u^\lambda(\zeta,\vartheta): (\zeta,\vartheta)\in B_{r_o}(\zeta_o,\vartheta_o)\cap(\Df\x\R^d) \mbox{ and } \lambda\in(0,\lambda_o]
                        \right\}<\infty.
                $$
            \item\label{ass: comparison for u}
               (Comparison)
               A viscosity solution $u$ of the Second Corrector Equation \eqref{eq: 2nd corrector equation} exists.
               Moreover, there is a class of functions $\Cc$ which contains
               $u, \bar u_\ast(\cdot,\theta^0(\cdot))$ and $\bar u^\ast(\cdot,\theta^0(\cdot))$
               such that $u_1\ge u_2$ for all $u_1, u_2\in\Cc$ with $u_1$ (resp. $u_2$) being a lower-semicontinuous (resp. upper-semicontinuous)
                viscosity supersolution (resp. subsolution) of the Second Corrector Equation \reff{eq: 2nd corrector equation}.\footnote{In particular, $u$ is the unique viscosity solution of  \eqref{eq: 2nd corrector equation} in the class $\Cc$.} Here, $ \bar u^\ast, \bar u_\ast$ denote the following relaxed semilimits:
        \beq\label{eq: def bar u ast}
            \bar u^\ast(\zeta,\vartheta) := \Limsup_{\lambda\to 0,(\zeta',\vartheta')\rightarrow(\zeta,\vartheta)} \bar u^\lambda(\zeta',\vartheta')\;,
            \qquad
            \bar u_\ast(\zeta,\vartheta) := \Liminf_{\lambda\to 0,(\zeta',\vartheta')\rightarrow(\zeta,\vartheta)} \bar u^\lambda(\zeta',\vartheta')\;,
        \eeq
        for all $(\zeta,\vc)\in\Df\x\R^d$, which are well-defined upper- resp.\ lower-semicontinuous functions under Assumption~\ref{ass: u locally bounded from above}. 
          \enumlab
        \end{assumptionA}

        Assumptions \reff{ass: v smooth} and \reff{ass: comparison for u} are technical and can be guaranteed by imposing sufficient regularity conditions on the coefficient functions of the model. The crucial assumption is \reff{ass: u locally bounded from above}, which postulates that the leading-order correction of the value function due to small price impact $\lambda\Lambda$ is indeed of order $O(\lambda^{1/2})$ as $\lambda \to 0$. This condition needs to be verified with more specific arguments. See Sections \ref{sec: condition for Assumption B} and \ref{sec:final} for a verification theorem that achieves this for sufficiently regular classical solutions of the dynamic programming equations.

\begin{lemma}\label{lem: explicit resolution 1st corrector}
Suppose Assumption \reff{ass: v smooth} is satisfied. Then, the First Corrector Equation \reff{eq: first corrector equation} is solved by the locally bounded function
            \begin{equation}\label{eq:a}
            a(\zeta)= \Tr{c_{\theta^0} k_2}(\zeta)
            \end{equation}
            and the map
            $$\varpi:\xi \longmapsto \xi^\top k_2(\zeta)\xi,$$
            where $c_{\theta^0}=d\langle \theta^0\rangle/dt$ is the local quadratic variation of the frictionless target strategy $\theta^0$, and the positive semidefinite function $k_2 \in C^{1,2}(\Df;\S^d)$ is defined as
            $$k_2(\zeta)=\frac{\partial_xv^0}{\sqrt{-2\partial_x v^0/\partial_{xx} v^0}} \left[\Lambda^{1/2}(\Lambda^{-1/2}\sigma_S\sigma_S^\top \Lambda^{-1/2})^{1/2}\Lambda^{1/2}\right](\zeta).$$
If, in addition, Assumption \reff{ass: u locally bounded from above} holds, then,
evaluated along the frictionless optimal strategy $\theta^0$, the semilimits $\bar u^\ast(\cdot,\theta^0(\cdot)), \bar u_\ast(\cdot,\theta^0(\cdot))$ are viscosity sub- and supersolutions, respectively, of the Second Corrector Equation \eqref{eq: 2nd corrector equation}.
\end{lemma}

\proof
Under \eqref{ass: v smooth}, the first part of the assertion is readily verified by direct computation. For the second part, first notice that the relaxed semilimits are finite by Assumption \reff{ass: u locally bounded from above} and are upper- resp.\ lower-semicontinuous by definition. Using Assumptions \reff{ass: v smooth} and \reff{ass: u locally bounded from above},
            we show in Propositions~\ref{prop: sub sol for u}, \ref{prop: super sol for u}, and \ref{prop: terminal condition for u}
            that
            $\zeta\in\Df\longmapsto \bar{u}^\ast(\zeta,\theta^0(\zeta))$
            and $\zeta\in\Df\longmapsto \bar{u}_\ast(\zeta,\theta^0(\zeta))$ are viscosity sub- resp.\ supersolutions of the Second Corrector Equation \reff{eq: 2nd corrector equation} with $a$ defined as in \eqref{eq:a}.
            \ep

        \begin{remark}
            For later use, observe that the function $\varpi$ satisfies, for all $\xi\in\R^d$:
          \beq\label{eq: estimate for w}
            \frac{(\abs\varpi+|D_{(t,\zeta)}\varpi|)(\cdot,\xi)}{1+\abs\xi^2}
            +\frac{(|D_\xi\varpi|+|D_{(t,\zeta)}(D_{\xi}\varpi)|)(\cdot,\xi)}{1+\abs\xi}
            +|D^2_{\xi\xi}\varpi|(\cdot,\xi)
            \le \varrho,
            \quad\mbox{on } \Df,
          \eeq
            for some continuous function $\varrho:\Df\rightarrow\R$.
        \end{remark}

We now state our main result, which determines the leading-order coefficient of the value function, under the Assumption~\ref{ass: u locally bounded from above} that the first nontrivial term in its expansion is of order $O(\lambda^{1/2})$:

\begin{theorem}{\bf (Expansion of the Value Function)}\label{theo: main result}
Suppose Assumptions \ref{prop:dpefric} and {\rm A} are satisfied. 
Then, for any initial data $(\zeta,\vartheta)\in\Df\x\R^d$:
$$\bar u^\lambda(\zeta,\vartheta) \longrightarrow  u(\zeta)+\varpi\left(\zeta,\vartheta-\theta^0(\zeta)\right),$$
locally uniformly as $\lambda \to 0$. That is, the frictional value function $v^\lambda(\zeta,\vartheta)$ has the expansion
$$v^\lambda(\zeta,\vartheta)=v^0(\zeta)-\lambda^{1/2}(u(\zeta)+\varpi\left(\zeta,\vartheta-\theta^0(\zeta)\right)+o(\lambda^{1/2}).$$
\end{theorem}

    The lengthy proof of this result is postponed to Section~\ref{sec: PDE derivation}.

 \begin{remark}
In view of the explicit formula in Lemma~\ref{lem: explicit resolution 1st corrector}, the penalty for deviations of the initial portfolio $\vartheta$ from the frictionless target $\theta^0$ is given by
$$\lambda^{1/2}\varpi\left(\zeta,\vartheta-\theta^0(\zeta)\right)=\lambda^{1/2}\frac{\partial_x v^0(\zeta)}{\sqrt{2 R(\zeta)}}(\vartheta-\theta^0(\zeta))^\top\left(\left(\Lambda^{1/2}(\Lambda^{-1/2}\sigma_S \sigma_S^{\top} \Lambda^{-1/2}\right)^{1/2}(\zeta)\right)(\vartheta-\theta^0(\zeta)).
$$
Hence, it is negligible at the leading order $O(\lambda^{1/2})$ for initial positions $\vartheta$ sufficiently close to the frictionless optimizer $\theta^0(\zeta)$.
\end{remark}

\begin{remark}
By Lemma~\ref{lem: explicit resolution 1st corrector}, the term $a$ from the First Corrector Equation \reff{eq: first corrector equation} is nonnegative. Hence, if the regularity conditions of \cite[Remark 5.7.8]{karatzas.shreve.91} or, more generally \cite[Chapter I]{friedman.64} are satisfied, a smooth classical solution of the Second Corrector Equation \reff{eq: 2nd corrector equation} exists. It admits the Feynman-Kac representation
         \begin{equation}\label{eq: definition of u}
         \begin{split}
            u(\zeta)&=\Esp{\int_t^T a\left(r, S^{t,s,y}_r,Y^{t,y}_r,X^{\zeta,\theta^0}_r\right)dr},\\
            &=\E
        \left[
            \int_t^T \left(\frac{\partial_x v^0}{\sqrt{2 R}}
                \Tr{
                  \frac{d\langle\theta^0\rangle_r}{dr} \Lambda^{1/2}(\Lambda^{1/2}\sigma_S \sigma_S^\top \Lambda^{1/2})^{1/2}\Lambda^{1/2}
                }\right)(r,S_r^{\zeta},Y_r^{\zeta},X^{\zeta,\theta^0}_r)
            dr
        \right].
          \end{split}
          \end{equation}
          Here, $X^{\zeta,\theta^0}$ denotes the optimal frictionless wealth process and $R(\zeta):=-\partial_xv^0(\zeta)/\partial_{xx}v^0(\zeta)$ represents the risk tolerance of the frictionless indirect utility function; the second equality in \eqref{eq: definition of u} follows from the explicit formula for $a$ in Lemma~\ref{lem: explicit resolution 1st corrector}.

          Conversely, if the frictionless solution and in turn \eqref{eq: definition of u} are sufficiently regular, then the probabilistic representation \eqref{eq: definition of u} provides a solution of the Second Corrector Equation \reff{eq: 2nd corrector equation}. This is exploited in Section \ref{sec:final}.
\end{remark}

\begin{remark}\label{rem:repQ}
As is well known, the dual minimizer for the frictionless version of the problem is typically the density process of a dual martingale measure $\Q$ (the ``marginal pricing measure''). It is given by $\partial_x v^0(r,S_r,Y_r,X_r)/\partial_x v^0(t,s,y,x)$, the normalized wealth-derivative of the corresponding value function, evaluated along the optimal frictionless wealth process (see, e.g., Section \ref{sec:final} for a simple example; compare \cite{schachermayer.01} for a general setting). If the initial portfolio equals the frictionless target, $\vartheta=\theta^0(\zeta)$, Theorem \ref{theo: main result}, \eqref{eq: definition of u}, and a first-order Taylor expansion therefore show that
\begin{align*}
v^\lambda(t,s,y,x,\vartheta)= v^0\Big(t,s,y,x-\mathrm{CE}(t,s,y,x)\Big)+o(\lambda^{1/2}),
        \end{align*}
        where
        $$
        \mathrm{CE}(\zeta)=\mathbb{E}_\Q\left[\lambda^{1/2}
            \int_t^T \frac{\Tr{
                  \frac{d\langle\theta^0\rangle_r}{dr} \Lambda^{1/2}(\Lambda^{1/2}\sigma_S \sigma_S^\top \Lambda^{1/2})^{1/2}\Lambda^{1/2}}
                }{\sqrt{2 R}}(r,S_r^{\zeta},Y_r^{\zeta},X^{\zeta,\theta^0}_r)
            dr
        \right].$$
        Hence, the certainty equivalent loss $\mathrm{CE}$ due to small price impact is given by the above $\Q$-expectation. This is the amount of initial endowment the investor would give up to trade without frictions. For a single risky asset, Formula \eqref{eq:cevloss} from the introduction obtains.

\end{remark}

Under the sufficient Condition {\rm B} for the abstract Assumption {\rm A} provided in Section \ref{sec: condition for Assumption B}, we can also produce an ``almost optimal'' policy that achieves the leading-order optimal performance in Theorem~\ref{theo: main result}:

\begin{theorem}{\bf(Almost Optimal Policy)}\label{theo: main result 2}
Suppose Assumptions \ref{prop:dpefric} and {\rm B} are satisfied. Then, the feedback control
$$\dot{\theta}^\Lambda(\zeta,\vartheta)=\lambda^{-1/2}\left(\frac{\Lambda^{-1/2}(\Lambda^{-1/2}\sigma_S \sigma_S^\top \Lambda^{-1/2})^{1/2}\Lambda^{1/2}}{(2 R)^{1/2}}\right)(\zeta)(\theta^0(\zeta)-\vartheta), \quad  \zeta \in \Df,\ \vartheta\in\R^d,$$
is optimal at the leading order $O(\lambda^{1/2})$, where $R(\zeta)=-\partial_xv^0(\zeta)/\partial_{xx}v^0(\zeta)$ denotes the risk tolerance of the frictionless value function $v^0$. For a single risky asset ($d=1$), this formula simplifies to
$$
    \dot{\theta}^\Lambda(\zeta,\vartheta)=
        \sqrt{\left(\frac{\sigma_S^2}{2\lambda\Lambda R}\right)(\zeta)}(\theta^0(\zeta)-\vartheta),
$$
in accordance with \eqref{eq:rate}.
\end{theorem}

This result is proved in Section~\ref{sec: condition for Assumption B}.

\section{Interpretation and Application}

In this section, we discuss the interpretation of our main results, their connections to the extant literature on portfolio choice with market frictions, and how they can be applied to determine utility-based option prices and hedging strategies. For simplicity, we mostly focus on the case of a single risky asset ($d=1$), and refer the interested reader to Guasoni and Weber \cite{guasoni.weber.13} for a detailed discussion of portfolio choice in a multivariate Black-Scholes model with price impact.

\subsection{Connections to Other Portfolio Choice Models with Price Impact}\label{sec:lit}

Let us first place our results in context by comparing them to the most closely related studies from the extant literature.

Garleanu and Pedersen \cite{garleanu.pedersen.13a,garleanu.pedersen.13b} consider investors with an infinite horizon and \emph{local} mean-variance preferences, who consume trading gains immediately. These investors trade several risky assets driven by arithmetic Brownian motion with returns following a stationary Markovian state variable. In this setting, and also for time-varying risk aversion or volatility, the optimal policy is characterized by the solution of a multidimensional nonlinear ordinary differential equation (henceforth ODE). The latter can be solved in closed form if the state variable is of Ornstein-Uhlenbeck-type, risk aversion and volatility are constant, and price impact is proportional to the assets' covariance matrix.\footnote{More generally, explicit solutions in a class of policies linear in the state variable are studied by \cite{dufresne.al.12}.}

Like Garleanu and Pedersen, Almgren and Li \cite{almgren.li.11} also focus on local mean-variance preferences. For a single risky asset following arithmetic Brownian motion, traded with constant linear price impact, they study the hedging of European options.  Explicit formulas for the optimal trading rate obtain under the assumption that the option's ``Gamma'' is constant.

Guasoni and Weber \cite{guasoni.weber.12,guasoni.weber.13} study a global optimization problem, namely an investor with constant relative risk aversion who maximizes utility from terminal wealth over a long horizon. For asset prices following geometric Brownian motions and price impact inversely proportional to the (representative) investor's wealth, they characterize the optimal policy and the corresponding welfare by the solution of an Abel ODE. In the limit for small trading costs, explicit formulas obtain, that are found to provide an excellent approximation of the exact solution.

The above studies differ with respect to preferences (local vs.\ global criteria, constant absolute vs. constant relative risk aversion), asset dynamics (arithmetic vs.\ geometric Brownian motions), price impacts (proportional to number of shares vs.\ proportional to amount of wealth traded), and time horizons (infinite vs.\ finite). For small price impact parameters, the broad conclusions nevertheless are the same in each model. Indeed, consider a single risky asset for simplicity.\footnote{The discussion for several risky assets is analogous, but the formulas are more involved and harder to interpret.} Then, for small trading costs, the trading rate -- interpreted appropriately in each model -- is linear in i) the displacement from the frictionless target position and ii) a constant determined by the constant market, cost, and preference parameters.

The present study extends and unifies these results. Our optimal policy in Theorem \ref{theo: main result 2} shows that -- asymptotically -- this structure indeed applies universally, even for general Markovian dynamics of asset prices, factors, and costs, as well as for arbitrary preferences over terminal wealth. In each case, the optimal trading rate (in numbers of shares traded) is given by
\begin{equation}\label{eq:reprate}
\dot\theta^\Lambda_t=\sqrt{\frac{(\sigma^S_t)^2}{2\Lambda_t R_t}}(\theta_t-\theta^\Lambda_t).
\end{equation}
If the driving Brownian motion is arithmetic, the asset's local variance $(\sigma^S_t)^2$ is constant, so that a constant trading rate obtains for a constant price impact $\Lambda$ proportional to the number of shares traded, and constant risk tolerance $R$, in line with the results of Garleanu and Pedersen \cite{garleanu.pedersen.13a,garleanu.pedersen.13b} as well as Almgren and Li \cite{almgren.li.11}. If the driving Brownian motion is geometric, as in Guasoni and Weber~\cite{guasoni.weber.12,guasoni.weber.13}, then $(\sigma^S_t)^2=\sigma^2 S_t^2$ is proportional to the squared asset price. Hence, a constant trading rate (in terms of relative wealth turnover $\dot{\theta}_t^\Lambda S_t/X^{\theta^\Lambda}_t$) obtains if risk tolerance $R_t$ is proportional to current wealth $X^{\theta^\Lambda}_t$ (i.e., if relative risk aversion is constant), and price impact is proportional to the square of the current stock price and inversely proportional to current wealth,  $\Lambda_t=\lambda S_t^2/X^{\theta^\Lambda}_t$ as in Guasoni and Weber \cite{guasoni.weber.12,guasoni.weber.13}.

For more general preferences as well as price and cost dynamics, the same policy remains optimal if variance, risk tolerance, and impact costs are updated dynamically. These inputs are all ``myopic'', in the sense that they are determined by the frictionless problem and the current state of the model. In particular, the same leading-order corrections obtain for local preferences (as in \cite{garleanu.pedersen.13a,almgren.li.11}) and for global maximization problems (like in \cite{guasoni.weber.12} and the present study). This parallels the situation for proportional transaction costs, where local and global preferences also lead to the same leading-order corrections for small costs \cite{soner.touzi.13,kallsen.muhlekarbe.13,martin.12,kallsen.li.13}.

\subsection{Connections to the Optimal Execution Literature}

The optimal trading rate \eqref{eq:reprate} can also be connected to the optimal execution literature, which studies how to split up a single, exogenously given order efficiently.

 Indeed, the key parameter -- the square root of variance, times risk aversion, divided by two times the trading cost -- also plays a pivotal role in the analysis of Almgren and Chriss \cite{almgren.chriss.01} as well as Schied and Sch\"oneborn \cite{schied.schoeneborn.09}. This can be related to the present model for dynamic portfolio choice as follows. Suppose that the investor currently holds a position $\theta^\Lambda_t$. In the absence of frictions ($\lambda=0$), she would immediately trade towards the optimal frictionless allocation $\theta^0_t$. With price impact ($\lambda>0$),  she instead trades towards the latter at the finite absolutely continuous rate $\dot{\theta}_t^\Lambda$ from \eqref{eq:reprate}. \emph{Locally}, the latter corresponds to the optimal initial execution rate for the order $\theta^\Lambda_t-\theta^0_t$ determined by Almgren and Chriss \cite{almgren.chriss.01} as well as Schied and Sch\"oneborn \cite{schied.schoeneborn.09}.\footnote{Almgren and Chriss \cite{almgren.chriss.01} consider mean-variance preferences, whereas Schied and Sch\"oneborn \cite{schied.schoeneborn.09} extend their analysis to general von Neumann-Morgenstern utilities.} The same remains true in a multidimensional setting, where optimal execution has been studied by Schied, Sch\"oneborn, and Tehranchi \cite{schied.al.10} as well as Sch\"oneborn \cite{schoeneborn.11}.

On each infinitesimally short time interval, the dynamic portfolio choice policy therefore corresponds to the Almgren-Chriss execution path towards the frictionless target position. That is, for small price impacts, the local trade scheduling is the same, with market, price impact, and preference parameters updated dynamically over time. The key difference is that there is not a single buy or sell order to be executed here; instead one tracks a moving target that evolves dynamically over time.

\subsection{Application to Utility-Based Option Pricing and Hedging}\label{sec:options}

Suppose that the investor under consideration has constant absolute risk aversion $\eta>0$, i.e., an exponential utility function $U(x)=-e^{-\eta x}$. Then, it is well known that a random endowment $H$ at the terminal time $T$ can be absorbed by a change of measure. To wit, defining
$$\frac{d\P^H}{d\P}= \frac{e^{-\eta H}}{\mathbb{E}[e^{-\eta H}]},$$
the investor's problem is then equivalent to the pure investment problem without random endowment under the equivalent probability $\P^H$. If the change of measure leaves the structure of the model intact, random endowments can therefore be dealt with without additional difficulties.

In the present setting, suppose the investor has sold a European option with payoff $h(S_T)$ at time $T$ for a premium $p$. Then, $H=p-h(S_T)$, so that the change of measure is governed by the Radon-Nikodym derivative $d\P^H/d\P=e^{\eta h(S_T)}/\mathbb{E}[e^{\eta h(S_T)}]$. Given sufficient regularity, the Markov property implies that the corresponding density process $Z^H_t=\mathbb{E}[\frac{d\P^H}{d\P}|\mathcal{F}_t]$ is given by a function $f(t,S_t,Y_t)$ of time, the underlying, and the state variable, which can be determined from It\^o's formula and the martingale property of $Z^H$. The model dynamics under $\P^H$ can in turn be computed with Girsanov's theorem by adjusting the drift rates of prices and state variables accordingly. If $f$ and its derivatives are sufficiently regular to satisfy Condition {\rm B} also under $\mathbb{P}^H$, then our main results,
Theorems \ref{theo: main result} and \ref{theo: main result 2}, still apply. In particular, this shows that the trading rate
of Theorem \ref{theo: main result 2} is universal, in that it applies both for pure investment problems (as in \cite{garleanu.pedersen.13a,guasoni.weber.12}), and option hedging (as in \cite{almgren.li.11}). The only change is the frictionless target strategy. The expansion of the value function from Theorem~\ref{theo: main result} in turn enables us to compute first-order approximations of utility-indifference prices \`a l\`a Hodges and Neuberger \cite{hodges.neuberger.89} as well as Davis, Panas and Zariphopoulou \cite{davis.al.93}.\footnote{For proportional transaction costs, a number of corresponding results have been obtained, formally \cite{WhWi97,kallsen.muhlekarbe.13a} and rigorously \cite{bichuch.13,bouchard.al.13,PoRo13}.
}

\subsection{Connections to Models with Proportional and Fixed Transaction Costs}\label{subsec:tac}

In the above sections, we have argued that the trading rate \eqref{eq:reprate} is ubiquitous in all kinds of optimization problems with small linear price impact. Now, we want to compare this policy to its counterparts for other market frictions, namely proportional and fixed transaction costs.

At first glance, the respective policies are radically different. With linear price impact, one always trades towards the frictionless target at a finite, absolutely continuous rate. In contrast, proportional and fixed transaction costs both lead to a ``no-trade region'' around the frictionless optimizer. In this region, investors remain inactive, and only trade once its boundaries are breached. This different ``fine structure'' is a consequence of the different penalizations of trades of various sizes: the quadratic trading costs induced by linear price impact are low for small trades, so that it is optimal to trade at all times. Conversely, they are prohibitively high for large orders, so that bulk trades (as for fixed costs) or ``local-time-type'' reflection (like for proportional costs) cannot be implemented, and the displacement from the frictionless target cannot be kept uniformly small. Compared to quadratic costs, proportional trading costs punish small trades more severely, leading to a no-trade region. However, as larger trades are penalized less, the position can always be kept inside this region by reflection at the boundaries (``pushing at an infinite rate''). With fixed costs, all trades are penalized alike. Whence, infinitely many small trades become infeasible and positions are immediately rebalanced to the frictionless target once the boundaries of the no-trade region are breached.

Despite these fundamental differences, all three market frictions nevertheless induce a surprisingly similar ``coarse structure'' as we now argue informally.\footnote{These arguments could be made rigorous similarly as in \cite{kallsen.li.13}.} Indeed, with proportional transaction costs $\Lambda_t$, investors always keep their actual position in a no-trade region around the frictionless target, whose halfwidth can be determined explicitly for small costs \cite{martin.12,soner.touzi.13,kallsen.muhlekarbe.13,kallsen.li.13}. In the interior of this region, the investor's portfolio evolves uncontrolled, with instantaneous reflection at the boundaries. At the leading order, the distribution of such diffusion processes can be approximated by the uniform stationary law for reflected Brownian motion \cite{rogers.04,janecek.shreve.04,goodman.ostrov.10,kallsen.muhlekarbe.13a,kallsen.muhlekarbe.13,kallsen.li.13}. Hence, the average squared deviation of the actual position from the frictionless target is given by one third of the halfwidth of the corresponding no-trade region:
$$\frac{1}{\sqrt[3]{12}}\left(\frac{R_t \Lambda_t}{(\sigma^S_t)^2}\right)^{2/3} \left(\sigma^{\theta^0}_t\right)^{4/3},$$
where $\sigma^{\theta^0}_t=\sqrt{d\langle \theta^0\rangle_t/dt}$ is the volatility of the frictionless optimizer $\theta^0$.

For fixed transaction costs, the portfolio again moves uncontrolled inside a no-trade region, but is rebalanced directly to the frictionless target position once its boundaries are breached. At the leading order, this leads to a deviation with probability density given by a ``hat function'', which arises as the stationary law for Brownian motion killed and restarted at the origin upon hitting the boundaries of a symmetric interval. As a result, the variance of the corresponding deviation from the frictionless optimizer equals one sixth of the halfwidth of the respective no-trade region:
$$\frac{1}{\sqrt{3}}\left(\frac{R_t \Lambda_t}{(\sigma^S_t)^2}\right)^{1/2} \sigma^{\theta^0}_t.$$
Up to the change of powers and a constant, the optimal policy is therefore determined by the same quantities in each case.

The optimal trading rate \eqref{eq:rate} with linear price impact leads to a deviation $\Delta_t=\theta^\Lambda_t-\theta^0_t$ following a mean-reverting diffusion process:
$$d\Delta_t= - \sqrt{\frac{(\sigma^S_t)^2}{2 \Lambda_t R_t}} \Delta_t dt +d\theta^0_t.$$
For small price impact ($\Lambda \sim 0$) this is locally an Ornstein-Uhlenbeck process (globally, if the frictionless target strategy follows Brownian motion and the mean-reversion speed is constant), with Gaussian stationary law and leading-order variance
$$
\sqrt{2}\left(\frac{R_t\Lambda_t}{(\sigma^S_t)^2}\right)^{1/2}\left(\sigma^{\theta^0}_t\right)^2.
$$
Again, the specific friction contributes the respective powers and a universal constant. In contrast, the input parameters and the corresponding comparative statics are universal: the effect of a small friction is large if market risk is high compared to the investor's risk tolerance, if trading costs are substantial, or if the frictionless target strategy prescribes a lot of rebalancing.

In summary, even though different trading costs lead to fundamentally different optimal policies on a ``microscopic'' level, the ``macroscopic'' picture turns out to be surprisingly robust.

\section{Proof of Theorem \ref{theo: main result}}\label{sec: PDE derivation}

This section contains the proof of our first main result, the asymptotic expansion of the value function $v^\lambda$ for small price impacts $\lambda \Lambda(\cdot) \sim 0$ from Theorem \ref{theo: main result}. Throughout, we write\footnote{Here, $E$ is the unique symmetric, positive definite matrix for which this representation holds true.}
$$\lambda=\varepsilon^4 \quad \mbox{and} \quad \Lambda(\zeta)=E(\zeta)^4,$$
to avoid the use of fractional powers. With a slight abuse of notation, we also index all quantities associated to the problem with price impact by $\varepsilon$. For example, we write $v^\varepsilon$ for the frictional value function $v^\lambda$, denote the corresponding optimal portfolio $\theta^\Lambda$ by $\theta^\varepsilon$, etc.\\

The strategy for the proof of Theorem~\ref{theo: main result} is as follows: 
   Lemma \ref{lem: explicit resolution 1st corrector} together with
    the results of Section \ref{sec: PDE characterization} (see Propositions \ref{prop: sub sol for u}, \ref{prop: super sol for u}, and \ref{prop: terminal condition for u})
    and Assumption \reff{ass: comparison for u} yield 
    $$
    \bar{u}_\ast(\zeta,\theta^0(\zeta))\ge u(\zeta)\ge \bar{u}^\ast(\zeta,\theta^0(\zeta)), \quad \mbox{for all $\zeta\in\Df$.}
    $$
            On the other hand, we show in Proposition~\ref{prop: u indep pi comparison}
            (the functions	 $u_\ast$ and $u^\ast$ therein are defined in Section \ref{sec: u ast}) that, for all $(\zeta,\vartheta)\in\Df\x\R^d$:
            \begin{equation*}\label{eq: intermed eq}
                \bar{u}_\ast(\zeta,\theta^0(\zeta))
                    \le \bar u_\ast(\zeta,\vartheta)-\varpi\circ\xib_1(\zeta,\vartheta)
                    \le \bar{u}^\ast(\zeta,\vartheta)-\varpi\circ\xib_1(\zeta,\vartheta)
                    \le \bar{u}^\ast(\zeta,\theta^0(\zeta)).
            \end{equation*}
            Together, these two estimates prove Theorem~\ref{theo: main result}.

    \subsection{Remainder Estimate}

       The first -- and the most tedious -- step is to estimate the remainders of the expansion in Theorem~\ref{theo: main result}.
        This parallels \cite[Remark 3.4, Section 4.2]{soner.touzi.13};
        see also \cite[Lemma 4.4]{bouchard.al.13}.

                \begin{lemma}\label{lem: remainder estimate}
          Suppose Assumption \reff{ass: v smooth} is satisfied, and recall $\xib_\veps(\zeta,\vartheta)=(\vartheta-\theta^0(\zeta))/\veps$.
          Fix $\veps>0$, two $C^{1,2}(\Df\x\R^d)$-functions $\phi$ and $w$, and define
          $$
            \psi^\veps:(\zeta,\vc)\longmapsto v^0(\zeta)-\veps^2\phi(\zeta,\vc)-\veps^4w^\veps(\zeta,\vc), \quad \mbox{with }
            w^\veps(\zeta,\vc):=w\circ\xib_\veps(\zeta,\vc)=w(\zeta,\xib_\eps(\zeta,\vc)).
          $$
          Set $D^\iota_\veps:=\{\partial_x\psi^\veps>0\}\cap\{\veps^2\partial_x(\phi+\veps^2w^\veps)/\partial_xv^0\le\iota\}$ for some $\iota<1$.
          Then:
          \b*
            \displaystyle
                \Lc^\vartheta\psi^\veps
                    &=&
                        \veps^2
                        \left(
                            \frac12\abs{\xib_\veps^\top \sigma_S}^2\partial_{xx}v^0
                            -\Lc^{\theta^0}\phi
                            -\frac12\Tr{c_{\theta^0} D^2_{\xi\xi}w}
                            +\Rc^\veps_\Lc
                        \right),
                \\
            \displaystyle
                \Hc^\veps\psi^\veps
                    &=&
                        \veps^2
                            \left(
                                \frac{(D_\xi w\circ\xib_\veps)^\top E^{-4}D_\xi w\circ\xib_\veps}{4\partial_xv^0}
                                +\Rc^\veps_\Hc
                            \right)
                +\hat\Lc^\veps\phi,
                \quad \mbox{on } D^\iota_\veps,
          \e*
          with
          \beq\label{eq: definition hat Lc}
            \hat\Lc^\veps\phi:=
                \frac
                    {(D_\vartheta\phi)^\top E^{-4}(D_\vartheta\phi+2\veps D_\xi w\circ\xib_\veps)}
                    {4\partial_xv^0}
                +
                \frac
                    {\veps^2\partial_x\phi}
                    {4(\partial_xv^0)^2}(D_\vartheta\phi)^\top E^{-4}D_\vartheta\phi,
          \eeq
          $\theta^0$ defined as in \reff{eq: optimum control frictionless},
          and where $\Rc^\veps_\Lc$ and $\Rc^\veps_\Hc$ are continuous maps defined on $D^\iota_\veps$ such that:
          \benumlabi{R}
            \item\label{eq: remainder for subsol} For each bounded set $B\subset \Df\x\R^d\x\R^d$,
                there exists $\veps_B>0$ such that
                $$\left\{\veps^{-1}\left(\abs{\Rc^\veps_\Lc}+\abs{\Rc^\veps_\Hc}\right)(\zeta,\vartheta):(\zeta,\vartheta,\xib_\veps(\zeta,\vartheta))\in B, \veps\in(0,\veps_B]\right\}$$
                is bounded;
            \item\label{eq: remainder for supersol} Let $B\subset \Df$ be a bounded set.
                Assume that $\phi\in C^\infty_b(B\x\R^d)$ and that $w$ satisfies \reff{eq: estimate for w}.
                Then, there exist $\veps_B>0$ and $C_B>0$ such that
                $$
                    \abs{\Rc^\veps_\Lc(\zeta,\vartheta)}+\abs{\Rc^\veps_\Hc(\zeta,\vartheta)}
                        \le C_B\left(
                            1+\veps\abs{\xib_\veps}+\veps^2\abs{\xib_\veps}^2
                        \right),
                $$
                for all $\veps\in(0,\veps_B]$ and $(\zeta,\vartheta)\in B\x\R^d$.
          \enumlab
        \end{lemma}

        \proof

                For the sake of clarity, write                $$
                    \bar\mu^0_{\vartheta}:=\begin{pmatrix} 0\\0\\\vartheta^\top \mu_S \end{pmatrix}
                    \et
                    \bar\sigma^0_\vartheta:=\begin{pmatrix} 0\\0\\\vartheta^\top \sigma_S \end{pmatrix},
                $$
                for any $\vartheta\in\R^d$. We work on $D^\iota_\veps$ and omit the corresponding arguments for brevity.

            \emph{Step 1: expand the linear operator.}
                First, use $\vartheta=\theta^0+\veps\xib_\veps$, obtaining
                \be
                    \displaystyle \Lc^\vartheta v^0
                        &=&\Lc^{\theta^0} v^0 + \bar\mu^0_{\veps\xibu_\veps}D_\zeta v^0+\Tr{\sigma_{\theta^0}(\bar\sigma^0_{\veps\xibu_\veps})^\top D^2_\zeta v^0}
                            +\frac12\veps^2\abs{\xib_\veps^\top \sigma_S}^2\partial_{xx}v^0\nonumber\\
                    \displaystyle
                        &=& \Lc^{\theta^0} v^0
                            +(\veps\xib_\veps)^\top \left(\mu_S\partial_xv^0+\sigma_S\bar\sigma^\top_0D_{(s,y)}(\partial_x v^0) +\sigma_S\sigma^\top_S\theta^0\partial_{xx}v^0 \right)
                        +\frac12\veps^2\abs{\xib_\veps^\top \sigma_S}^2\partial_{xx}v^0\nonumber\\
                    \displaystyle
                        &=&\frac12\veps^2\abs{\xib_\veps^\top \sigma_S}^2\partial_{xx}v^0,\nonumber
                \ee
                by the frictionless DPE \reff{eq: PDE friction less simplified with pf} and the first-order condition \reff{eq: optimum control frictionless} for the frictionless optimizer $\theta^0$,
                which hold due to Assumption \reff{ass: v smooth}.
                The same calculation also yields $\Lc^\vartheta (\veps^2\phi)
                        =\veps^2\Lc^{\theta^0} \phi + \veps^2 \Rc^\veps_1$,
                with
                \b*
                    \displaystyle
                        \Rc^\veps_1 &:=&(\veps\xib_\veps)^\top \left(\mu_S\partial_x\phi+\sigma_S\bar\sigma^\top_0D_{(s,y)}(\partial_x\phi) +\sigma_S\sigma^\top_S\theta^0\partial_{xx}\phi \right)
                        +\frac12\veps^2\abs{\xib_\veps^\top \sigma_S}^2\partial_{xx}\phi.
                \e*
                Now, observe $\xib_\veps=\xib_1/\veps$ so that, by definition of $\xib_1$ and $w^\veps$:
                \begin{gather*}
                    D_\zeta w^\veps=D_\zeta w-\frac1\veps D_\zeta\theta^0 D_\xi w,\\
                    D^2_{\zeta\zeta} w^\veps
                    =\frac1{\veps^2}D_\zeta\theta^0 D^2_{\xi\xi}w D^\top_\zeta\theta^0
                        -\frac1{\veps}
                            \left(
                                D_\zeta\theta^0 D^\top_\zeta(D_\xi w)
                                +D_\zeta(D_\xi w)D^\top_\zeta\theta^0
                                +D^2_{\zeta\zeta}\theta^0 D^\top_\xi w
                            \right)
                    +D^2_{\zeta\zeta}w.
                \end{gather*}
                As a result (recall \reff{eq: def delta multi d}):
                \begin{equation}\label{eq: first line for w}
                    \Lc^\vartheta(\veps^4 w^\veps)
                    =\veps^2\frac12\Tr{D^\top_\zeta\theta^0\sigma_{\theta^0}\sigma_{\theta^0}^\top D_\zeta\theta^0 D^2_{\xi\xi}w}
                        +\veps^2\Rc^\veps_2,
                \end{equation}
                with
                $$\bal
                    \Rc^\veps_2:=\;&
                        \veps^2\partial_t w^\veps
                        + \veps^2\mu_{\theta^0+\veps\xibu_\veps}\cdot D_\zeta w^\veps
                        +\veps^2\frac12\Tr{
                            \sigma_{\theta^0+\veps\xibu_\veps}\sigma_{\theta^0+\veps\xibu_\veps}^\top D^2_{\zeta\zeta} w^\veps
                            -\frac1{\veps^2}D^\top_\zeta \theta^0\sigma_{\theta^0}\sigma_{\theta^0}^\top D_\zeta \theta^0 D^2_{\xi\xi}w
                            }\\
                        =\;&
                            \veps^2\partial_tw
                            -\veps D_t\theta^0\cdot D_\xi w
                        + \veps^2\mu_{\theta^0+\veps\xibu_\veps}\cdot D_\zeta w
                        - \veps\mu_{\theta^0+\veps\xibu_\veps}\cdot D_\zeta\theta^0 D_\xi w\\
                        \;&
                        +\frac12\Tr{
                            \left(
                                \sigma_{\theta^0}\bar\sigma^{0\top}_{\veps\xibu_\veps}
                                +\bar\sigma^{0\top}_{\veps\xibu_\veps}\sigma_{\theta^0}^\top
                                +\bar\sigma^{0}_{\veps\xibu_\veps}\bar\sigma^{0\top}_{\veps\xibu_\veps}
                            \right)
                            D_\zeta\theta^0 D^2_{\xi\xi}w D_\zeta^\top\theta^0}\\
                        \;& -\Tr{\sigma_{\theta^0+\veps\xibu_\veps}\sigma_{\theta^0+\veps\xibu_\veps}^\top
                            \left(
                                \veps\left(
                                    D_\zeta\theta^0 D^\top_\zeta(D_\xi w)
                                    +D_\zeta(D_\xi w)D^\top_\zeta\theta^0
                                    +D^2_{\zeta\zeta}\theta^0 D^\top_\xi w
                                \right)
                                -\veps^2 D^2_{\zeta\zeta}w
                            \right)}.
                \eal$$
                The asserted estimates for $\Rc^\veps_\Lc:=\Rc^\veps_1+\Rc^\veps_2$ now follow from Assumption \reff{ass: v smooth},
                \reff{eq: estimate for w}, and the continuity of the coefficients of the SDEs
                \reff{eq:dynS}, \reff{eq:dynY}, \reff{eq:strat}, and \reff{eq:wealth}.

            \emph{Step 2: expand the nonlinear operator.} First, observe that $\partial_x\psi^\veps>0$ on $D^\iota_\veps$;
                whence (recall Remark \ref{rem: simplification HJB with friction}):
                $$
                    \Hc^\veps\psi^\veps=
                        \frac{(D_\vartheta\psi^\veps)^\top E^{-4}D_\vartheta\psi^\veps}{4\veps^4\partial_xv^0}
                            \x\frac{1}{1-\veps^2\partial_x(\phi+\veps^2w^\veps)/\partial_xv^0}.
                $$
                A first-order expansion of the right-hand side in turn gives
                $$
                    \Hc^\veps\psi^\veps=
                        \frac{(D_\vartheta\psi^\veps)^\top E^{-4}D_\vartheta\psi^\veps}{4\veps^4\partial_xv^0}
                        \left(1+\veps^2\frac{\partial_x\phi}{\partial_xv^0}\right)+\veps^2\Rc^\veps_3,
                $$
                with
                $$\bal
                    \abs{\Rc^\veps_3}
                    \le \;&
                        \frac{(D_\vartheta\psi^\veps)^\top E^{-4}D_\vartheta\psi^\veps}{4\veps^6\partial_xv^0}
                        \x\left(\veps^4\frac{\partial_xw^\veps}{\partial_xv^0}+
                            \frac{2}{(1-\iota)^3}
                            \x\frac{\veps^4\abs{\partial_x(\phi+\veps^2w^\veps)}^2}{(\partial_xv^0)^2}\right)\\
                    =\;& \frac
                        {(D_\vartheta\phi+\veps^2D_\xi w)^\top E^{-4}(D_\vartheta\phi+\veps^2D_\xi w)}
                        {4\partial_xv^0}\\
                            &\;\x\left(\frac{\veps^2\partial_xw-\veps\partial_x\theta^0\cdot D_\xi w}{\partial_xv^0}+
                                \frac{2\veps^2\abs{\partial_x\phi-\veps\partial_x\theta^0\cdot D_\xi w+\veps^2\partial_xw}^2}{(1-\iota)^3(\partial_xv^0)^2}\right),
                \eal$$
                where we have used for the first estimate that we are working on $D^\iota_\veps$.
                Thus, we compute
                $$\bal
                    \Hc^\veps\psi^\veps=\;&
                        \frac
                            {\veps^2(D_\xi w)^\top E^{-4}D_\xi w + (D_\vartheta\phi)^\top E^{-4}(D_\vartheta\phi+2\veps D_\xi w)}
                            {4\partial_xv^0}\\
                        \;&
                        +\frac{\veps^2\partial_x\phi}{4(\partial_xv^0)^2}(D_\vartheta\phi)^\top E^{-4}D_\vartheta\phi
                        +\veps^2(\Rc^\veps_3+\Rc^\veps_4),
                \eal$$
                with
                $$
                    \Rc^\veps_4:=
                        \frac{2\veps \partial_x\phi(D_\vartheta\phi)^\top E^{-4}D_\xi w+\veps^2(D_\xi w)^\top E^{-4}D_\xi w}{4(\partial_xv^0)^2}.
                $$
                Again, the asserted estimates for $\Rc^\veps_\Hc:=\Rc^\eps_3+\Rc^\eps_4$ now follow from the continuity of the involved functions,
                Assumption \reff{ass: v smooth}, and \reff{eq: estimate for w}. Together with Step 1, this completes the proof.
        \qed

    \subsection{The Adjusted Relaxed Semi-Limits $u^\ast, u_\ast$}\label{sec: u ast}

Unlike for models with proportional \cite{soner.touzi.13,possamai.al.12,bouchard.al.13} or fixed transaction costs \cite{altarovici.al.13}, the relaxed semilimits of $\bar u^\varepsilon=(v^0-v^\veps)/\veps^2$ do depend on the number of shares in the investor's portfolio for the present price impact model. As a result, the crucial simplification offered by homogenization apparently breaks down: the number of variables in the first-order correction term is the same as in the original frictional value function, rather than being reduced to the variables of its frictionless counterpart as in \cite{soner.touzi.13,possamai.al.12,bouchard.al.13,altarovici.al.13}.

However -- crucially --  the heuristic arguments from Section~\ref{sec:heuristics} suggest that $\bar{u}^\eps$ only depends on the initial number of risky shares $\vartheta$ through the quadratic function $\varpi$ determined by the first corrector equation. For intermediate times, this follows from the expansion of the frictional DPE, at the terminal time this is a consequence of the definition of the liquidation penalty in \eqref{eq: definition friction value function}. In fact, the latter is chosen precisely so that a simple quadratic function does the job here, see Remark \ref{rem:2.4}.

After subtracting this penalty term, the remaining first-order correction becomes independent of the current portfolio like for proportional and fixed costs.

To proceed, define for all $\veps>0$ the map
        $u^\veps:\Df\x\R^d\rightarrow\R$ by
        \beq\label{eq: def u_eps}
            u^\veps := \bar u^\veps - \veps^2 \varpi\circ\xib_\veps,
        \eeq
        where the normalized deviation $\xib_\veps(\zeta,\vartheta)=(\vartheta-\theta^0(\zeta))/\veps$ from the frictionless target $\theta^0$ is defined as in \reff{eq: fast variable} and $\varpi(\xi)$ is the solution of the first corrector equation constructed in
        Lemma \ref{lem: explicit resolution 1st corrector}. In analogy with \reff{eq: def bar u ast}, the corresponding relaxed semilimits are then defined as
        $$
            u^\ast(\zeta,\vartheta) := \Limsup_{\veps\to 0,(\zeta',\vartheta')\rightarrow(\zeta,\vartheta)} u^\veps(\zeta',\vartheta'),
            \qquad
            u_\ast(\zeta,\vartheta) := \Liminf_{\veps\to 0,(\zeta',\vartheta')\rightarrow(\zeta,\vartheta)} u^\veps(\zeta',\vartheta').
        $$
        Evidently, the families $\{\bar u^\veps: \veps>0\}$ and $\{u^\veps: \veps>0\}$ do not have the same relaxed semilimits.
        Indeed, $\bar u^\ast$ and $\bar u_\ast$ are not independent of the $\vartheta$-variable, as is immediately apparent for $t=T$.
        In contrast, we shall see that $u^\ast$ and $u_\ast$ do not depend on the $\vartheta$-variable (this is again evident for $t=T$).
        This will be verified \emph{a posteriori}, contrary to \cite{soner.touzi.13}, where this can be checked \emph{a priori} for the relaxed semilimits $\bar u^\ast$ and $\bar u_\ast$,
        and is crucially used to establish the main result.

       Define, for all $\eps>0$ and $(\zeta,\vc)\in\Df\x\R^d$,
            $$
                u^\eps_\ast(\zeta,\vc):=\frac{v^0(\zeta)-v^{\eps\ast}(\zeta,\vc)}{\varepsilon^2}
                \et
                u^{\eps\ast}(\zeta,\vc):=\frac{v^0(\zeta)-v^{\eps}_{\ast}(\zeta,\vc)}{\varepsilon^2},
            $$
            where $v^{\eps\ast}$ and $v^{\eps}_{\ast}$ denote the upper and lower semicontinuous envelopes of $v^\eps$, respectively,
            and observe that
        \beq\label{eq: same relaxed semi-limits}
            u^\ast(\zeta,\vartheta) = \Limsup_{\veps\to 0,(\zeta',\vartheta')\rightarrow(\zeta,\vartheta)} u^{\veps\ast}(\zeta',\vartheta'),
            \qquad
            u_\ast(\zeta,\vartheta) = \Liminf_{\veps\to 0,(\zeta',\vartheta')\rightarrow(\zeta,\vartheta)} u^\veps_\ast(\zeta',\vartheta').
        \eeq

        The following is a simple consequence of Assumptions~\reff{ass: u locally bounded from above}, \reff{ass: v smooth},
        as well as Lemma~\ref{lem: explicit resolution 1st corrector}:

        \begin{lemma}\label{lem: bar u unif bounded}
          Suppose Assumptions \eqref{ass: u locally bounded from above} and \eqref{ass: v smooth} are satisfied. Then,
          for all $(\zeta_o,\vartheta_o)\in\Df\x\R^d$, there are $r_o,\veps_o>0$ such that
          $$
            -\infty<u^\veps_\ast\le u^{\veps\ast}<+\infty, \quad \mbox{on } B_{r_o}(\zeta_o,\vartheta_o)\cap\Df, \pourtout\veps\in(0,\veps_o].
          $$
          In particular, the relaxed semilimits $u_\ast$ and $u^\ast$ are locally bounded.
        \end{lemma}

    \subsection{PDE Characterization Along the Frictionless Optimizer}\label{sec: PDE characterization}

        In this section, we show that $\zeta\in\Df\longmapsto u^\ast(\zeta,\theta^0(\zeta))=\bar{u}^\ast(\zeta,\theta^0(\zeta))$ and $\zeta\in\Df\longmapsto u_\ast(\zeta,\theta^0(\zeta))=\bar{u}_\ast(\zeta,\theta^0(\zeta))$
        are viscosity sub- and supersolutions, respectively, of the Second Corrector Equation~\reff{eq: 2nd corrector equation}, where
        $(a,\varpi)$ is the solution of the First Corrector Equation \reff{eq: first corrector equation} constructed in Lemma~\ref{lem: explicit resolution 1st corrector}.

    \subsubsection{Viscosity Subsolution Property}\label{sec: visco subsol}

        \begin{proposition}\label{prop: sub sol for u}
          Suppose Assumptions \ref{prop:dpefric} and { \rm A} are satisfied. Then,
          $\zeta\in\Df\longmapsto u^\ast(\zeta,\theta^0(\zeta))=\bar{u}^\ast(\zeta,\theta^0(\zeta))$ is a viscosity subsolution of the Second Corrector Equation
          \reff{eq: 2nd corrector equation} on $\Dfi$.
        \end{proposition}

        \proof

            Consider $\zeta_o\in \Dfi$ and $\vp\in C^{1,2}(\Dfi)$ such that
            \beq\label{eq: subsol Ac max strict for u}
                \max_{\zeta \in \Dfi}\strict(u^\ast(\zeta,\theta^o(\zeta))-\vp(\zeta))=u^\ast(\zeta_o,\vartheta_o)-\vp(\zeta_o)=0,
            \eeq
            where $\vartheta_o:=\theta^0(\zeta_o)$.
            We have to show that
            $
                -\Lc^{\theta^0}\vp(\zeta_o)\le a(\zeta_o).
            $

            \emph{ Step 1: provide a localizing sequence.} By \reff{eq: same relaxed semi-limits} and continuity of $\vp$,
                there exist $(\zeta^\veps,\vartheta^\veps)_{\veps>0}\subset\Dfi\x\R^d$ such that
                \beq\label{eq: sous sol Ac convergence t, x, y_eps}
                    \displaystyle (\zeta^\veps,\vartheta^\veps)\underset{\veps\to 0}{\longrightarrow}(\zeta_o,\vartheta_o)
                        \;,\quad
                    u^{\veps\ast}(\zeta^\veps,\vartheta^\veps)\underset{\veps\to0}{\longrightarrow}u^*(\zeta_o,\vartheta_o),
                    \quad\mbox{and}\quad p^\veps
                    \underset{\veps\to0}{\longrightarrow}0,
                \eeq
                where
                \beq\label{eq: definition p eps for referee}
                    p^\veps:= u^{\veps\ast}(\zeta^\veps,\vartheta^\veps)-\vp(\zeta^\veps).
                \eeq
                Now, on the one hand, Lemma \ref{lem: bar u unif bounded} guarantees the existence of $r_o, \veps_0>0$
                such that, with $B_o:=B_{r_o}(\zeta_o)\x B_{r_o}(\vartheta_o)$,\footnote{Here and in the following viscosity proofs, we always choose $r_o$ sufficiently small to guarantee that the respective neighborhoods are contained in $\Dfi$ resp.\ $\Df$.} we have
                $
                    b^* := \sup \left\{ u^{\veps\ast}(\zeta,\vartheta)\;,\; (\zeta,\vartheta)\in B_o \;,\; \veps\in(0,\veps_0] \right\}<\infty.
                $
                On the other hand, by Assumption \reff{ass: v smooth},
                there exists $\alpha\in(0,r_o]$ for which
                \beq\label{eq: subsol Ac pf bounded r_o}
                    \theta^0 \in\bar B_{\frac{r_o}4}(\vartheta_o), \quad \mbox{on }\bar B_\alpha(\zeta_o),
                \eeq
                and, for some $\iota>0$:
                \beq\label{eq: v_x greater iota}
                    2/\iota>-\partial_{xx}v^0 \wedge \partial_xv^0>\iota, \quad\mbox{on } \bar B_{\alpha}(\zeta_o).
                \eeq
                Now, choose $\mathbf d>0$ such that:
                \begin{equation*}\label{eq: subsol d for zeta zeta ' boundary}
                    \abs{\zeta-\zeta'}^4\ge\mathbf d,
                    \quad\pourtout (\zeta,\zeta')\in \left(\bar B_\alpha(\zeta_o)\backslash B_{\alpha/2}(\zeta_o)\right)\x\bar B_{\alpha/4}(\zeta_o).
                \end{equation*}
                By continuity of $\vp$, we have
                $
                    \sup \left\{ 2+b^\ast - \vp(\zeta)\; ; \; \zeta\in \bar B_{\alpha}(\zeta_o) \right\}=: M<+\infty,
                $
                and we in turn define the constant
                $
                    c_o:=M/(\mathbf d\wedge(\frac{r_o}4)^4)\;.
                $
                In view of \reff{eq: sous sol Ac convergence t, x, y_eps}, Assumption \reff{ass: v smooth}, as well as Lemma \ref{lem: explicit resolution 1st corrector},
                and reducing $\veps_o>0$ if necessary, we obtain:
                \beq\label{eq: subsol Ac eps_1}\bal
                        \left|\zeta^\veps-\zeta_o\right|\vee
                                \left|\vartheta^\veps-\vartheta_o\right|\le \frac{\alpha}{4}\;,
                        \quad \abs{\vartheta^\veps-\theta^0(\zeta^\veps)}^4\le1/3c_o\;,
                        &\\
                        \left| p^\veps\right|\le1, \quad
                            \mbox{and }
                        \varpi\circ\xib_1(\zeta^\veps,\vartheta^\veps)\le1/3,&
                    \quad\pourtout \veps\in(0,\veps_o].
                \eal\eeq
                Then, with $B_{\alpha}:=B_\alpha(\zeta_o)\x B_{r_o}(\vartheta_o)$, 
                observe that we still have
                \begin{equation*}\label{eq: subsol Ac sup for u eps}
                    u^{\veps\ast}(\zeta,\vartheta)\le b^\ast, \quad\pourtout (\zeta,\vartheta)\in \bar B_\alpha \mbox{ and }\veps\in(0,\veps_o].
                \end{equation*}

            \emph{ Step 2: construct a test function for $v^{\veps}_{\ast}$ and a sequence of local interior minimizers.}  For each $\veps\in(0,1)$, define
                $$
                    \phi^\veps:(\zeta,\vartheta)\in\Df\x\R^d\longmapsto
                        c_o \left(
                                \left| \zeta-\zeta^\veps \right|^4
                                + \left| \vartheta-\theta^0(\zeta) \right|^4
                            \right)
                $$
                and introduce the following subset of $\bar B_\alpha$:
                $$
                B_{o,\alpha}:=\bar{B}_{\alpha/2}(\zeta_o)\times \bar{B}_{r_0/2}(\vartheta_o).
                $$
                Recalling \reff{eq: subsol Ac pf bounded r_o}, \reff{eq: subsol Ac eps_1}, and the choice of $c_o$, it follows that
                \beq\label{eq: phi eps grand}
                     \displaystyle\phi^\veps(\zeta,\vartheta) \ge 2+b^*-\vp(\zeta),
                        \quad\pourtout\veps\le\veps_o\mbox{ and } (\zeta,\vartheta)\in \bar B_\alpha\backslash B_{o,\alpha}.
                \eeq
                On the other hand, the last estimate in the first line of \reff{eq: subsol Ac eps_1} gives:
                \beq\label{eq: subsol localization for pi eps}
                    \phi^\veps(\zeta^\veps,\vartheta^\veps)\le1/3.
                \eeq
                We now define, for all $\veps,\eta\in(0,1]$, the function
                $$
                    \psi^{\veps,\eta}:=
                        v^0
                        - \veps^2
                            \left(  p^\veps + \vp + \phi^\veps
                            \right)
                        - \veps^4 (1+\eta)\varpi\circ\xib_\veps,
                $$
                and show that $v^{\veps}_\ast-\psi^{\veps,\eta}$
                (or equivalently $I^{\veps,\eta}:=(v^{\veps}_\ast-\psi^{\veps,\eta})/\veps^2$)
                admits an interior local minimizer.
                By definition of $u^\veps$ in \reff{eq: def u_eps},
                $$\bal
                    I^{\veps,\eta}
                    &=
                        -u^{\veps\ast}
                        + ( p^\veps+\vp + \phi^\veps)
                        + \eta \varpi\circ\xib_1.
                \eal$$
                Combining the definition of $p^\veps$ with \reff{eq: subsol localization for pi eps}
                and the last term in \reff{eq: subsol Ac eps_1},
                we first notice that, for all $(\veps,\eta)\in(0,\veps_o]\x(0,1]$:
                $$
                    \inf_{\bar B_\alpha}I^{\veps,\eta}\le
                        \inf_{B_{o,\alpha}}I^{\veps,\eta}\le
                        I^{\veps,\eta}(\zeta^\veps,\vartheta^\veps)\le 2/3.
                $$
                On the other hand,
                because $\varpi\ge0$ by Lemma \ref{lem: explicit resolution 1st corrector}, it follows from
                \reff{eq: subsol Ac eps_1} and
                \reff{eq: phi eps grand} that
                $$
                    I^{\veps,\eta}(\zeta,\vartheta)\ge1, \quad\pourtout (\zeta,\vartheta)\in \bar B_\alpha\backslash B_{o,\alpha}\mbox{ and } \veps\in(0,\veps_o].
                $$
                Hence,
                by lower-semicontinuity of
                $I^{\veps,\eta}$
                and compactness of $B_{o,\alpha}$,
                there exists a minimizer
                $(\tilde \zeta^{\veps},\tilde \vartheta^{\veps})\in \bar B_{o,\alpha}\subset\bar B_\alpha$.
                (The latter also depends on $\eta$, but we do not explicitly note this dependence as it is of no importance here.)
               This minimizer satisfies, for all $\veps\in(0,\veps_o]$ and $\eta\in(0,1)$:
                \beq\label{eq: subsol Ac xi epx bounded}
                    I^{\veps,\eta}\left(\tilde \zeta^{\veps},\tilde \vartheta^{\veps}\right)\le0
                    \quad\mbox{and}\quad
                    \left|\veps\xib_\veps(\tilde\zeta^\veps,\tilde\vartheta^\veps)\right|
                            \vee\left| \tilde \zeta^\veps - \zeta_o \right|
                        \le r_1,
                \eeq
                for some constant $r_1>0$,
                where we recall that $\veps\xib_\veps(\tilde\zeta^\veps,\tilde\vartheta^\veps)=\tilde \vartheta^\veps-\theta^0(\tilde\zeta^\veps)$.

        \emph{Step 3: show that for each $\eta\in(0,1]$, there is $C_\eta>0$ such that
        $
            |\xib_\veps(\tilde\zeta^\veps,\tilde\vartheta^\veps)|\le C_\eta$, $\forall \veps\in(0,\veps_o].
        $}
            As $(\tilde{\zeta}^\veps,\tilde{\vartheta}^\veps)$ are interior local minimizers of $v_*^\veps-\psi^{\veps,\eta}$ by Step 2, the viscosity supersolution property of $v^\veps$ for \eqref{eq: viscosity solution for friction case} yields
            \beq\label{eq: subsol Ac supersol for v eps}
                -\left(\Lc^{\tilde\vartheta^\veps}+\Hc^\veps\right)\psi^{\veps,\eta}\left(\tilde \zeta^\veps,\tilde \vartheta^\veps\right)\ge0.
            \eeq
            Observe from \reff{eq: v_x greater iota}
            and \reff{eq: subsol Ac xi epx bounded} that, after possibly reducing $\veps_o>0$,
            we have $\partial_x\psi^{\veps,\eta}>0$ and
            $\veps^2\partial_x(\phi+\veps^2w^\veps)\le\iota\partial_xv^0$, for $\veps \in (0,\veps_o]$.
            Hence, the requirements of \reff{eq: remainder for subsol} in Lemma \ref{lem: remainder estimate} are satisfied so that, for all $\veps\in(0,\veps_o]$:
            \be
                \displaystyle \Lc^{\tilde\vartheta^\veps}\psi^{\veps,\eta}(\tilde\zeta^\veps,\tilde\vartheta^\veps) &=&
                    \veps^2
                    \left(
                            \frac12\abs{\xib_\veps^\top \sigma_S}^2\partial_{xx}v^0
                            -\Lc^{\theta^0}\bar\phi^\veps
                            -\frac12(1+\eta)\Tr{c_{\theta^0} D^2_{\xi\xi}\varpi}
                    \right)(\tilde\zeta^\veps,\tilde\vartheta^\veps)\nonumber\\
                    \displaystyle &&+\veps^2 \Rc^\veps_\Lc(\tilde\zeta^\veps,\tilde\vartheta^\veps),\nonumber\\
                \displaystyle \Hc^\veps\psi^{\veps,\eta}(\tilde\zeta^\veps,\tilde\vartheta^\veps)
                    &=&
                        \veps^2
                            \left(
                                \frac{(1+\eta)^2(D_\xi \varpi\circ\xib_\veps)^\top E^{-4}D_\xi \varpi\circ\xib_\veps}{4\partial_xv^0}
                                +\frac{\hat\Lc^\veps\bar\phi^\veps}{\veps^2}
                            \right)(\tilde\zeta^\veps,\tilde\vartheta^\veps)
                        \label{eq: simplification of Hc subsol}\\
                            \displaystyle && +\veps^2\Rc^\veps_\Hc.\nonumber
            \ee
            Here (recall \reff{eq: definition p eps for referee}),
            \beq\label{eq: definition bar phi eps for referee}
                \bar\phi^\veps:=p^\veps+\vp+\phi^\veps
            \eeq
            and $\Rc^\veps:=\Rc^\veps_\Lc+\Rc^\veps_\Hc$, which satisfies
            \beq\label{eq: subsol Ac remainder for Lc}
                \left|\Rc^\veps\right|(\tilde\zeta^\veps,\tilde\vartheta^\veps)\le c_1,
                \quad\pourtout \veps\in(0,\veps_o],
            \eeq
            for some constant $c_1>0$.
            Now, rewrite $\Lc^{\tilde\vartheta^\veps}\psi^{\veps,\eta}$ above using
            that $\varpi$ is a solution of the First Corrector Equation \reff{eq: first corrector equation}. For all $\veps\in(0,\veps_o]$,
            Estimate \reff{eq: subsol Ac supersol for v eps} then leads to:
            \beq\label{eq: subsol all operators}\bal
                &\left\{\frac\eta2\abs{\xib_\veps^\top \sigma_S}^2\partial_{xx}v^0
                            +\Lc^{\theta^0}\bar\phi^\veps
                    +(1+\eta)a
                                    -\Rc^\veps
                    +\frac{(1+\eta)(D_\xi \varpi\circ\xib_\veps)^\top E^{-4}D_\xi \varpi\circ\xib_\veps}{4\partial_xv^0}
                \right.\\
                &\quad
                    \left.
                                    -\frac{(1+\eta)^2(D_\xi \varpi\circ\xib_\veps)^\top E^{-4}D_\xi \varpi\circ\xib_\veps}{4\partial_xv^0}
                                -\frac{\hat\Lc^\veps\bar\phi^\veps}{\veps^2}
                    \right\}
                        (\tilde\zeta^\veps,\tilde\vartheta^\veps)\ge0.
            \eal\eeq
            Observe that, as $E$ is positive-definite and $\eta \geq 0$:
            $$
                \frac{[1+\eta-(1+\eta)^2](D_\xi \varpi\circ\xib_\veps)^\top E^{-4}D_\xi \varpi\circ\xib_\veps}{4\partial_xv^0}\le0.
            $$
            We prove in Step 4 below that there is a constant $c_2>0$ such that, for $\veps\in(0,\veps_o]$:
            \beq\label{eq: subsol Ac claim for dpi}
                                -\frac{\hat\Lc^\veps\bar\phi^\veps}{\veps^2}
                        (\tilde\zeta^\veps,\tilde\vartheta^\veps)\le c_2.
            \eeq
            Combining this with \reff{eq: subsol all operators}, \reff{eq: v_x greater iota}, \reff{eq: subsol Ac remainder for Lc},
            and the Ellipticity Condition \reff{eq: ellipticity}
            gives
            $$
                    c_1+c_2+
                            \left\{
                                    (1+\eta)a
                                    +\Lc^{\theta^0}\bar\phi^{\veps}
                            \right\}(\tilde\zeta^\veps,\tilde\vartheta^\veps)
                                    \ge(\iota\eta\gamma_o/2)\left|\xib_{\veps}\right|^2(\tilde\zeta^\veps,\tilde\vartheta^\veps),
                                    \quad \pourtout \veps\in(0,\veps_o],
            $$
            for some $\gamma_o>0$. The assertion of Step 3 now follows by taking into account the continuity of $a$ and $\Lc^{\theta^0}\bar\phi^{\veps}$ as well as
            \reff{eq: subsol Ac eps_1} and
            \reff{eq: subsol Ac xi epx bounded}.

            \emph{Step 4: prove \reff{eq: subsol Ac claim for dpi}.}
            Recall the definition of $\hat\Lc^{\veps}$ in \reff{eq: definition hat Lc};
            as $E$ and $k_2$ are positive-definite, it follows that
            $$\bal
                -\frac{\hat\Lc^\veps\bar\phi^\veps}{\veps^2}
                &\le
                        -\frac
                            {(1+\eta)(D_\vartheta\bar\phi^\veps)^\top E^{-4}(D_\xi \varpi\circ\xib_\veps)}
                            {2\veps\partial_xv^0}
                        -
                        \frac
                            {\partial_x\bar\phi^\veps}
                            {4(\partial_xv^0)^2}(D_\vartheta\bar\phi^\veps)^\top E^{-4}D_\vartheta\bar\phi^\veps\\
                &\le
                        -\frac
                            {(1+\eta)4c_o|\xib_1|^2\xib_1^\top E^{-4}k_2\xib_1}
                            {\veps^2\partial_xv^0}
                        -
                        \frac
                            {\partial_x\bar\phi^\veps}
                            {4(\partial_xv^0)^2}(D_\vartheta\bar\phi^\veps)^\top E^{-4}D_\vartheta\bar\phi^\veps\\
                &\le
                        -
                        \frac
                            {\partial_x\bar\phi^\veps}
                            {4(\partial_xv^0)^2}(D_\vartheta\bar\phi^\veps)^\top E^{-4}D_\vartheta\bar\phi^\veps,
            \eal$$
            where the second inequality follows from direct computations based on the definition of $\bar\phi^\eps$ in
            \reff{eq: definition bar phi eps for referee} and the construction of $\varpi$ in Lemma \ref{lem: explicit resolution 1st corrector}.
            By construction of $\bar\phi^\veps$, as well as \reff{eq: subsol Ac xi epx bounded} and \reff{eq: v_x greater iota},
            this yields the desired upper bound $c_2$ at $(\tilde\zeta^\veps,\tilde\vartheta^\veps)$.

        \emph{ Step 5: conclude the proof of the proposition.} By the previous step,
            $(\tilde \zeta^\veps,\xib_\veps(\tilde\zeta^\veps,\tilde\vartheta^\veps))_{\veps\in(0,\bar\veps_\eta]}$
            is uniformly bounded.
            Hence, there is $(\bar\zeta,\bar \xi)$ such that,
            possibly along a subsequence,
            $(\tilde \zeta^\veps,\xib_\veps(\tilde\zeta^\veps,\tilde\vartheta^\veps))\rightarrow(\bar \zeta,\bar \xi)$
            as $\veps\to 0$. Moreover, by \eqref{eq: subsol Ac max strict for u}, classical arguments in the theory of viscosity solutions give $\bar\zeta=\zeta_o$, see, e.g., \cite{CrIsLi92}.
            (Observe that $\bar\xi$ depends on $\eta$, but we shall see below that this dependence is harmless.)
            By \reff{eq: subsol Ac supersol for v eps},
            $$
                \lim_{\veps\to 0}
                -\frac1{\veps^2}\left(\Lc^{\tilde\vartheta^\veps}+\Hc^\veps\right)\psi^{\veps,\eta}
                    \left(\tilde \zeta^\veps,\tilde \vartheta^\veps\right)\ge0.
            $$
            Using \reff{eq: simplification of Hc subsol},
            we further deduce that
            $$\bal
                \lim_{\veps\to 0}
                    &\left(
                            -\frac12\abs{\xib_\veps^\top \sigma_S}^2\partial_{xx}v^0
                            +\Lc^{\theta^0}\vp
                            +\Lc^{\theta^0}\phi^\veps
                            +\frac{1+\eta}2\Tr{c_{\theta^0} D^2_{\xi\xi}\varpi\circ\xib_\veps}
                    \right.\\
                    &\qquad\left.
                                -\frac{(1+\eta)^2(D_\xi \varpi\circ\xib_\veps)^\top E^{-4}D_\xi \varpi\circ\xib_\veps}{4\partial_xv^0}
                            +\Rc^\veps
                            -\frac{\hat\Lc^\veps\phi^\veps}{\veps^2}
                    \right)
                \left(\tilde \zeta^\veps,\tilde \vartheta^\veps\right)\ge0,
            \eal$$
            where, by \reff{eq: remainder for subsol} in Lemma \ref{lem: remainder estimate}:
            $
                \Rc^\veps\left(\tilde \zeta^\veps,\tilde \vartheta^\veps\right)
                \rightarrow0,
                \quad\mbox{as }
                \veps\rightarrow0.
            $
            By definition of $\phi^\veps$ and Step 3,
            $(\Lc^{\theta^0}\phi^\veps-\frac{\hat\Lc^\veps\phi^\veps}{\veps^2})(\tilde \zeta^\veps,\tilde \vartheta^\veps)
            \rightarrow0$ as $\veps\rightarrow0$. Hence, also taking into account that $\varpi$ is a solution of the First Corrector Equation \reff{eq: first corrector equation}:
            \begin{equation}\label{eq:bla}
                    \left(
                            \Lc^{\theta^0}\vp
                            +\frac{\eta}2\Tr{c_{\theta^0} D^2_{\xi\xi}\varpi(\cdot,\bar\xi)}
                                -\frac{(2\eta+\eta^2)(D_\xi \varpi\circ(\cdot,\bar\xi))^\top E^{-4}D_\xi \varpi(\cdot,\bar\xi)}{4\partial_xv}
                                +a
                    \right)
                (\zeta_o)\ge0.
            \end{equation}
            Now, note that
            $$\frac{(2\eta+\eta^2)(D_\xi \varpi\circ(\zeta_o,\bar\xi))^\top E^{-4}D_\xi \varpi(\cdot,\bar\xi)}{4\partial_xv(\zeta_o)}\ge0$$
            due to \reff{eq: v_x greater iota}. Together with \eqref{eq:bla}, this shows
            $$
                    \left(
                            \Lc^{\theta^0}\vp
                            +\frac{\eta}2\Tr{c_{\theta^0} D^2_{\xi\xi}\varpi(\cdot,\bar\xi)}
                                +a
                    \right)
                (\zeta_o)\ge0.
            $$
      Finally, note that
            $\frac{\eta}2\mbox{Tr}[c_{\theta^0} D^2_{\xi\xi}\varpi(\zeta_o,\bar\xi)]
            =\eta\mbox{Tr}[c_{\theta^0} k_2(\zeta_o)]$
            does not depend on $\bar\xi$. We now send $\eta$ to zero to arrive at
            $
                    -\Lc^{\theta^0}\vp(\zeta_o)
                \le a(\zeta_o).
            $
            This completes the proof.
        \qed

    \subsubsection{Viscosity Supersolution Property}\label{sec: visco supersol}

        \begin{proposition}\label{prop: super sol for u}
          Suppose Assumptions \ref{prop:dpefric} and { \rm A} are satisfied. Then, $\zeta\in\Df\longmapsto u_\ast(\zeta,\theta^0(\zeta))=\bar{u}_\ast(\zeta,\theta^0(\zeta))$ is a viscosity supersolution of the Second Corrector Equation
          \reff{eq: 2nd corrector equation} on $\Dfi$.
        \end{proposition}

        \proof

        Consider $\zeta_o\in\Dfi$ and $\vp\in C^{1,2}(\Dfi)$ such that
        \begin{equation}\label{eq:min}
            \min_{\zeta \in \Dfi}(\mbox{strict})(u_*(\zeta,\theta^o(\zeta))-\vp(\zeta))=u_*(\zeta_o,\vartheta_o)-\vp(\zeta_o)=0,
        \end{equation}
        where $\vartheta_o:=\theta^0(\zeta_o)$. We have to show
        $
            -\Lc^{\theta^0}\vp(\zeta_o)\ge a(\zeta_o).
        $
            By \reff{eq: same relaxed semi-limits} and continuity of $\vp$,
            there exist $(\zeta^\veps,\vartheta^\veps)_{\veps>0}\subset\Dfi\x\R^d$ such that
            \begin{equation*}\label{eq: sur sol convergence t, s, p_eps}
                (\zeta^\veps,\vartheta^\veps)\underset{\veps\to 0}{\longrightarrow}(\zeta_o,\vartheta_o),
                \quad u^{\veps}_{\ast}(\zeta^\veps,\vartheta^\veps)\underset{\veps\to 0}{\longrightarrow}u_*(\zeta_o,\vartheta_o),
                \quad \mbox{and } p^\veps
                \underset{\veps\to 0}{\longrightarrow}0,
            \end{equation*}
            where $p^\veps:= u^\veps_\ast(\zeta^\veps,\vartheta^\veps)-\vp(\zeta^\veps)$.
            By Assumption \reff{ass: v smooth} and Lemma \ref{lem: explicit resolution 1st corrector},
            there are $r_o>0$ and $\veps_o\in(0,1]$ satisfying
                \beq\label{eq: supersol Ac t eps small enough}
                    |\zeta^\veps-\zeta_o| \le \frac{r_o}{2}, \quad
                    |p^\veps|\le 1,
                    \quad \mbox{and} \quad
                    \varpi\circ\xib_1(\zeta^\veps,\vartheta^\veps)\le1/3,
                    \quad\pourtout\veps\le\veps_0.
                \eeq
            Moreover, Assumption \reff{ass: v smooth} ensures the existence of $\iota>0$ such that
            \beq\label{eq: sursol v_x iota}
                    2/\iota>-\partial_{xx}v^0 \wedge \partial_xv^0>2\iota, \quad\mbox{on } \bar B_{r_o}(\zeta_o).
            \eeq

        \emph{ Step 1: for each $\veps\in(0,\bar\veps]$,
        provide
        a penalization function $\phi^{\veps}$, in order to construct a convenient test function for $v^\veps$ in Steps 2 and 3.
        Also provide a constant $\xi^*$, independent of $\veps$, that will be used in Steps 5 and 6.
        }

            As $\vp$ is smooth, there exists a constant $M<\infty$ such that
            \beq\label{eq: supersol Ac sup vp finite c_o}
                \sup\left\{\vp(\zeta)\; ; \; \zeta\in\bar B_{r_o}(\zeta_o)\right\} \le M-4.
            \eeq
            In view of \reff{eq: supersol Ac t eps small enough}, there is a finite $\mathbf d>0$ so that
            $|\zeta-\zeta^\veps|^4\ge\mathbf d$ for all $\zeta\in\partial B_{r_o}(\partial_o)$,
            and we choose $c_o>0$ such that $c_o\mathbf d\ge M$.
          With this notation, define
            $$
                \phi^\veps(\zeta) := \vp(\zeta) + p^\veps - c_o|\zeta-\zeta^\veps|^4,
            $$
            and observe from \reff{eq: supersol Ac t eps small enough},
            \reff{eq: supersol Ac sup vp finite c_o}, and the choice of $c_o$ that
            \beq\label{eq: supersol Ac diff vp and phi eps}
                    \phi^\veps(\zeta)
                    \le -3, \quad\pourtout \zeta\in\partial B_{r_o}(\zeta_o)
                    \mbox{ and } \veps\in(0,\veps_o].
            \eeq
            Recall the definition of $p^\veps$ and the last term in \reff{eq: supersol Ac t eps small enough},
            and observe for later use that
            \beq\label{eq: supersol Ac u - vp = 0 in t eps}
                -\bar u_\ast^\veps(\zeta^\veps,\vartheta^\veps)+\phi^\veps(\zeta^\veps)
                \ge-1/3, \quad\pourtout \veps\in(0,\veps_o].
            \eeq
            Now, on the one hand, combining \reff{eq: sursol v_x iota} with the positive-definiteness of $k_2E^{-4}k_2$ yields the existence of
            $\gamma_E>0$ such that
            \beq\label{eq: sursol gamma E}
                \frac{{\rm \mathbf{x}}^\top (k_2E^{-4}k_2)(\zeta) {\rm \mathbf{x}}}{4\partial_xv(\zeta)}
                \ge \gamma_E \abs{{\rm \mathbf{x}}}^2,
                \quad\pourtout(\zeta,{\rm \mathbf{x}})\in \bar B_{r_o}(\zeta_o)\x\R^d.
            \eeq
            On the other hand, \reff{eq: sursol v_x iota} together with the continuity of $E^{-4}$ and $k_2$ ensures that
            there is $K_E>0$ such that
            \beq\label{eq: sursol K E}
                \frac{\abs{E^{-4}}\abs{k_2}^2(\zeta)}{4\partial_xv(\zeta)}\le K_E,
                \quad\pourtout \zeta\in\bar B_{r_o}(\zeta_o).
            \eeq
            Also denote for later use by $K_0, K_2, K_{\theta^0}>0$ three finite constants such that
            \beq\label{eq: sursol K 2 K pf}
                2\abs{k_2(\zeta)}\le K_2\;,
                \quad
                \abs{c_{\theta^0}(\zeta)}\le 2K_{\theta^0},
                \et
                \abs{\Lc^{\theta^0}\phi^0(\zeta)}\le K_0,
                \quad\pourtout\zeta\in\bar B_{r_o}(\zeta_o),
            \eeq
            where $\phi^0(\zeta):=\vp(\zeta) - c_o|\zeta-\zeta_o|^4$.
            By a slight adaptation of \cite[Lemma 5.4]{possamai.al.12},
            there exist $(h^\eta)_{\eta\in(0,1]}\subset C^\infty(\R^d;[0,1])$ and $(a_\eta)_{\eta\in(0,1]}\subset (1,\infty)$ satisfying
            \beq\label{eq: sursol estim h eta}\begin{matrix}
                \displaystyle h^\eta=1, \quad \mbox{on } \bar B_1(0)\;,
                \quad
                h^\eta=0,\quad \mbox{on } \bar B^c_{a_\eta}(0)\;,\lpt6
                \displaystyle
                \abs{{\rm \mathbf{x}}}\abs{D_{{\rm \mathbf{x}}}h^\eta({\rm \mathbf{x}})}\le\eta
                \et
                \abs{{\rm \mathbf{x}}}^2\abs{D^2_{{\rm \mathbf{x}}{\rm \mathbf{x}}}h^\eta({\rm \mathbf{x}})}\le C^\ast,
            \end{matrix}\eeq
            for all ${\rm \mathbf{x}}\in\R^d$ and some constant $C^\ast>0$ independent of $\eta$.
            Finally, for each $\delta\in(0,1]$, we choose $\xi^{\ast,\delta}>0$ satisfying
            \begin{equation*}\label{eq: construction xi ast}
                (\xi^{\ast,\delta})^2=1+\frac{2[K_0+K_{\theta^0} K_2(6+C^\ast)}{\gamma_E(2\delta-\delta^2)}.
            \end{equation*}

        \emph{ Step 2: construct a ``first draft'' of a test function for $v^\veps$, that will be used to construct the ``true'' test function in Step 3.}
        
            For every $(\veps,\eta,\delta)\in(0,\veps_o]\x(0,1)^2$,
            define
            $$
                \psi^{\veps,\eta,\delta} := v^0
                    -\veps^2\phi^\veps
                    - \veps^4 (\varpi H^{\eta,\delta})\circ\xib_\veps,
            $$
            where
            $$H^{\eta,\delta}:\xi\in\R^d\longmapsto (1-\delta)h^\eta\left(\frac\xi{\xi^{*,\delta}}\right),$$
            the normalized deviation
            $\xib_\veps$ is defined as in \reff{eq: fast variable},
            and $\varpi$ is the solution of the first corrector equation from Lemma \ref{lem: explicit resolution 1st corrector}.
            We want to construct a local maximizer of $v^{\veps\ast}-\psi^{\veps,\eta,\delta}$
            (or equivalently $I^{\veps,\eta,\delta}:=\frac1{\veps^2}(v^{\veps\ast}-\psi^{\veps,\eta,\delta})$).
            However, it will turn out below that $\psi^{\veps,\eta,\delta}$ needs to be modified further to make this possible.
            Indeed, consider
            \begin{eqnarray*}
                \displaystyle I^{\veps,\eta,\delta}
                &=&
                        -\bar u^\veps_\ast+\phi^\veps +\veps^2(\varpi H^{\eta,\delta})\circ\xib_\veps.
                    \label{eq: supersol Ac I eps eta with bar u}
            \end{eqnarray*}
            By \reff{eq: supersol Ac u - vp = 0 in t eps} and because $\varpi H^{\eta,\delta}\ge0$,
            \beq\label{eq: supersol Ac I=0 in t eps}
                I^{\veps,\eta,\delta}(\zeta^\veps,\vartheta^\veps)\ge-1/3.
            \eeq
            On the other hand,
 the construction of $\varpi$ in Lemma \ref{lem: explicit resolution 1st corrector} together
            with \reff{eq:baru},
            \reff{eq: supersol Ac t eps small enough},
            \reff{eq: sursol K 2 K pf}
            $\eta,\delta\in(0,1)$, and
            $0\le H^{\eta,\delta}(\xi)\le\1_{\{|\xi|\le a_\eta\xi^*\}}$
            implies that, for all $(\zeta,\vartheta)\in\bar B_{r_o}(\zeta_o)\x\R^d$:
            \be
                \displaystyle I^{\veps,\eta,\delta}(\zeta,\vartheta) &\le&
                    \phi^\veps(\zeta)+K_2\veps^2|\xib_\veps|^2\1_{\{|\xibu_\veps|\le a_\eta\xi^{*,\delta}\}}(\zeta,\vartheta)\nonumber\\
                \displaystyle &\le& \phi^\veps(\zeta)+K_2\veps^2(a_\eta\xi^{*,\delta})^2\nonumber\\
                \displaystyle &\le& \phi^\veps(\zeta)+1,\qquad\pourtout \veps\le\veps_{\eta,\delta},\label{eq: supersol Ac min for I eps eta}
            \ee
            where $\veps_{\eta,\delta}:= \veps_o\wedge(K_2^{1/2}a_\eta\xi^{*,\delta})^{-1}$.
            Observe that in \reff{eq: supersol Ac min for I eps eta},
            unlike in the proof of the subsolution property in Proposition \ref{prop: sub sol for u}, deviations of $\vartheta$ from $\theta^0(\zeta)$ are not penalized by $\phi^\veps$.
            Hence, the supremum -- even if it is finite -- is not necessarily attained.

            Define the set $\Qc_o:=\{(\zeta,\vartheta)\in\Dfi\x\R^d: \zeta\in \bar B_{r_o}(\zeta_o)\}$,
            and observe from \reff{eq: supersol Ac min for I eps eta} that
            $$
                \sup_{(\zeta,\vartheta)\in\Qc_o}I^{\veps,\eta,\delta}(\zeta,\vartheta)\le
                    \sup_{\zeta\in\bar B_{r_o}(\zeta_o)}\left\{\phi^\veps(\zeta)+1\right\},\quad\pourtout \veps\le\veps_{\eta,\delta}.
            $$
            Hence, by compactness of $\bar B_{r_o}(\zeta_o)$, continuity of $\phi^\veps$,
            \reff{eq: supersol Ac t eps small enough},
            and the fact that $\veps_{\eta,\delta}\le\veps_o$,
            we have:
            \begin{equation*}\label{eq: inf of I eps eta in the interior}
                \Ic^{\veps,\eta,\delta}:=\sup_{(\zeta,\vartheta)\in\Qc_o}I^{\veps,\eta,\delta}(\zeta,\vartheta)<\infty, \quad \forall \veps\le\veps_{\eta,\delta}.
            \end{equation*}
            As a result, for each $\veps\in(0,\veps_{\eta,\delta}]$,
            there exists $(\hat \zeta^{\veps,\eta,\delta},\hat \vartheta^{\veps,\eta,\delta})\in$Int$(\Qc_o)$ satisfying
            \beq\label{eq: supersol Ac n-1 optimal minimizer for I}
                I^{\veps,\eta,\delta}\left(\hat \zeta^{\veps,\eta,\delta},\hat \vartheta^{\veps,\eta,\delta}\right)\ge \Ic^{\veps,\eta,\delta}-\frac{\veps^2}2.
            \eeq

        \emph{ Step 3: for each $\eta,\delta\in(0,1)$ and $\veps\in(0,\veps_{\eta,\delta}]$,
        finally provide a test function $\bar\psi^{\veps,\eta,\delta}$ and
        a test point $(\tilde \zeta^{\veps,\eta,\delta},\tilde \vartheta^{\veps,\eta,\delta})\in$Int$(\Qc_o)$, satisfying
        $$
            \max_{\Qc_o}(v^{\veps\ast}-\bar\psi^{\veps,\eta,\delta})
            =(v^{\veps\ast}-\bar\psi^{\veps,\eta,\delta})(\tilde \zeta^{\veps,\eta,\delta},\tilde \vartheta^{\veps,\eta,\delta}).
        $$}

            Introduce an even real-valued function $f\in C^\infty_b(\R)$ satisfying
            $0\le f\le1$, $f(0)=1$ and $f(x)=0$ whenever $|x|\ge1$.
            Also fix $\eta,\delta\in(0,1)$ and $\veps\in(0,\veps_{\eta,\delta}]$.
            Consider
            $$
                \bar\psi^{\veps,\eta,\delta}(\cdot,\vartheta) := \psi^{\veps,\eta,\delta}(\cdot,\vartheta)-\veps^4f\left( \left|\vartheta-\hat \vartheta^{\veps,\eta,\delta}\right| \right)
            $$
            as well as
            $$
                \bar I^{\veps,\eta,\delta}(\cdot,\vartheta) :=\frac1{\veps^2}\left(v^{\veps\ast}-\bar\psi^{\veps,\eta,\delta}\right)(\cdot,\vartheta)
                    =I^{\veps,\eta,\delta}(\cdot,\vartheta)+\veps^2f\left( \left|\vartheta-\hat \vartheta^{\veps,\eta,\delta}\right| \right).
            $$
            By \reff{eq: supersol Ac n-1 optimal minimizer for I} and $f(0)=1$,
            \beq\label{eq: link I eps eta and I eps eta n}
                \bar I^{\veps,\eta,\delta}\left(\hat \zeta^{\veps,\eta,\delta},\hat \vartheta^{\veps,\eta,\delta}\right)
                    =I^{\veps,\eta,\delta}\left(\hat \zeta^{\veps,\eta,\delta},\hat \vartheta^{\veps,\eta,\delta}\right)+\veps^2
                    \ge \Ic^{\veps,\eta,\delta}+\frac{\veps^2}{2}.
            \eeq
            Moreover, by definition of $f$, if $\vartheta\in\R^d$ satisfies $|\vartheta-\hat\vartheta^{\veps,\eta,\delta}|>1$ then
            $$
                \bar I^{\veps,\eta,\delta}(\zeta,\vartheta)=I^{\veps,\eta,\delta}(\zeta,\vartheta).
            $$
            Hence, setting $\Qc^\veps_1:=\{(\zeta,\vartheta)\in\Qc_o: |\vartheta-\hat\vartheta^{\veps,\eta,\delta}|\le1\}$
            and because $(\hat \zeta^{\veps,\eta,\delta}, \hat \vartheta^{\veps,\eta,\delta})\in\Qc^\veps_1$,
            this equality combined with \reff{eq: link I eps eta and I eps eta n}
            implies
            $$
                \sup_{\Qc_1^\veps}\bar I^{\veps,\eta,\delta}>\sup_{\Qc_o}I^{\veps,\eta,\delta}
                    \ge\sup_{\Qc_o\backslash\Qc_1^\veps}I^{\veps,\eta,\delta}=\sup_{\Qc_o\backslash\Qc_1^\veps}\bar I^{\veps,\eta,\delta}.
            $$
            As a result:
            $$
                \sup_{(\zeta,\vartheta)\in\Qc_o}\bar I^{\veps,\eta,\delta}(\zeta,\vartheta)=\sup_{(\zeta,\vartheta)\in\Qc_1^\veps}\bar I^{\veps,\eta,\delta}(\zeta,\vartheta).
            $$
            Thus, by upper-semicontinuity of $\bar I^{\veps,\eta,\delta}$
            and compactness of $\Qc_1^\veps$,
            there exists $(\tilde \zeta^{\veps,\eta,\delta},\tilde\vartheta^{\veps,\eta,\delta})\in\Qc_o$
            maximizing $\bar I^{\veps,\eta,\delta}$.
            In fact, $(\tilde \zeta^{\veps,\eta,\delta},\tilde \vartheta^{\veps,\eta,\delta})\in$Int$(\Qc_o)$,
            because \reff{eq: supersol Ac t eps small enough}, \reff{eq: supersol Ac I=0 in t eps}, $f\ge0$,
            and $\veps\in(0,\veps_{\eta,\delta}]$ give
            $$
                \bar I^{\veps,\eta,\delta}\left(\tilde \zeta^{\veps,\eta,\delta},\tilde \vartheta^{\veps,\eta,\delta}\right)
                    \ge \bar I^{\veps,\eta,\delta}\left(\zeta^\veps,\vartheta^\veps\right)
                    \ge I^{\veps,\eta,\delta}\left(\zeta^\veps,\vartheta^\veps\right)=0,
            $$
            whereas \reff{eq: supersol Ac diff vp and phi eps},
            \reff{eq: supersol Ac min for I eps eta}, $f\le1$, and $\veps\in(0,\veps_{\eta,\delta}]$ with $\veps_{\eta,\delta}\le1$ imply
            \begin{equation*}\label{eq: supersol Ac (t,x,y) on the boundary of the ball}
                \displaystyle
                \bar I^{\veps,\eta,\delta}
                \le I^{\veps,\eta,\delta}
                \le -2+\veps^2<0,
                \quad\mbox{on } \partial\Qc_o.
            \end{equation*}

        \emph{Step 4: show that, for each $\eta,\delta\in(0,1)$,
        $\{\xib_\veps(\tilde\zeta^{\veps,\eta,\delta},\tilde\vartheta^{\veps,\eta,\delta})\;;\;\veps\in(0,\bar\veps_{\eta,\delta}]\}$ is uniformly bounded and therefore
        converges along a subsequence towards some $\bar\xi^{\eta,\delta}\in\R^d$ as $\veps\to 0$.}
        
                By the previous step and Proposition \ref{eq: viscosity solution for friction case},
                \begin{equation*}\label{eq: super sol Ac for v eps}
                    -\left(\Lc^{\tilde\vartheta^{\veps,\eta,\delta}}+\Hc^\veps\right)\bar\psi^{\veps,\eta,\delta}
                        (\tilde \zeta^{\veps,\eta,\delta},\tilde \vartheta^{\veps,\eta,\delta})\le0.
                \end{equation*}
                Moreover,
                by \reff{eq: sursol v_x iota},
                construction of $H^{\eta,\delta}$, as $\xi^\ast$ does not depend on $\veps$ and $f\in C^\infty_b(\R)$,
                possibly diminishing $\veps_{\eta,\delta}>0$ yields
                $\partial_x\bar\psi^{\veps,\eta,\delta}(\tilde \zeta^{\veps,\eta,\delta},\tilde \vartheta^{\veps,\eta,\delta})>0$
                and $\veps^2\partial_x(\phi+\veps^2(\varpi H^{\eta,\delta})\circ\xib_\veps)\le\iota\partial_xv^0$.
                Applying \reff{eq: remainder for supersol} in Lemma \ref{lem: remainder estimate} then gives
                \beq\label{eq: sursol remainder}\bal
                        \left\{
                            -\frac12\abs{\xib_\veps^\top \sigma_S}^2\partial_{xx}v^0
                            +\Lc^{\theta^0}\bar\phi^\veps
                            +\frac12\Tr{c_{\theta^0} D^2_{\xi\xi}(\varpi H^{\eta,\delta})\circ\xib_\veps}
                            \right.\quad&\\
                            \left.
                                -\mathcal{R}^\veps_\Lc
                                -\frac{(D_\vartheta\bar\psi^{\veps,\eta,\delta})^\top E^{-4}D_\vartheta\bar\psi^{\veps,\eta,\delta}}{4\veps^6\partial_x\bar\psi^{\veps,\eta,\delta}}
                        \right\}&(\tilde \zeta^{\veps,\eta,\delta},\tilde \vartheta^{\veps,\eta,\delta})
                        \le 0,
                \eal\eeq
                where $\bar\phi^\veps(\cdot,\vartheta):=\phi^\veps-\veps^2f(|\vartheta-\hat\vartheta^{\veps,\eta,\delta}|)$ and,
                for some constant $C>0$ and all $\veps\in(0,\veps_{\eta,\delta}]$:
                $$
                    |\mathcal{R}^\veps_\Lc|(\tilde \zeta^{\veps,\eta,\delta},\tilde \vartheta^{\veps,\eta,\delta})\le C\left(\veps+\abs{\veps\xib_\veps}+\abs{\veps\xib_\veps}^2\right)(\tilde \zeta^{\veps,\eta,\delta},\tilde \vartheta^{\veps,\eta,\delta}).
                $$
                Assume now that $\{\xib_\veps(\tilde \zeta^{\veps,\eta,\delta},\tilde\vartheta^{\veps,\eta,\delta})\;;\;\veps\in(0,\bar\veps_{\eta,\delta}]\}$ is \emph{not} uniformly bounded along some subsequence.
                Then, by construction of $H^{\eta,\delta}$ and as $\xi^{\ast,\delta}$ does not depend on $\veps$,
                it follows that $(\varpi H^{\eta,\delta})\circ\xib_\veps$ and all of its derivatives vanish.
                On the other hand, $f\in C^\infty_b(\R)$ implies that
                $|(D_\vartheta\bar\psi^{\veps,\eta,\delta})^\top E^{-4}D_\vartheta\bar\psi^{\veps,\eta,\delta}|\le \veps^8 c_f$ for some constant $c_f$.
                Finally, by construction of $\bar\phi^{\veps,\eta,\delta}$
                and $\tilde\zeta^{\veps,\eta,\delta}\in\bar B_{r_o}(\zeta_o)$,
                we conclude that
                $$
                    \frac{(D_\vartheta\bar\psi^{\veps,\eta,\delta})^\top E^{-4}D_\vartheta\bar\psi^{\veps,\eta,\delta}}{4\veps^6\partial_x\bar\psi^{\veps,\eta,\delta}}
                        (\tilde \zeta^{\veps,\eta,\delta},\tilde \vartheta^{\veps,\eta,\delta})
                        \rightarrow0,
                        \quad\mbox{as } \veps\rightarrow0.
                $$
                After possibly increasing $C>0$, it follows that
                $$
                        \left\{
                            -\frac12\abs{\xib_\veps^\top \sigma_S}^2\partial_{xx}v^0
                            +\Lc^{\theta^0}\bar\phi^\veps
                        \right\}(\tilde \zeta^{\veps,\eta,\delta},\tilde \vartheta^{\veps,\eta,\delta})
                        \le C\left(1+\abs{\veps\xib_\veps}+\abs{\veps\xib_\veps}^2\right)(\tilde \zeta^{\veps,\eta,\delta},\tilde \vartheta^{\veps,\eta,\delta}).
                $$
                Denote by $\gamma>0$ the constant in \reff{eq: ellipticity} corresponding to the set $\bar B_{r_o}(\zeta_o)$.
                Combining \reff{eq: sursol v_x iota} with the continuity of $\Lc^{\theta^0}\bar\phi_\veps$
                and $\tilde\zeta^{\veps,\eta,\delta}\in \bar B_{r_o}(\zeta_o)$,
                we then obtain
                $$
                    \gamma\iota|\xib_\veps|^2(\tilde\zeta^{\veps,\eta,\delta},\tilde\vartheta^{\veps,\eta,\delta})
                    \le C\left(1+\abs{\veps\xib_\veps}+\abs{\veps\xib_\veps}^2\right)
                        (\tilde \zeta^{\veps,\eta,\delta},\tilde \vartheta^{\veps,\eta,\delta}).
                $$
                This contradicts the assumption that $\{\xib_\veps(\tilde \zeta^{\veps,\eta,\delta},\tilde\vartheta^{\veps,\eta,\delta})\;;\;\veps\in(0,\bar\veps_{\eta,\delta}]\}$ is unbounded. In particular, along a subsequence,
                $(\tilde \zeta^{\veps,\eta,\delta},\xib_\veps(\tilde \zeta^{\veps,\eta,\delta},\tilde \vartheta^{\veps,\eta,\delta}))$
                therefore converges towards some finite
                $(\bar \zeta^{\eta,\delta},\bar\xi^{\eta,\delta})\in \Dfi\x\R^d$
                as $\veps\rightarrow0$.

        \emph{Step 5: show that, for each $\delta\in(0,1)$,
        there is $\bar\eta_\delta\in(0,1)$
        such that $\{\bar\xi^{\eta,\delta}\;;\;\eta\in(0,\bar\eta_\delta]\}\subset B_{\xi^{\ast,\delta}}(0)$
        and therefore converges, possibly along a subsequence, to a point $\hat\xi^\delta\in B_{\xi^{\ast,\delta}}(0)$.}
        
                First, notice that the previous step implies that the requirements of \reff{eq: remainder for subsol}
                in Lemma \ref{lem: explicit resolution 1st corrector} are satisfied,
                so that the remainder $\Rc^\veps_\Lc(\tilde \zeta^{\veps,\eta},\tilde \vartheta^{\veps,\eta})$
                in \reff{eq: sursol remainder}
                converges to zero as $\veps\rightarrow0$.
                By continuity of all the involved functions, sending $\veps\to 0$ in \reff{eq: sursol remainder} gives
                \beq\label{eq: sursol for xi bounded xi delta}\begin{matrix}
                        \displaystyle\left\{
                            -\frac12\abs{(\bar\xi^{\eta,\delta})^\top\sigma_S}^2\partial_{xx}v^0
                            -\frac{[D_\xi(H^{\eta,\delta}\varpi)]^\top E^{-4}D_\xi(H^{\eta,\delta}\varpi)}{4\partial_xv^0}
                        \right\}(\bar\zeta^{\eta,\delta},\bar\xi^{\eta,\delta})\lpt6
                        \displaystyle\le
                            \left\{
                                \abs{\Lc^{\theta^0}\phi^0}
                                +\frac12\abs{\Tr{c_{\theta^0} D^2_{\xi\xi}(H^{\eta,\delta}\varpi)}}
                            \right\}(\bar\zeta^{\eta,\delta},\bar\xi^{\eta,\delta}).
                \end{matrix}\eeq
                We focus first on the right-hand side of this inequality.
                As $(\bar\zeta^{\eta,\delta})_{(\eta,\delta)\in(0,1)^2}\subset\bar B_{r_o}(\zeta_o)$,
                combining
                Lemma \ref{lem: explicit resolution 1st corrector} with
                \reff{eq: sursol K 2 K pf} and the last term in \reff{eq: sursol estim h eta}
                gives, for all
                $(\eta,\delta)\in(0,1)^2$:
                \beq\label{eq: sursol one before last ineq for xi bounded xi delta}
                    \left\{
                                \abs{\Lc^{\theta^0}\phi^0}
                                +\frac12\abs{\Tr{c_{\theta^0} D^2_{\xi\xi}(H^{\eta,\delta}\varpi)}}
                            \right\}(\bar\zeta^{\eta,\delta},\bar\xi^{\eta,\delta})
                            \le K_0+K_{\theta^0}(6K_2+C^\ast K_2).
                \eeq
                Consider now the left-hand side in \reff{eq: sursol for xi bounded xi delta}
                and omit the parameters $(\bar\zeta^{\eta,\delta},\bar\xi^{\eta,\delta})$ to ease notation.
                As $0\le\abs{H^{\eta,\delta}}\le(1-\delta)$ and $E^{-4}$ is positive definite, we have
                \be
                    \displaystyle
                        &&
                            -\frac12\abs{(\bar\xi^{\eta,\delta})^\top \sigma_S}^2\partial_{xx}v^0
                            -\frac{[D_\xi(H^{\eta,\delta}\varpi)]^\top E^{-4}D_\xi(H^{\eta,\delta}\varpi)}{4\partial_xv^0}
                        \nonumber
                    \\
                    \displaystyle
                    &&\;\ge
                        -\frac12\abs{(\bar\xi^{\eta,\delta})^\top \sigma_S}^2\partial_{xx}v^0
                        -(1-\delta)^2\frac{
                                [D_\xi\varpi]^\top E^{-4} D_\xi\varpi
                            }
                            {4\partial_xv^0}
                        \label{eq: sursol 1st line}\\
                    \displaystyle
                        &&\quad-\frac{
                                2(1-\delta)H^{\eta,\delta}\varpi\frac1{\xi^{\ast,\delta}}[D_\xi\varpi]^\top E^{-4}
                                    D_{{\rm \mathbf{x}}} h\left(\frac{\cdot}{\xi^{\ast,\delta}}\right)
                            }
                            {4\partial_xv^0}
                        \label{eq: sursol 2nd line}\\
                    \displaystyle
                        &&\quad-\frac{
                                (1-\delta)^2\varpi^2\left(\frac1{\xi^{\ast,\delta}}\right)^2
                                    [D_{{\rm \mathbf{x}}} h\left(\frac{\cdot}{\xi^{\ast,\delta}}\right)]^\top
                                    E^{-4}
                                    D_{{\rm \mathbf{x}}} h\left(\frac{\cdot}{\xi^{\ast,\delta}}\right)
                            }
                            {4\partial_xv^0}.\label{eq: sursol 3rd line}
                \ee
                Because $\varpi$ solves the First Corrector Equation \reff{eq: first corrector equation},
                 the terms in \reff{eq: sursol 1st line} satisfy
                $$\bal
                        -\frac12\abs{(\bar\xi^{\eta,\delta})^\top \sigma_S}^2\partial_{xx}v^0
                        -(1-\delta)^2\frac{
                                [D_\xi\varpi]^\top E^{-4} D_\xi\varpi
                            }
                            {4\partial_xv^0}
                        &=
                            (2\delta-\delta^2)\frac{
                        [D_\xi\varpi]^\top E^{-4} D_\xi\varpi
                    }
                    {4\partial_xv^0}\\
                    &\ge (2\delta-\delta^2)\gamma_E\abs{\bar\xi^{\eta,\delta}}^2,
                \eal$$
                where the second inequality follows from \reff{eq: sursol gamma E} and Lemma \ref{lem: explicit resolution 1st corrector},
                recall that $\bar\zeta^{\eta,\delta}\in\bar B_{r_o}(\zeta_o)$.
                Next, Lemma \ref{lem: explicit resolution 1st corrector},
                \reff{eq: sursol K E}, \reff{eq: sursol estim h eta}, and $\bar\zeta^{\eta,\delta}\in\bar B_{r_o}(\zeta_o)$ imply the following estimate for
                \reff{eq: sursol 2nd line}:
                $$\bal
                    -\frac{
                                2(1-\delta)H^{\eta,\delta}\varpi\frac1{\xi^{\ast,\delta}}[D_\xi\varpi]^\top E^{-4}
                                    D_{{\rm \mathbf{x}}} h\left(\frac{\cdot}{\xi^{\ast,\delta}}\right)
                            }
                            {4\partial_xv^0}
                    &\ge
                    -4(1-\delta)\eta K_E\abs{\bar\xi^{\eta,\delta}}^2.
                \eal$$
                Likewise, for \reff{eq: sursol 3rd line}, we have
                $$
                    -\frac{
                                (1-\delta)^2\varpi^2\left(\frac1{\xi^{\ast,\delta}}\right)^2
                                   \left [D_{{\rm \mathbf{x}}} h\left(\frac{\cdot}{\xi^{\ast,\delta}}\right)\right]^\top
                                    E^{-4}
                                    D_{{\rm \mathbf{x}}} h\left(\frac{\cdot}{\xi^{\ast,\delta}}\right)
                            }
                            {4\partial_xv^0}
                    \ge
                        -(1-\delta)^2\eta^2K_E\abs{\bar\xi^{\eta,\delta}}^2.
                $$
                Together, these three inequalities give
                \b*
                    &\displaystyle
                        -\frac12\abs{(\bar\xi^{\eta,\delta})^\top \sigma_S}^2\partial_{xx}v^0
                        -\frac{[D_\xi(H^{\eta,\delta}\varpi)]^\top E^{-4}D_\xi(H^{\eta,\delta}\varpi)}{4\partial_xv^0}
                    &\\
                    &\displaystyle
                                \ge
                        \abs{\bar\xi^{\eta,\delta}}^2
                        \left[
                            (2\delta-\delta^2)\gamma_E
                            -K_E(1-\delta)\eta\left(
                                    4
                                    +(1-\delta)\eta
                                \right)
                        \right].
                    &
                \e*
                Now, notice that $(2\delta-\delta^2)\gamma_E>0$ for all $\delta\in(0,1)$. Hence,
                for each $\delta\in(0,1)$, there exists $\bar\eta_\delta\in(0,1)$ such that
                $-K_E(1-\delta)\eta(4+(1-\delta)\eta)\ge-(2\delta-\delta^2)\gamma_E/2$ and in turn
                \beq\label{eq: sursol last ineq for xi bounded xi delta}
                        -\frac12\abs{(\bar\xi^{\eta,\delta})^\top \sigma_S}^2\partial_{xx}v^0
                        -\frac{[D_\xi(H^{\eta,\delta}\varpi)]^\top E^{-4}D_\xi(H^{\eta,\delta}\varpi)}{4\partial_xv^0}
                        \ge \frac{(2\delta-\delta^2)\gamma_E}2\abs{\bar\xi^{\eta,\delta}}^2.
                \eeq
                Finally, combining \reff{eq: sursol for xi bounded xi delta} with
                \reff{eq: sursol one before last ineq for xi bounded xi delta} and
                \reff{eq: sursol last ineq for xi bounded xi delta} gives
                $$
                    \abs{\bar\xi^{\eta,\delta}}^2
                    \le
                    \frac{2\left[K_0+K_{\theta^0}(6K_2+C^\ast K_2)\right]}{(2\delta-\delta^2)\gamma_E}
                    <(\xi^{\ast,\delta})^2,
                $$
                completing Step 5.

        \emph{Step 6: conclude the proof of the proposition.} First, observe that $\abs{\bar\xi^{\eta,\delta}}<\xi^{\ast,\delta}$, for all $\eta\in(0,\bar\eta_\delta]$, together with the definition of $H^{\eta,\delta}$ gives that
            $H^{\eta,\delta}(\bar\xi^{\eta,\delta})=1-\delta$ and that its derivatives vanish for all $(\delta,\eta)\in(0,1)\x(0,\bar\eta_\delta]$.
            Let $(\hat\zeta^\delta,\hat\xi^\delta)$ denote the limits of the (sub)sequence $(\bar\zeta^{\eta,\delta},\bar\xi^{\eta,\delta})$ as $\eta\rightarrow0$.
            By classical arguments in the theory of viscosity solutions (cf, e.g., \cite{CrIsLi92}), \eqref{eq:min} implies that $\hat\zeta^\delta=\zeta_o$.
            Combining \reff{eq: sursol remainder} with the fact that $\varpi$ solves the First Corrector Equation
            \reff{eq: first corrector equation} in turn yields
            $$\bal
                    0&\ge
                        \left\{
                            (2\delta-\delta^2)\frac{(D_\xi\varpi)^\top E^{-4}D_\xi(\varpi)}{4\partial_xv}
                            +\Lc^{\theta^0}\vp
                            +(1-\delta)a
                        \right\}(\zeta_o,\hat\xi^{\delta})\\
                        &\ge
                            \Lc^{\theta^0}\vp(\zeta_o)
                            +(1-\delta)a(\zeta_o).
            \eal$$
            Here, the last inequality follows directly from $\delta\in(0,1)$,
            Lemma \ref{lem: explicit resolution 1st corrector},
            \reff{eq: sursol v_x iota}, and the positive-definiteness of $E^{-4}$.
            As $a(\zeta_o)$ does not depend on $\delta$,
            sending $\delta\rightarrow0$ completes the proof of the proposition.
        \ep

    \subsubsection{ Terminal Condition}\label{sec: Terminal condition}

        \begin{proposition}\label{prop: terminal condition for u}
          Suppose Assumptions \ref{prop:dpefric} and { \rm A} are satisfied. Then,
          $$u^\ast(\zeta,\theta^0(\zeta))=u_\ast(\zeta,\theta^0(\zeta))=0,\quad  \mbox{for all } \zeta\in\Dfb.$$
        \end{proposition}

        \proof
       By definition, we have $u^\ast(\zeta,\theta^0(\zeta))\ge u_\ast(\zeta,\theta^0(\zeta))\ge 0$. Hence, it suffices to show $u^\ast(\zeta,\theta^0(\zeta)))\le0$, for all $\zeta\in\Dfb$.
       Assume to the contrary that there is $(\zeta_o,\delta)\in\Dfb\x(0,\infty)$ such that, with $\vartheta_o:=\theta^0(\zeta_o)$:
            \beq\label{eq: terminal sub sol contradiction}
                u^\ast(\zeta_o,\vartheta_o)\ge5\delta>0.
            \eeq

            \emph{ Step 1: provide a test function $\psi^\veps$ for $v^\veps_\ast$ and a local minimizer of $v^\veps_\ast-\psi^\veps$.}
            By \reff{eq: same relaxed semi-limits},
            there exist $(\zeta_\veps,\vartheta_\veps)_{\veps>0}\subset\Df\x\R^d$ such that
            \beq\label{eq: convergence points terminal}\begin{matrix}
                \displaystyle (\zeta_\veps,\vartheta_\veps)\underset{\veps\to 0}{\longrightarrow}(\zeta_o,\vartheta_o)
                \et
                u^{\veps\ast}(\zeta_\veps,\vartheta_\veps)\underset{\veps\to 0}{\longrightarrow}u^*(\zeta_o,\vartheta_o).
            \end{matrix}\eeq
            Assume that, possibly along a subsequence,
            $\zeta_\veps\in\Dfb$.
            Then, the terminal conditions in Assumption~\ref{prop:dpefric} and Proposition~\ref{theo: friction less value function}
            combined with $\varpi\ge0$ (cf.\ Lemma~\ref{lem: explicit resolution 1st corrector})
            yield
            $$u^{\veps\ast}(\zeta_\veps,\vartheta_\veps)=(\bar u^{\veps\ast}-\varpi\circ\xib_1)(\zeta_\veps,\vartheta_\veps)\le0,$$
            which contradicts \reff{eq: terminal sub sol contradiction} for small $\veps$.
            Therefore we can assume without loss of generality that
            \beq\label{eq: interior point terminal}
                \zeta_\veps\in \Dfi.
            \eeq
            By similar arguments as in the proof of Proposition \ref{prop: sub sol for u},
            Assumptions \reff{ass: v smooth} and \reff{ass: u locally bounded from above} combined with
            \reff{eq: terminal sub sol contradiction} and \reff{eq: convergence points terminal}
            enable us to find
            $r_o\ge\alpha>0$, $c_o>0, \iota>0$, and $\veps_o>0$ such that, for all $\veps\in(0,\veps_o]$:
            \be
                &\displaystyle
                        (\zeta_\veps,\vartheta_\veps)\in B_{o,\alpha}\;,
                        \quad
                        \abs{\vartheta_\veps-\theta^0(\zeta_\veps)}^2\le \delta/c_o,
                        \et
                        u^{\veps\ast}(\zeta_\veps,\vartheta_\veps)\ge4\delta,
                    \nonumber\label{eq: terminal u ge delta}
                &\\
                &\displaystyle
                        \partial_xv^0\wedge(-\partial_{xx}v^0)\ge2\iota
                        \et
                        \varpi\circ\xib_1\le\delta
                        \quad \mbox{on } \bar B_\alpha\;,
                    \label{eq: terminal v_p iota}
                &\\
                &\displaystyle
                        u^{\veps\ast}-\bar\phi(\cdot;\zeta_\veps)<0\quad \mbox{on } B_\alpha\backslash B_{o,\alpha},
                    \label{eq: terminal u - phi}
                &
            \ee
            where $B_\alpha:=(B_\alpha(\zeta_o)\cap\Df)\x B_{r_o}(\vartheta_o)$ as well as
            \b*
                &\displaystyle
                    B_{o,\alpha}:=
                        \left\{
                            (\zeta,\vartheta)\in \bar B_\alpha:
                                 \zeta\in\bar B_{\frac{\alpha}2}(\zeta_o)
                                    \mbox{ and }\vartheta\in\bar B_{\frac{\bar r_o}2}(\vartheta_o)
                        \right\},
                &\\
                &\displaystyle
                    \bar\phi:(\zeta,\vartheta;\zeta')\in \Df\x\R^d\x \Df\longmapsto c_o
                        \left( \abs{\zeta-\zeta'}^4 + \abs{\vartheta-\theta^0(\zeta)}^2 \right).
                &
            \e*
            By positive-definiteness and continuity of $E^{-4}$ combined with Assumption \reff{ass: v smooth}, there exists $\gamma_E>0$ such that
    \begin{equation}\label{eq:pd1}
        \frac{\xi^\top E^{-4}\xi}{\partial_xv^0}(\zeta)\ge \gamma_E\abs{\xi}^2,
        \quad\pourtout\xi\in\R^d \mbox{ and all }\zeta\in\bar B_\alpha.
    \end{equation}
    On the other hand, continuity of $\sigma_S$ and Assumption \reff{ass: v smooth} imply that there is $\bar\gamma>0$ such that
    \begin{equation}\label{eq:pd2}
        -\frac12\abs{\xi^\top \sigma_S}^2\partial_{xx}v^0\le \bar\gamma |\xi|^2,
        \quad\pourtout\xi\in\R^d \mbox{ and all }\zeta\in\bar B_\alpha.
    \end{equation}
            Hence, we can choose the constant $c_o$ in the definition of $\bar{\phi}$ large enough to satisfy \beq\label{eq: terminal choice c_o}
        \bar\gamma-c_o^2\gamma_E\le0.
    \eeq
        Define
            $$
                \phi^\veps:(\zeta,\vartheta)\in \Df\x\R^d\longmapsto \delta\frac{T-t}{T-t_\veps}+\bar\phi(\zeta,\vartheta;\zeta_\veps).
            $$
             Then, by Assumption \reff{ass: v smooth} and \reff{eq: interior point terminal},
            the function $\psi^\veps:=v^0-\veps^2\phi^\veps$ is smooth. The lower-semicontinuity of $v^\veps_\ast$
            in turn allows to deduce from \reff{eq: terminal u - phi} that, on $\bar B_\alpha$, the function $v^\veps_\ast-\psi^\veps$ has a local minimizer
            $(\tilde\zeta^\veps,\tilde\vartheta^\veps)\in B_{o,\alpha}\subset$Int$(B_\alpha)$. Moreover, by \reff{eq: terminal v_p iota},
            this minimizer satisfies
            $
                u^{\veps\ast}(\tilde\zeta_\veps,\tilde\vartheta_\veps)\ge \delta,
            $
            and repeating the arguments leading to \reff{eq: interior point terminal} shows $\tilde\zeta_\veps\in \Dfi$.

            \emph{ Step 2: conclude the proof.}
		In view of the previous step and Assumption \ref{prop:dpefric}, we have
                $$
                    -\left(\Lc^{\tilde\vartheta_\veps}+\Hc^\veps\right)\psi^\veps(\tilde \zeta_\veps,\tilde\vartheta_\veps)\ge0,
                    \quad\pourtout\veps\in(0,\veps_o].
                $$
                By construction of $\psi^\veps$ and because $(\tilde\zeta_\veps,\tilde\vartheta_\veps)\in\bar B_\alpha$,
                possibly reducing $\veps_o$ gives
                $$(\tilde\zeta_\veps,\tilde\vartheta_\veps)\in\{\partial_x\psi^\veps>0\}
                    \cap\{\veps^2\partial_x(\phi+\veps^2w^\veps)\le\iota\partial_xv^0\},$$
                so that \reff{eq: remainder for subsol} holds.
                Hence, Lemma \ref{lem: remainder estimate} yields
                $$
                    \left\{
                        -\frac12\abs{\xib_\veps^\top \sigma_S}^2\partial_{xx}v^0
                            +\Lc^{\theta^0}\phi^\veps
                            -\frac{(D_\vartheta\bar\phi)^\top E^{-4}D_\vartheta\bar\phi}{4\veps^2\partial_xv^0}
                            -\frac{\partial_x\bar\phi}{4(\partial_xv^0)^2}(D_\vartheta\bar\phi)^\top E^{-4}D_\vartheta\bar\phi
                            +\Rc^\veps
                    \right\}(\tilde\zeta_\veps,\tilde\vartheta_\veps)
                            \ge0,
                $$
                where $\Rc^\veps(\tilde\zeta_\veps,\tilde\vartheta_\veps)$ is uniformly bounded for $\veps\in(0,\veps_o]$.
                Thus, by Assumption \reff{ass: v smooth} and construction of $\psi^\veps$,
                there is a constant $C>0$ independent of $\veps$ such that:
                $$
                    \left\{
                            -\frac{\delta}{T-t_\veps}
                               -\frac12\abs{\xib_\veps^\top \sigma_S}^2\partial_{xx}v^0
                                -\frac{4c_o^2\abs{\xib_\veps}^\top E^{-4}\abs{\xib_\veps}}{4\partial_xv^0}
                    \right\}(\tilde\zeta_\veps,\tilde\vartheta_\veps)
                            \ge -C,
                            \quad\pourtout\veps\in(0,\veps_o].
                $$
                Recall that $(\tilde\zeta_\veps,\tilde\vartheta_\veps)\in\bar B_\alpha$; therefore, (\ref{eq:pd1}-\ref{eq: terminal choice c_o})
    yield
    $$
                                       \left\{-\frac12\abs{\xib_\veps^\top \sigma_S}^2\partial_{xx}v^0
                                -\frac{4c_o^2\abs{\xib_\veps}^\top E^{-4}\abs{\xib_\veps}}{4\partial_xv^0}
                                \right\}(\tilde\zeta_\veps,\tilde\vartheta_\veps)
                                \le(\bar\gamma-c_o^2\gamma_E)\xib^2_\veps(\tilde\zeta_\veps,\tilde\vartheta_\veps)\le0.
    $$
    As a result:
    $
        \delta/(T-t_\veps)\le C,
    $   
       for all $\veps\in(0,\veps_o]$.
               Note that the time component of $\zeta_o$ is $T$, because $\zeta_o\in \Dfb$. In contrast, the time component of $\zeta_\eps$ is $t_\eps$.
               For small $\veps$,
               this contradicts \reff{eq: convergence points terminal}, completing the proof.
                \qed


    \subsection{The Eikonal Equation}\label{sec: Eikonal Equation}

        This section is devoted to the proof of the following result, which is crucially used in the proof of our Main Theorem \ref{theo: main result}.

        \begin{proposition}\label{prop: u indep pi comparison}
          Suppose Assumptions \ref{prop:dpefric}, \reff{ass: v smooth} and \reff{ass: u locally bounded from above} are satisfied. Then,
          $$
            u_\ast(\zeta,\theta^0(\zeta))\le u_\ast(\zeta,\vartheta)\le u^\ast(\zeta,\vartheta)\le u^\ast(\zeta,\theta^0(\zeta)),
            \quad\pourtout (\zeta,\vartheta)\in\Df\x\R^d.
          $$
        \end{proposition}

        For notational convenience, define
        \begin{equation}\label{eq:n}
            \nf:(\zeta,\vartheta)\in\Df\x\R^d\longmapsto-2\partial_xv^0\partial_{xx}v^0\abs{\xib_1^\top \sigma_S}^2(\zeta,\vartheta).
        \end{equation}
        By Assumption \reff{ass: v smooth}, this is a nonnegative smooth function.

        \begin{lemma}\label{lem: gradient constraint for bar u}
         Suppose Assumptions \ref{prop:dpefric}, \reff{ass: v smooth} and \reff{ass: u locally bounded from above} are satisfied. Then,
          $\bar u^\ast$ and $\bar u_\ast$ are (discontinuous) viscosity sub- and supersolutions, respectively, of the \emph{Eikonal equation}
          $$
            (D_\vartheta\bar u^\ast)^\top E^{-4}D_\vartheta\bar u^\ast\le\nf,
            \quad\mbox{respectively}\quad
            (D_\vartheta\bar u_\ast)^\top E^{-4}D_\vartheta\bar u_\ast\ge\nf,
            \quad \mbox{on }\Dfi\x\R^d.
          $$
        \end{lemma}

        \proof

            We focus on the subsolution property; the supersolution property is obtained similarly.
 Consider $(\zeta_o,\vartheta_o)\in \Dfi\x\R^d$ and a smooth function $\vp$ such that
            $$
                \max_{\Dfi\x\R^d}\strict(\bar u^\ast-\vp)=(\bar u^\ast-\vp)(\zeta_o,\vartheta_o)=0.
            $$
            By definition of $\bar u^\ast$, there exist $(\zeta_\veps,\vartheta_\veps)_{\veps>0}\subset\Dfi\x\R^d$, for which
            \beq\label{eq: for lemma indep pi 1}\begin{matrix}
                \displaystyle (\zeta_\veps,\vartheta_\veps)\underset{\veps\to 0}{\longrightarrow}(\zeta_o,\vartheta_o)\;,
                \quad\bar u^{\veps\ast}(\zeta_\veps,\vartheta_\veps)\underset{\veps\to 0}{\longrightarrow}\bar u^\ast(\zeta_o,\vartheta_o),\\
                \displaystyle \mbox{and}\quad
                p^\veps:=\bar u^{\veps\ast}(\zeta_\veps,\vartheta_\veps)-\vp(\zeta_\veps,\vartheta_\veps)\underset{\veps\to 0}{\longrightarrow}0.
            \end{matrix}\eeq
            By Assumptions \reff{ass: v smooth}, \reff{ass: u locally bounded from above}, and \reff{eq: for lemma indep pi 1},
            there are $r_o, \veps_o, \iota>0$
            such that
            \beq\label{eq: for lemma indep pi 2}\begin{matrix}
                \displaystyle
                    2/\iota\ge -\partial_{xx}v\wedge \partial_xv\ge \iota \mbox{ on } B_o\;,
                    \quad
                    \abs{p^\veps}\le1\;,
                    \quad
                    (\zeta_\veps,\vartheta_\veps)\in B_{r_o}(\zeta_o,\vartheta_o),
                \\
                \displaystyle
                    \mbox{and}\quad
                    b^* := \sup \left\{ \bar u^{\veps\ast}(\zeta,\vartheta): (\zeta,\vartheta)\in B_o \;,\; \veps\in(0,\veps_o] \right\}<\infty,
            \end{matrix}\eeq
            where $B_o:=B_{4r_o}(\zeta_,\vartheta_o)$.
            The last estimate implies the existence of
            $\mathbf d>0$ for which
            $$
                \abs{\zeta-\zeta_\veps}^4+\abs{\vartheta-\vartheta_\veps}^4\ge\mathbf d,
                \quad\mbox{for all }
                (\zeta,\vartheta)\in\partial B_o
                \mbox{ and } \veps\in(0,\veps_o].
            $$
            On the other hand, continuity of $\vp$ yields
            $
                1\vee\sup \left\{ 2+b^\ast - \vp(\zeta,\vartheta): (\zeta,\vartheta)\in B_{o} \right\}=: M<+\infty,
            $
            so that we can choose a constant $c_o\ge M/\mathbf d>0$, independent of $\veps$.
            It follows that
            \beq\label{eq: for lemma indep pi 3}
                \phi^\veps(\zeta,\vartheta)\ge 2+b^\ast-\vp(\zeta,\vartheta),
                \quad\mbox{for all }
                (\zeta,\vartheta)\in\partial B_o
                \mbox{ and } \veps\in(0,\veps_o],
            \eeq
            where
            $$
                \phi^\veps:(\zeta,\vartheta)\in \Df\x\R^d\longmapsto c_o\left( \abs{\zeta-\zeta_\veps}^4 + \abs{\vartheta-\vartheta_\veps}^4 \right).
            $$
            Now, define $\psi^\veps:=v^0-\veps^2(p^\veps+\vp+\phi^\veps)$ and $I^\veps:=(v^\veps_\ast-\psi^\veps)/\veps^2$.
            Then, on the one hand, we have $I^\veps(\zeta_\veps,\vartheta_\veps)=0$.
            On the other hand,
            by definition of $p^\veps, \bar u^{\veps\ast}$, and $\phi^\veps$, as well as
            \reff{eq: for lemma indep pi 2} and \reff{eq: for lemma indep pi 3}:
            $I^\veps(\zeta,\vartheta)\ge1$ for all $(\zeta,\vartheta)\in\partial B_o$.
            By upper-semicontinuity of $I^\veps$,
            it follows that $I^\veps$ admits an interior minimizer $(\tilde\zeta_\veps,\tilde\vartheta_\veps)$ on $B_o$.
            Moreover, classical arguments \cite{CrIsLi92} show $(\tilde\zeta_\veps,\tilde\vartheta_\veps)\rightarrow(\zeta_o,\vartheta_o)$ as $\veps\to 0$.
            Hence, the viscosity supersolution property in Assumption \ref{prop:dpefric} implies
            $
                -( \Lc^{\tilde\vartheta_\veps}+\Hc^\veps)\psi^\veps(\tilde\eta_\veps,\tilde\vartheta_\veps)\ge0,
            $    
           for all $\veps\in(0,\veps_o]$.
            After possibly reducing $\veps_o>0$, we obtain $\partial_x\psi^\veps(\tilde\zeta_\veps,\tilde\vartheta_\veps)>0$. Hence,
            Lemma \ref{lem: remainder estimate}, continuity of $\vp$, and the fact that $\phi^\veps$ as well as its derivatives vanish
            as $\veps\rightarrow0$ yield
            $$
                \left(
                    -\frac12\abs{\xib_1^\top \sigma_S}^2\partial_{xx}v^0
                    +\veps^2 \Rc_\veps
                    -\frac{(D_\vartheta\vp)^\top E^{-4}D_\vartheta\vp}{4\partial_x\phi^\veps}
                \right)(\tilde\zeta_\veps,\tilde\vartheta_\veps)\ge0,
            $$
            where $\veps^2\Rc_\veps\to 0$ as $\veps\to 0$.
            Sending $\veps\to 0$ in turn gives
            $$
                    -\frac12\abs{\xib_1^\top \sigma_S}^2\partial_{xx}v^0(\zeta_o,\vartheta_o)
                \ge\frac{(D_\vartheta\vp)^\top E^{-4}D_\vartheta\vp}{4\partial_xv^0}(\zeta_o,\vartheta_o),
            $$
            which proves the asserted viscosity subsolution property.
        \ep\\

        Next, we show that $\bar u^\ast$ and $\bar u_\ast$ satisfy a generalized terminal condition as in
        \cite[Definition 7.4]{CrIsLi92}:

        \begin{lemma}\label{lem: gradient for u at the boundary}
          Suppose Assumptions \ref{prop:dpefric}, \reff{ass: v smooth} and \reff{ass: u locally bounded from above} are satisfied. Then,
          $\bar u^\ast$ and $\bar u_\ast$ are (discontinuous) viscosity sub- and supersolutions, respectively, of
          \b*
            &\displaystyle
                \min
                    \left\{
                        \bar u^\ast-\xib_1^\top k_2\xib_1 \;;\;
                        (D_\vartheta\bar u^\ast)^\top E^{-4}D_\vartheta\bar u^\ast-\nf
                    \right\}\le0, \quad \mbox{on } \Dfb\x\R^d,
            &\\
            &\displaystyle
            \mbox{and}\quad
                \max
                    \left\{
                        \bar u_\ast-\xib_1^\top k_2\xib_1 \;;\;
                        (D_\vartheta\bar u_\ast)^\top E^{-4}D_\vartheta\bar u_\ast-\nf
                    \right\}\ge0, \quad \mbox{on } \Dfb\x\R^d.
            &
          \e*

        \end{lemma}
        \proof
            Consider $(\zeta_o,\vartheta_o)\in\Dfb\x\R^d$ and a smooth function $\vp$ such that
            $$0=(\bar u^\ast-\vp)(\zeta_o,\vartheta_o)=\max_{\Df\x\R^d}\strict(\bar u^\ast-\vp).$$
            Assume that there is $\delta>0$ for which
            $
                \bar u^\ast(\zeta_o,\vartheta_o)-\xib_1(\zeta_o,\vartheta_o)^\top k_2(\zeta_o)\xib_1(\zeta_o,\vartheta_o)\ge\delta.
            $
            Repeating the arguments of Proposition \ref{prop: terminal condition for u}
            then gives
            $$
                    -\frac12\abs{\xib_1^\top \sigma_S}^2\partial_{xx}v^0(\zeta_o,\vartheta_o)
                \ge\frac{(D_\vartheta\vp)^\top E^{-4}D_\vartheta\vp}{4\partial_xv^0}(\zeta_o,\vartheta_o),
            $$
            and the subsolution property follows. The supersolution property is obtained similarly.
        \ep\\

        Next, we show that $\bar u^\ast, \bar u_\ast$ also solve the Eikonal equation if the $\zeta$-variable is fixed and they are considered as functions of the $\vartheta$-variable only:

        \begin{lemma}\label{lem: gradient for u just in pi}
          Suppose Assumptions \ref{prop:dpefric}, \reff{ass: v smooth} and \reff{ass: u locally bounded from above} are satisfied. Then,
          for any $\zeta_o\in\Df_<$, the functions
          $\vartheta \longmapsto \bar u^\ast(\zeta_o,\vartheta)$ and $\vartheta \longmapsto \bar u_\ast(\zeta_o,\vartheta)$ are viscosity sub- and supersolutions, respectively, of
          $$
            \left\{\bal
                (D_\vartheta\vp)^\top E^{-4}D_\vartheta\vp=\nf,
                &\quad\mbox{on }\R^d\backslash\{\theta^0(\zeta_o)\},\\
                \vp\ge\bar u_\ast(\zeta_o,\cdot)\mbox{ (resp. $\le\bar u^\ast(\zeta_o,\theta^0(\zeta_o))$)},
                    &\quad\mbox{on }\{\vartheta=\theta^0(\zeta_o)\}.
            \eal\right.
          $$
          For any $\zeta_o\in\Dfb$, the functions
          $\vartheta \longmapsto \bar u^\ast(\zeta_o,\vartheta)$ and $\vartheta \longmapsto \bar u_\ast(\zeta_o,\vartheta)$
          are viscosity sub- and supersolutions, respectively, of
          \b*
            &\displaystyle
                \min
                    \left\{
                        \bar u^\ast(\zeta_o,\cdot)-\xib_1^\top k_2(\zeta_o)\xib_1(\zeta_o,\cdot),
                        (D_\vartheta\bar u^\ast)^\top E^{-4}(\zeta_o)D_\vartheta\bar u^\ast(\zeta_o,\cdot)-\nf(\zeta_o,\cdot)
                    \right\}\le0,
            &\\
            &\displaystyle
                \max
                    \left\{
                        \bar u_\ast(\zeta_o,\cdot)-\xib_1^\top k_2(\zeta_o)\xib_1(\zeta_o,\cdot),
                        (D_\vartheta\bar u_\ast)^\top E^{-4}(\zeta_o)D_\vartheta\bar u_\ast(\zeta_o,\cdot)-\nf(\zeta_o,\cdot)
                    \right\}\ge0.
            &
          \e*
        \end{lemma}
        \proof

            We focus on the viscosity supersolution property on $\R^d\backslash\{\theta^0(\zeta_o)\}$ for $\zeta_o\in\Dfi$;
            the other properties are either evident, or obtained similarly (compare Lemma \ref{lem: gradient for u at the boundary}).

            Fix an arbitrary $\zeta_o \in \Df_<$, and consider a smooth function $\vp$ and $\vartheta_o\in\R^d\backslash\{\theta^0(\zeta_o)\}$ such that \beq\label{eq: gradient for u just in pi strict min}
                0=\bar u_\ast(\zeta_o,\vartheta_o)-\vp(\vartheta_o)=\min_{\R^d\backslash\{\vartheta_o\}}\strict(\bar u_\ast(\zeta_o,\cdot)-\vp(\cdot)).
            \eeq
            For each $n\in\N$, define
            \b*
                &\displaystyle
                    \psi^n:(\zeta,\vartheta)\in\Df\x\R^d\longmapsto \vp(\vartheta)-n\abs{\zeta-\zeta_o}^2,
                &\\
                &\displaystyle
                    \mbox{and}\quad
                    I^n:(\zeta,\vartheta)\in\Df\x\R^d\longmapsto \bar u_\ast(\zeta,\vartheta)-\psi^n(\zeta,\vartheta).
                &
            \e*
            By Lemma \ref{lem: bar u unif bounded}, there are $r_o>0$ and $b_o\ge0$ for which
            \beq\label{eq: gradient for u just in pi u bounded below}
                \bar u_\ast\ge -b_o,\quad\mbox{on } B_o,
            \eeq
            where $B_o:=\bar B_{r_o}(\zeta_o,\vartheta_o)$ and $r_o$ is chosen so that $B_o\subset\Dfi$.
            By compactness of $B_o$ and lower-semicontinuity of $I^n$,
            there is $(\zeta_n,\vartheta_n)\in B_o$ minimizing $I^n$ on $B_o$ for each $n\in\N$. Moreover,
            there exist $(\zeta^\ast,\vartheta^\ast)\in B_o$ such that
            $(\zeta_n,\vartheta_n)\rightarrow(\zeta^\ast,\vartheta^\ast)$ as $n\rightarrow+\infty$, possibly along a subsequence.
            Now, on the one hand, the minimality of $I^n(\zeta_n,\vartheta_n)$ on $B_o$ implies that
            $I^n(\zeta_n,\vartheta_n)\le I^n(\zeta_o,\vartheta_o)=\bar u_\ast(\zeta_o,\vartheta_o)-\vp(\vartheta_o),$ which is finite and does not depend on $n$.
            On the other hand, if $\zeta^\ast\ne\zeta_o$, \reff{eq: gradient for u just in pi u bounded below} gives
            $I^n(\zeta_n,\vartheta_n)\rightarrow+\infty$ as $n\rightarrow+\infty$.
            Hence, $\zeta^\ast=\zeta_o$.

            Observe now that $\bar u_\ast(\zeta_o,\vartheta_o)-\vp(\vartheta_o)=I^n(\zeta_o,\vartheta_o)\ge I^n(\zeta_n,\vartheta_n)$ implies
            $$
                \bar u_\ast(\zeta_o,\vartheta_o)-\vp(\vartheta_o)
                   \ge \liminf_{n\rightarrow+\infty} I^n(\zeta_n,\vartheta_n)
                   \ge \bar u_\ast(\zeta_o,\vartheta^\ast)-\vp(\vartheta^\ast).
            $$
            Therefore, $\vartheta^\ast=\vartheta_o$ by the strict minimum property in \reff{eq: gradient for u just in pi strict min}.
            Hence, $(\zeta_n,\vartheta_n)\in$ Int$(B_o)$ for sufficiently large $n$
            so that, by construction, $(\zeta_n,\vartheta_n)$ is a local minimum of $I^n$.
            Lemma \ref{lem: gradient constraint for bar u} in turn yields
            $
                (D_\vartheta\psi^n)^\top E^{-4}D_\vartheta\psi^n(\zeta_n,\vartheta_n)\ge \nf(\zeta_n,\vartheta_n).
            $
            As a result, sending $n\rightarrow+\infty$ finally proves the assertion after
            recalling from Lemma \ref{lem: explicit resolution 1st corrector} that $\nf$ is continuous.
        \ep\\

        In view of Lemma \ref{lem: gradient for u just in pi} and Proposition \ref{prop: terminal condition for u} define, for each $\zeta\in\Df$, the following subsets of $\R^d$:
        \b*
            &\displaystyle
                \Oc^{\zeta\ast}:=
                    \left\{
                        \vartheta\in\R^d: (D_\vartheta\bar u^\ast)^\top E^{-4}D_\vartheta\bar u^\ast(\zeta,\vartheta)\le\nf(\zeta,\vartheta)
                    \right\}
                        \backslash\{\theta^0(\zeta)\},
            &\\
            &\displaystyle
                \Oc^{\zeta}_{\ast}:=
                    \left\{
                        \vartheta\in\R^d: (D_\vartheta\bar u_\ast)^\top E^{-4}D_\vartheta\bar u_\ast(\zeta,\vartheta)\ge\nf(\zeta,\vartheta)
                    \right\}
                        \backslash\{\theta^0(\zeta)\}.
            &
        \e*
        (Here, the inequalities have to be understood in the viscosity sense.)
        By construction, $\bar u^\ast$ and $\bar u_\ast$ are
         viscosity sub- resp.\ supersolutions of the Eikonal equation
        $$
            (D_\vartheta\vp)^\top E^{-4}D_\vartheta\vp(\zeta,\cdot)=\nf(\zeta,\cdot),
        $$
        on $\Oc^{\zeta\ast}$ resp.\ $\Oc^{\zeta}_{\ast}$.
            Observe from the first part of Lemma \ref{lem: gradient for u just in pi} that, for all $\zeta\in\Dfi$ (i.e., before the terminal time),
            we have the following simplification: $\Oc^{\zeta\ast}=\Oc^{\zeta}_{\ast}=\R^d\backslash\{\theta^0(\zeta)\}$,
            or equivalently $(\Oc^{\zeta\ast})^c=(\Oc^{\zeta}_{\ast})^c=\{\theta^0(\zeta)\}$.
            Hence, we have the following estimate for all $\zeta\in\Dfi$:
            \beq\label{eq: compar on Oc ast compl}\bal
                    \bar u^\ast(\zeta,\cdot)&\le \bar u^\ast(\zeta,\theta^0(\zeta))+\xib_1(\zeta,\cdot)^\top k_2(\zeta)\xib_1(\zeta,\cdot),
                        \quad \mbox{on } (\Oc^{\zeta\ast})^c\;,\\
                    \bar u_\ast(\zeta,\cdot)&\ge \bar u_\ast(\zeta,\theta^0(\zeta))+\xib_1(\zeta,\cdot)^\top k_2(\zeta)\xib_1(\zeta,\cdot),
                        \quad \mbox{on } (\Oc^{\zeta\ast})^c\;.
            \eal\eeq
            For $\zeta\in\partial_T\Df$, such a simplification of $\Oc^{\zeta\ast}$ or $\Oc^{\zeta}_{\ast}$ is not available.
            However, combining the second part of Lemma \ref{lem: gradient for u just in pi} with Proposition \ref{prop: terminal condition for u},
            we find that \reff{eq: compar on Oc ast compl} holds for all $\zeta\in\partial_T\Df$ as well, and hence for all $\zeta\in\Df$.
        
        For later use, also note the following. For any $\zeta\in\Df$, we have $\theta^0(\zeta) \notin \Oc^{\zeta\ast}\cup\Oc^{\zeta}_{\ast}$. Hence, Assumption~\reff{ass: v smooth} and the ellipticity of $\sigma_S\sigma_S^\top$ imply the following estimate for the function $\mathfrak{n}$ defined in \eqref{eq:n}:
            $$
            \nf(\zeta,\vartheta)>0
            \on \Oc^{\zeta\ast}\cup\Oc^{\zeta}_{\ast}.
        $$
        Now introduce, for any $\zeta\in\Df$, the operator
        $$
            H^\zeta:(\vartheta,r,q)\in\R^d\x\R\x\R^d\longmapsto -\nf(\zeta,\vartheta)r^2+q^\top E^{-4}(\zeta)q.
        $$
        Also define, for $M>0$, the class $\Cc^-_M$ of negative functions $\mathbb{R}^d \to \mathbb{R}$ bounded from below by $-M$.
        We can then establish the comparison property for $H^\zeta$ on $\Cc^-_M$:

        \begin{lemma}\label{lem: u indep pi comparion for H}
          Suppose Assumption \reff{ass: v smooth} is satisfied. For any $\zeta\in\Df$, let $\Oc^\zeta$ be a subset of $\R^d$ for which
          $\nf(\zeta,\cdot)>0$ on $\Oc^\zeta$, and let $\vf^{1\zeta}$, $\vf^{2\zeta}$, $\vf^{3\zeta}\in\Cc^-_M$ (for some $M>0$)
          be lower-semicontinuous, smooth, and upper-semicontinuous functions,
          satisfying (in the viscosity sense for $\vf^{1\zeta}$ and $\vf^{3\zeta}$):
          \beq\label{eq: u indep pi eq for vf}
            H^\zeta(\cdot,\vf^{1\zeta},D_\vartheta\vf^{1\zeta})\ge0\;,\quad
            H^\zeta(\cdot,\vf^{2\zeta},D_\vartheta\vf^{2\zeta})=0,
            \et
            H^\zeta(\cdot,\vf^{3\zeta},D_\vartheta\vf^{3\zeta})\le0, \quad \mbox{on } \Oc^\zeta.
          \eeq
          Then if $\vf^{1\zeta}\ge\vf^{2\zeta}\ge\vf^{3\zeta}$ on $\R^d\backslash\Oc^\zeta$,
          we have $\vf^{1\zeta}\ge\vf^{2\zeta}\ge\vf^{3\zeta}$ on $\R^d$.
        \end{lemma}
        \proof

            Fix $\zeta\in\Df$ and drop it from the notation for clarity.
            We focus on the inequality $\vf^1\ge\vf^2$; the other one is obtained analogously.
		 For $\vf^1$ and $\vf^2$ as in the statement of the lemma,
            assume that there are $\bar\vartheta\in\Oc$ and $\alpha>0$ such that
            \beq\label{eq: u indep pi comparison contrad}
                \vf^{1}(\bar\vartheta)-\vf^{2}(\bar\vartheta)\le-\alpha<0,
            \eeq
            and work towards a contradiction.
            Choose $\beta\in C^\infty(\R^d)$, satisfying $0\le\beta\le1$,
            $\beta(0)=1$, $D_\vc\beta(0)=0$
            and $\beta(x)=0$ for all $x\in\R^d\backslash\bar B_1(0)$, and define, for all $\eta>0$:
            $$
                \Phi_\eta:\vartheta\in\R^d\longmapsto (\vf^1-\vf^2-2M \beta_\eta(\cdot-\bar\vartheta))(\vartheta),
                \quad \mbox{where }
                \beta_\eta(x):=\beta(x/\eta).
            $$
            By definition of $\Cc^-_M$ and boundedness of $\beta_\eta$, we have
            $\inf_{\R^d}\Phi_\eta>-\infty$.
            Hence, for each $\delta>0$, there is $\vartheta_\delta\in\R^d$ such that
            \beq\label{eq: u indep pi comparison pi delta}
                \Phi_\eta(\vartheta_\delta)\le\inf_{\R^d}\Phi_\eta+\delta.
            \eeq
            Pick a function $\chi\in C^\infty(\R^d)$ satisfying
            
            $$
                0\le\chi\le1,
                    \quad
                \chi(0)=1,
                    \quad
                \chi(x)=0\mbox{ if } |x|^2>1,
                \et |D_\vartheta\chi|\le c,
            $$
            for a constant $c>0$ independent of $\delta$.
            For each $\delta>0$, let $\chi_\delta:=\chi(\cdot-\vc_\delta)$
            Then,
            for all $\delta>0$:
            $$
                0\le\chi_\delta\le1,
                    \quad
                \chi_\delta(\vartheta_\delta)=1,
                    \quad
                \chi_\delta(\vartheta)=0\mbox{ if } |\vartheta-\vartheta_\delta|^2>1,
                \et |D_\vartheta\chi_\delta|\le c.
            $$
            Now define, for every $\eta,\delta>0$:
            $$
                \Psi_{\eta,\delta}:\vartheta\in\R^d\longmapsto (\Phi_\eta-2\delta\chi_\delta)(\vartheta)
                    =(\vf^1-\vf^2-2M \beta_\eta(\cdot-\bar\vartheta)-2\delta\chi_\delta)(\vartheta).
            $$
            On the one hand, \reff{eq: u indep pi comparison pi delta} in turn enables us to deduce that, for all $\eta,\delta>0$,
            $$
                \Psi_{\eta,\delta}(\vartheta_\delta)=\Phi(\vartheta_\delta)-2\delta\le\inf_{\R^d}\Phi_\eta-\delta<\inf_{\R^d}\Phi_\eta.
            $$
            On the other hand:
            $$
                \Psi_{\eta,\delta}(\vartheta)=\Phi_\eta(\vartheta)\ge \inf_{\R^d}\Phi_\eta,
                \quad\mbox{for all $\vartheta\in\R^d$ such that $|\vartheta-\vartheta_\delta|^2>1$}.
            $$
            As a result, the lower-semicontinuity of $\Psi_{\eta,\delta}$ yields that we can find a minimizing sequence
            $(\hat\vartheta_{\eta,\delta})_{\eta,\delta>0}$ for $\Psi_{\eta,\delta}$.
            Moreover, $\chi_\delta\ge0$, \reff{eq: u indep pi comparison contrad}, and the definition of $\Psi_{\eta,\delta}$ give
            \beq\label{eq: u indep pi for pi in Oc}
                \Psi_{\eta,\delta}(\hat\vartheta_{\eta,\delta})\le
                    \Psi_{\eta,\delta}(\bar\vartheta)\le-\alpha-2M.
            \eeq
            As $\beta,\chi_\delta\le1$, it follows that
            $
                (\vf^1-\vf^2)(\hat \vartheta_{\eta,\delta})\le -\alpha+2\delta<0,
            $
            for all $\delta<\alpha/2$.
           Hence, $\hat \vartheta_{\eta,\delta}\in\Oc$ for all such small $\delta$.
           As $\vf^1,\vf^2\in\Cc^-_M$ and $\chi_\delta\le1$,
            $$
                \Psi_{\eta,\delta}(\hat\vartheta_{\eta,\delta})\ge -M-2M\beta_\eta(\hat\vartheta_{\eta,\delta}-\bar\vartheta)-2\delta.
            $$
            Combined with \reff{eq: u indep pi for pi in Oc}, this leads to
            $$
                2M\beta_\eta(\hat\vartheta_{\eta,\delta}-\bar\vartheta)\ge M-2\delta>0,
                \quad\pourtout(\eta,\delta)\in(0,\infty)\x(0,M/2).
            $$
            By definition of $\beta_\eta$,
            it in turn follows that $\hat\vartheta_{\eta,\delta}\in\bar B_\eta(\bar\vartheta)$ for all
            $(\eta,\delta)\in(0,\infty)\x(0,M/2)$.

            Because
            $\hat\vartheta_{\eta,\delta} \in \Oc$, \reff{eq: u indep pi eq for vf} yields, for all $(\eta,\delta)\in(0,\infty)\x(0,M/2\wedge\alpha/2)$:
            $$
                H(\cdot,\vf^{1},D_\vartheta(\vf^2+2M \beta_\eta(\cdot-\bar\vartheta)+2\delta\chi_\delta))(\hat\vartheta_{\eta,\delta})\ge0
                \et
                H(\cdot,\vf^{2},D_\vartheta\vf^{2})(\hat\vartheta_{\eta,\delta})=0.
            $$
            As $\nf>0$ on $\Oc$,
            this gives
            $$
                [(\vf^{1})^2-(\vf^{2})^2]
                    (\hat\vartheta_{\eta,\delta})
                    -\frac{
                        [D_\vartheta^\top\varrho E^{-4} D_\vartheta \varrho(\hat\vartheta_{\eta,\delta})]^2
                        -[D^\top_\vartheta\vf^2E^{-4}D_\vartheta\vf^2(\hat\vartheta_{\eta,\delta})]^2
                        }
                        {\nf(\hat\vartheta_{\eta,\delta})}
                    \le0,
            $$
            with
            $\varrho:=(\vf^2+2M \beta_\eta(\cdot-\bar\vartheta)+2\delta\chi_\delta)$.
            As we have seen above that $\hat\vartheta_{\eta,\delta}\in\bar B_{\eta}(\bar\vartheta)$,
            there exists $\bar\vartheta_\eta\in\bar B_\eta(\bar\vartheta)$ such that $\hat\vartheta_{\eta,\delta}\rightarrow\bar\vartheta_\eta$
            as $\delta\rightarrow0$, possibly along a subsequence,
            and in turn $\bar\vartheta_\eta\rightarrow\bar\vartheta$ as $\eta\rightarrow0$.
            Hence, taking into account Assumption \reff{ass: v smooth}, continuity of $\vf^2$ and its gradient,
            $D_\vartheta\beta(0)=0$,
            and $|D_\vartheta\chi_\delta|\le c$ independent of $\delta$, the following limit obtains after sending first $\delta\rightarrow0$ and then $\eta\rightarrow0$:
            $$
                \liminf_{\delta,\eta\rightarrow0}(\vf^{1})^2(\hat\vartheta_{\eta,\delta})
                    -(\vf^{2})^2(\bar\vartheta)
                    \le0.
            $$
            Because $\vartheta\longmapsto(\vf^{1})^2(\vartheta)$ is lower-semicontinuous, it follows that
            $
                (\vf^{1}+\vf^2)(\vf^{1}-\vf^2)(\bar\vartheta)\le0,
            $
            As $\vf^{1}+\vf^2<0$ because $\vf^1,\vf^2\in\Cc^-_M$, this contradicts
            \reff{eq: u indep pi comparison contrad} and thereby proves the assertion.
        \ep\\

        Now, for all $(\zeta,\vartheta)\in\Df\x\R^d$, define the mappings $\bar\uf^\ast, \bar\uf_\ast, \tilde\uf^\ast, \tilde\uf_\ast:\Df\x\R^d\rightarrow\R$
         as follows:
        \b*
            \displaystyle
                &\bar\uf^\ast(\zeta,\vartheta)=-e^{-\bar u^\ast(\zeta,\vartheta)}\;,
                &\quad\tilde\uf^\ast(\zeta,\vartheta)=-e^{-(\bar u^\ast(\zeta,\theta^0(\zeta))+\xibu_1^\top(\zeta,\vartheta)k_2(\zeta)\xibu_1(\zeta,\vartheta))}\;,
            \\
                &\bar\uf_\ast(\zeta,\vartheta)=-e^{-\bar u_\ast(\zeta,\vartheta)},
                &\quad\tilde\uf_\ast(\zeta,\vartheta)=-e^{-(\bar u_\ast(\zeta,\theta^0(\zeta))+\xibu_1^\top(\zeta,\vartheta)k_2(\zeta)\xibu_1(\zeta,\vartheta))}.
        \e*

       One readily verifies that this change of variable produces bounded solutions to the Eikonal equation from Lemma~\ref{lem: u indep pi comparion for H}, for which a comparison principle holds on the class of bounded functions by Lemma \ref{lem: u indep pi comparion for H}:

        \begin{lemma}\label{lem: u indep pi viscosity uf and vf}
          Suppose Assumptions \ref{prop:dpefric}, \reff{ass: v smooth} and \reff{ass: u locally bounded from above} are satisfied.
          Then, for all $\zeta_o\in\Df$, the mappings
          $\bar\uf^\ast(\zeta_o,\cdot)$,
          $\tilde\uf^\ast(\zeta_o,\cdot)$,
          $\bar\uf_\ast(\zeta_o,\cdot)$, and
          $\tilde\uf_\ast(\zeta_o,\cdot)$ are
          viscosity subsolution, classical solution, viscosity supersolution, and classical solution, respectively, of
          \b*
            &\displaystyle
                H^{\zeta_o}(\cdot,\bar\uf^\ast, D_\vartheta\bar\uf^\ast)\le0,\quad
                H^{\zeta_o}(\cdot,\tilde\uf^\ast, D_\vartheta\tilde\uf^\ast)=0,
                \quad \mbox{on } \Oc^{\zeta_o\ast},
            &\\
            &\displaystyle
                H^{\zeta_o}(\cdot,\bar\uf_\ast, D_\vartheta\bar\uf_\ast)\ge0,
                    \et
                H^{\zeta_o}(\cdot,\tilde\uf_\ast, D_\vartheta\tilde\uf_\ast)=0.
                \quad \mbox{on } \Oc^{\zeta_o}_{\ast}.
            &
          \e*
          Moreover, $\bar\uf^\ast=\tilde\uf^\ast$ on $(\Oc^{\zeta_o\ast})^c$ and $\bar\uf_\ast=\tilde\uf_\ast$ on $(\Oc^{\zeta_o}_{\ast})^c$.
        \end{lemma}

Putting together all the previous results, we can now prove Proposition \ref{prop: u indep pi comparison}:\\

        \proofof{Proposition \ref{prop: u indep pi comparison}}
            First observe from \reff{eq:baru}, \reff{eq: def bar u ast}, and the definition of $\bar\uf^\ast$ and $\bar\uf_\ast$ that
            $-1\le\bar\uf_\ast\le\bar\uf^\ast<0$ so that $\bar\uf^\ast,\bar\uf_\ast\in\Cc^-_1$.
            Lemmata \ref{lem: u indep pi comparion for H}
            and \ref{lem: u indep pi viscosity uf and vf} in turn yield that, for any $(\zeta,\vartheta)\in\Df\x\R^d$:
            $$
                \tilde\uf_\ast(\zeta,\vartheta)\le\bar\uf_\ast(\zeta,\vartheta)
                \et
                \bar\uf^\ast(\zeta,\vartheta)\le\tilde\uf^\ast(\zeta,\vartheta).
            $$
            As $\bar\uf_\ast\le\bar\uf^\ast  $ by definition, this yields
            $$
                \bar u_\ast(\zeta,\theta^0(\zeta))+\xib_1^\top(\zeta,\vartheta)k_2(\zeta)\xib_1(\zeta,\vartheta)
                \le \bar u_\ast(\zeta,\vartheta)
                \le \bar u^\ast(\zeta,\vartheta)
                \le \bar u^\ast(\zeta,\theta^0(\zeta))+\xib_1^\top(\zeta,\vartheta)k_2(\zeta)\xib_1(\zeta,\vartheta).
            $$
            Proposition \ref{prop: u indep pi comparison} now follows from the definition of $u_\ast$ and $u^\ast$ in \reff{eq: def u_eps}.
        \ep

\section{Sufficient Conditions for Assumption {\rm A}}\label{sec: condition for Assumption B}

In this section, we provide a set of sufficient conditions for the abstract Assumption {\rm A} under which our Main Theorem \ref{theo: main result} holds. These sufficient conditions are typical for verification theorems (compare, e.g., \cite{touzi13}), and can be readily verified in concrete models, see Section \ref{sec:final}. Moreover, under these conditions, the policy from Theorem \ref{theo: main result 2} is indeed optimal at the leading order for small price impact costs.

Throughout, we assume that the frictionless value function $v^0$ and the corresponding optimal policy $\theta^0$ are given.
The function $v^0$ satisfies $\partial_x v^0 \vee (-\partial_{xx} v^0)>0$ and is a classical $C^{1,2}$-solution of the frictionless DPE \eqref{eq: PDE friction less simplified with pf}. The policy $\theta^0$ is characterized by the First-Order Condition \eqref{eq: optimum control frictionless} and belongs to $C^{1,2}$. In particular, Assumption (A1) is satisfied.\footnote{These assumptions are satisfied if a \emph{classical} frictionless verification theorem applies, cf., e.g., \cite{touzi13} and the references therein. In particular, they typically hold in the concrete models that can be solved explicitly.}

  For any positive function $f:\Df\rightarrow\R$, we denote by $\Cc^f$ the class of functions $g$ dominated by $f$ in the following sense (here, $\partial\Df$ denotes the spatial boundary of $\Df$):
            \beq\label{eq: growth for comparison}
                \limsup_{\zeta\rightarrow\partial\Df}\frac{\abs g(\zeta)}{1+\abs f(\zeta)}=0.
            \eeq
            With this notation, the sufficient conditions for the validity of Assumption {\rm A} read as follows:

            \begin{assumptionB}
              \benumlab{B}
                \item\label{ass: sufficient for comparison chi} There is a nonnegative function $\chi \in C^{1,2}$ satisfying $-\Lc^{\theta^0}\chi>0$ on $\Dfi$;
                \item\label{ass: sufficient hat u sol} There exists a classical $C^{1,2}$-solution $\hat u$ of the Second Corrector Equation \eqref{eq: 2nd corrector equation}, where the pair $(a,\varpi)$ is the solution of the First Corrector Equation \eqref{eq: first corrector equation} from Lemma \ref{lem: explicit resolution 1st corrector};
                \item\label{ass: sufficient functions in class Cc} $\hat{u}$ and the function $u$ defined though the Probabilistic Representation \reff{eq: definition of u} belong to $\Cc^\chi$;
                \item\label{ass: sufficient strong sol}
                    The feedback policy
                    \begin{equation*}
                    \dot{\theta}^\veps(\zeta,\vartheta) :=-\frac{[E^{-4}D_\xi\varpi]\circ\xib_\veps(\zeta,\vartheta)}{2\veps \partial_x v^0(\zeta)}= \frac{E^{-2}(E^{-2} \sigma_S \sigma_S^\top E^{-2})^{1/2} E^2}{\veps^2 (-2\partial_x  v^0/\partial_{xx}  v^0)^{1/2}}(\zeta)\times(\theta^0(\zeta)-\vartheta),
                    \end{equation*}
                    from Theorem \ref{theo: main result 2} is an admissible control.
                                   \item\label{ass: sufficient growth condition} Set
                $\hat v^\veps:=v^0-\veps^2\hat u-\veps^4\varpi\circ\xib_\veps$.
                 For every $\veps>0$, there is a function $\gamma^\veps$ such that
                    $|\hat v^\veps| \le \gamma^\veps$ on $\Df\x\R$ and, for all $(\zeta,\vartheta,\veps)\in \Df\x\R\x(0,\infty)$:
                    $$
                        \sup_{t\le r\le T}\gamma^\veps\left( r,S^\zeta_r,Y^\zeta_r, X^{\zeta,\vartheta,\veps}_r,\theta^{t,\vartheta,\veps}_r\right)
                        \in L^1.
                    $$
                \item\label{ass: sufficient control of the remainder} The remainder $\Rc^\veps_\Lc$ of Lemma \ref{lem: remainder estimate}, computed for
                    $\psi^\veps=\hat v^\veps$, satisfies:
                    $$
                        \Esp{\int_t^T
                            \abs{\Rc_\Lc^\veps+\tilde \Rc}
                            \left(r,S^\zeta_r,Y^\zeta_r, X^{\zeta,\vartheta,\veps}_r,\theta^{t,\vartheta,\veps}_r\right)dr}
                            \le \veps\beta(\zeta,\vartheta),
                    $$
                    for some continuous function $\beta:\Df\x\R^d\rightarrow\R$,
                    where, for all $(\zeta,\vartheta)\in \Df\x\R^d$:
                    $$
                        \tilde \Rc(\zeta,\vartheta)
                        :=\frac{[(D_\xi\varpi)^\top E^{-4}D_\xi\varpi]\circ\xib_1}{4(\partial_x v^0)^2}
                            \left(
                                \partial_x\hat u-\partial_x\theta^0 D_\xi\varpi\circ\xib_1+\partial_x\varpi\circ\xib_1
                            \right)(\zeta,\vartheta).
                    $$
              \enumlab
            \end{assumptionB}
            
                \begin{remark}
                  Assumption \reff{ass: sufficient growth condition} requires extra integrability of the candidate strategy
                  from Assumption \reff{ass: sufficient strong sol}.
                  This enables us to apply dominated convergence along a sequence of localizing stopping times in the verification argument
                  in the proof of Proposition \ref{prop: sufficient assumption} below.
                  
                  Under Assumption \reff{ass: sufficient control of the remainder}, the remainder of the asymptotic expansion can be controlled
                  along the candidate almost optimal strategy.
                  Indeed, this remainder is then of order $\eps^3$, allowing us not only to recover Assumption~(\ref{ass: u locally bounded from above}) but also to prove that the proposed strategy is optimal at the leading order $O(\eps^2)$.
                  In concrete settings, these two assumptions can be verified using estimates on the diffusions driving the control $\dot\theta^\eps$ of
                  Assumption \reff{ass: sufficient strong sol}, compare Section~\ref{sec:final}.
                \end{remark}
            
            \begin{proposition}\label{prop: sufficient assumption}
             Assumption {\rm B} implies Assumption \reff{ass: u locally bounded from above}, Assumption \reff{ass: comparison for u}, with $\mathcal{C}=\mathcal{C}^\chi$, and $u_\ast=u^\ast=u=\hat u$.
            \end{proposition}

            \proof                \emph{Step 1: prove Assumption \reff{ass: u locally bounded from above}.}
                Fix $(\zeta,\vartheta,\veps)\in \Dfi\x\R^d\x(0,\infty)$, set $(X,\theta):=(X^{\zeta,\vartheta,\veps},\theta^{t,\vartheta,\veps})$ and
                $\Upsilon:=(S^\zeta,Y^\zeta,X^{\zeta,\vartheta,\veps},\theta^{t,\vartheta,\veps})$ to ease notation, and
                define the stopping times
                $$
                    \tau^\veps_n:=T\wedge\inf
                        \{
                            u\ge t: \Upsilon_u\notin B_n(\zeta,\vartheta)
                        \}
                        , \quad n \geq 1.
                $$
                By smoothness of $v^0, \theta^0$, and Assumption \reff{ass: sufficient hat u sol}, we have
                $\hat v^\veps\in C^{1,2}(\Df\x\R^d)$. It\^o's formula in turn yields
                $$\bal
                    \hat v^\veps(\zeta,\vartheta)
                    =\; \Esp{
                        \hat v^\veps
                            \left(
                                \tau_n^\veps,\Upsilon_{\tau_n^\veps}
                            \right)
                        -\int_t^{\tau^\veps_n}
                            \left(
                                \Lc^{\theta^\veps}v^\veps
                                +\veps^2\frac{[(D_\xi\varpi)^\top E^{-4}D_\xi\varpi]\circ\xib_\veps}{4\partial_x v^0}
                                +\veps^2\tilde \Rc
                            \right)\left(u,\Upsilon_u\right)du}.
                \eal$$
                In view of Lemma \ref{lem: remainder estimate},
                $$
                    \Lc^\vartheta\hat v^\veps(\zeta,\vartheta)=
                    \left\{
                    \Lc^{\theta^0}v^0
                        +\veps^2
                        \left(
                            \frac12\abs{\xib_\veps^\top \sigma_S}^2\partial_{xx}v^0
                            -\Lc^{\theta^0}\hat u
                            -\frac12\Tr{c_{\theta^0} D^2_{\xi\xi}\varpi\circ\xib_\veps}
                            +\hat \Rc^\veps_\Lc
                        \right)
                    \right\}(\zeta,\vartheta).
                $$
                Now, use the frictionless DPE \eqref{eq: friction less HJB simplified} for $v^0$, the Second Corrector Equation \eqref{eq: 2nd corrector equation} for $\hat u$ (which holds by Assumption \reff{ass: sufficient hat u sol}), and the definition of $\varpi$ (cf.\ Lemma \ref{lem: explicit resolution 1st corrector}), obtaining
                $$\bal
                    \hat v^\veps(\zeta,\vartheta)
                    =\;& \Esp{
                        \hat v^\veps
                            \left(
                                \tau_n^\veps,\Upsilon_{\tau_n^\veps}
                            \right)
                        -\veps^2\int_t^{\tau^\veps_n}
                            \left(
                                \hat \Rc^\veps_\Lc
                                +\tilde \Rc^\veps
                            \right)\left(u,\Upsilon_u\right)du}\\
                        \le& \Esp{
                        \hat v^\veps
                            \left(
                                \tau_n^\veps,\Upsilon_{\tau_n^\veps}
                            \right)
                        }+\veps^3\beta(\zeta,\vartheta),
                \eal$$
                where the inequality follows from \reff{ass: sufficient control of the remainder}.
                In view of \reff{ass: sufficient growth condition} and the terminal condition $\hat{u}(T,\cdot)=0$, dominated convergence in turn yields
                \beq\label{eq: ineq for opt strat}
                    \hat v^\veps(\zeta,\vartheta)\le \Esp{U \left(X^{\zeta,\vartheta,\veps}_T\right)-U'(X^{\zeta,\vartheta,\veps}_T)\Pf(T,\Upsilon_T)}+\veps^3\beta(\zeta,\vartheta)
                        \le v^\veps(\zeta,\vartheta)+\veps^3\beta(\zeta,\vartheta),
                \eeq
                as $n \to \infty$. Here, the last inequality follows from admissibility of the wealth process $X^{\zeta,\vartheta,\veps}$ (cf.~Assumption \reff{ass: sufficient strong sol}) and the definition of the frictional value function \eqref{eq: definition friction value function}.
                By definition of $\bar u^\veps$ in \reff{eq:baru},
                \eqref{eq: ineq for opt strat} gives
                \beq\label{eq: u eps bounded bay hat u}
                    \bar u^\veps(\zeta,\vartheta)\le (\hat u+\veps \beta + \varpi\circ\xib_1)(\zeta,\vartheta).
                \eeq
                Assumption \reff{ass: u locally bounded from above} in turn follows from the continuity of $\hat u$, $\beta$, and $\varpi$.

                \emph{Step 2: show that Assumption \reff{ass: comparison for u} holds, and $u_\ast=u^\ast=u=\hat u$.}

                    Let $\tilde u\in C^{1,2}(\Df)\cap\Cc^\chi$ be a classical solution of \reff{eq: 2nd corrector equation}, and let
                    $u_1\in\Cc^\chi$ (resp. $u_2\in\Cc^\chi$) be a lower-(resp. upper-) semicontinuous viscosity supersolution (resp. subsolution)
                    of \reff{eq: 2nd corrector equation} such that $u_1\ge \tilde{u} \ge u_2$ on $\Dfb$.
                    We prove that $u_1\ge\tilde u$ on $\Df$; the inequality $\tilde u\le u_2$ is obtained similarly.

                    Assume to the contrary that there is $\hat\zeta\in \Dfi$ such that $(u_1-\tilde u)(\hat\zeta)<0$.
                    For $\kappa>0$ small enough, we then have $(u_1-\tilde u+\kappa\chi)(\hat\zeta)<0$.
                    As, moreover, the definition of $\Cc^\chi$ in \reff{eq: growth for comparison} implies $(u_1-\tilde u+\kappa\chi)>0$ near the spatial boundary of $\Df$, it follows that there is $\zeta_\kappa\in \Df$ such that
                    $$
                        \min_\Df(u_1-\tilde u+\kappa\chi)=(u_1-\tilde u+\kappa\chi)(\zeta_\kappa)\le (u_1-\tilde u+\kappa\chi)(\hat\zeta)<0.
                    $$
                    As $u_1\ge\tilde u$ on $\Dfb$, $\zeta_\kappa\in\Dfb$ would imply $\chi(\zeta_\kappa)<0$, which contradicts $\chi \geq 0$ in \eqref{ass: sufficient for comparison chi}. Therefore, $\zeta_k$ is an interior minimum of $u_1-(\tilde{u}-\kappa \chi)$, and the viscosity supersolution property of $u_1$ gives
                                      $
                        -\Lc^{\theta^0}(\tilde u-\kappa\chi)(\zeta_\kappa)\ge a.
                    $
                    Because $\tilde u$ is a classical solution of $-\Lc^{\theta^0}\tilde u=a$, it follows that $\Lc^{\theta^0}\chi\ge0$,
                    which contradicts \reff{ass: sufficient for comparison chi}.
                    Thus, $u_1\ge \tilde u$ on $\Df$ as claimed.

                    Applying \reff{eq: u eps bounded bay hat u} to any subsequence $(\zeta_\veps,\vartheta_\veps)$
                    and using $\hat{u} \in \Cc^\chi$ (cf.\ \reff{ass: sufficient functions in class Cc}) yields
                    $u_\ast, u^\ast\in\Cc^\chi$.
                    As the classical solution $\hat{u}$ is also a viscosity solution of \eqref{eq: 2nd corrector equation}, Propositions \ref{prop: sub sol for u}, \ref{prop: super sol for u}, and the comparison result established above show that $u_\ast\ge\hat u\ge u^\ast$.  As $u^\ast\ge u_\ast$ by definition, this shows $\hat u=u_\ast=u^\ast$.

                    The function $u$ defined in \reff{eq: definition of u} is locally bounded because $u \in \Cc^\chi$ and $\chi \in C^{1,2}$. Hence, $u$ is a viscosity solution of \eqref{eq: 2nd corrector equation}, and it follows as above that $u=\hat{u}=u^\ast=u_\ast$.
           \ep\\

           As a corollary, we obtain our second main result, Theorem \ref{theo: main result 2}:

        \begin{corollary}\label{cor: optimal strat}
          Under Assumptions \ref{prop:dpefric} and {\rm B}, the investment strategy $\dot{\theta}^\veps$ defined in \reff{ass: sufficient strong sol}
          is optimal at the leading order $O(\varepsilon^2)$. That is, for each compact subset $B$ of $\Df\x\R^d$ and $\veps>0$,
          there is a constant $K_B^\veps>0$ such that $K_B^\veps\rightarrow0$ as $\veps\rightarrow0$ and
          $$
            v^\veps(\zeta,\vartheta)-\veps^2 K^\veps_B\le \Esp{U\left(X^{\zeta,\vartheta,\veps}_T \right)-U'(X^{\zeta,\vartheta,\veps}_T)\Pf(T,S_T^\zeta,Y_T^\zeta,X_T^{\zeta,\vartheta,\veps})},
            \quad\pourtout(\zeta,\vartheta)\in B\mbox{ and }\veps>0,
          $$
          where $(X^{\zeta,\vartheta,\veps},\theta^{t,\vartheta,\veps})$ is defined as in \reff{ass: sufficient strong sol}.
        \end{corollary}

        \proof

            In the proof of Proposition \ref{prop: sufficient assumption}, we have shown \reff{eq: ineq for opt strat}:
            $$
                    v^0(\zeta)-\veps^2u(\zeta)-\veps^2\varpi\circ\xib_1(\zeta,\vartheta)-\veps^3 \beta(\zeta,\vartheta)
                        \le \Esp{U \left(X^{\zeta,\vartheta,\veps}_T \right)-U'(X^{\zeta,\vartheta,\veps}_T)\Pf(T,S_T^\zeta,Y_T^\zeta,X_T^{\zeta,\vartheta,\veps})}.
            $$
            This corollary thus follows from
            the local uniform convergence of $\bar u^\lambda$ shown in
            Theorem \ref{theo: main result}.        \ep

\section{Examples}\label{sec:final}

In this section we show how all of our technical assumptions can be verified in concrete settings. For the sake of clarity, we do not strive for minimal assumptions. Throughout, we consider an investor with an exponential utility function $-e^{-\eta x}$ with constant absolute risk aversion $\eta>0$. 

\subsection{Portfolio Choice}\label{ex:pc}

First we focus on a portfolio choice problem. There is a single risky asset with dynamics\footnote{This specification allows for predictable returns as in \cite{bouchaud.al.12,martin.12,garleanu.pedersen.13a,garleanu.pedersen.13b,dufresne.al.12}. To ensure enough integrability for a rigorous verification theorem, we truncate large values of the state variable by assuming boundedness of all coefficients. Nonlinear dynamics and stochastic volatility can be handled without difficulties.}
$$dS_t=\mu_S(Y_t)dt+\sigma_S(Y_t)dW^1_t,$$
driven by a one-dimensional autonomous diffusion:
$$dY_t=\mu_Y(Y_t)dt+\sigma_Y(Y_t)d\left(\rho W^1_t+\sqrt{1-\rho^2}W^2_t\right).$$
Here, $W=(W^1,W^1)$ is a two-dimensional standard Brownian motion, $\rho \in [-1,1]$, and the mappings $\mu_S$, $\mu_Y$, $\sigma_S$, $\sigma_Y: \mathbb{R} \longmapsto \mathbb{R}$ all are bounded and smooth, with bounded derivatives of all orders, and the volatilities $\sigma_S, \sigma_Y$ are bounded away from zero. Then, $Y$ and in turn $S$ are well defined and it follows similarly as in \cite{zariphopoulou.01} that the frictionless value function $v^0$ is a classical solution of the frictionless DPE, which can be transformed into a linear, uniformly parabolic equation in this case. The value function $v^0$ can be written as
\begin{equation}\label{eq:v0}
v^0(t,y,x)=e^{-\eta x}w^0(t,y),
\end{equation}
and the corresponding optimal policy is given by
$$\theta^0_t=\theta^0(t,Y_t)=\frac{\mu_S(Y_t)}{\eta\sigma^2_S(Y_t)} +\frac{\rho\sigma_Y(Y_t)}{\eta\sigma_S(Y_t)}\frac{\partial_y w^0(t,Y_t)}{w^0(t,Y_t)}.$$
Similarly as in \cite[Theorem 3.1]{zariphopoulou.01}, one verifies that $w^0, \theta^0$ are also bounded and smooth, with bounded derivatives of all orders.\footnote{For $w^0$, this follows from the corresponding Feynman-Kac representation. As all coefficients are smooth, one can then differentiate the PDE for $w^0$ and argue analogously for all of its derivatives.} In particular, all regularity assumptions imposed on the frictionless problem in Section \ref{sec: condition for Assumption B} are satisfied. Moreover, it follows from Novikov's condition and Girsanov's theorem that $\partial_x v^0(t,Y_t,X^{\theta^0}_t)/\partial_x v^0(0,y,x)$ is the density process of an equivalent martingale measure $\mathbb{Q}$, the dual minimizer for the optimization problem at hand.

Now, consider constant linear price impact, $\Lambda_t=\lambda=\varepsilon^4>0$. Then, all of our technical assumptions hold and we have the following result:

\begin{theorem}\label{thm:pc}
In the setting of Section \ref{sec:final}, Assumptions \ref{prop:dpefric} and {\rm B} are satisfied, so that Theorems~\ref{theo: main result} and \ref{cor: optimal strat} are applicable. As a consequence, a leading-order optimal policy with small constant price impact $\Lambda_t=\lambda=\varepsilon^4$ is given in feedback form as
\begin{equation}\label{eq:rateBachelier}
\dot{\theta}^\varepsilon_t=\sqrt{\frac{\eta\sigma_S^2(Y_t)}{2\varepsilon^4}}(\theta^0_t-\theta^\varepsilon_t).
\end{equation}
The corresponding first-order correction of the value function reads as:
\begin{align*}
v^\varepsilon(t,y,x,\vartheta) = v^0\Big(t,y,x-\mathrm{CE}(t,y,\vartheta)\Big) +o(\varepsilon^2),
\end{align*}
where
$$\mathrm{CE}(t,y,\vartheta)=\frac{\varepsilon^2}{\sqrt{2\eta}}\left(\mathbb{E}_\mathbb{Q}\left[\int_t^T \Big(\partial_y \theta^0(Y^{t,y}_r)^2 \sigma_Y(Y^{t,y}_r)^2 \sigma_S(Y^{t,y}_r)\Big)dr\right]+\sigma_S(y)(\theta^0(0,y)-\vartheta)^2\right).$$
\end{theorem}

\begin{proof}
Because no state constraints are needed for exponential utility, (weak) dynamic programming and in turn the viscosity solution property of the frictional value function (Assumption \ref{prop:dpefric}) can be derived along the lines of Bouchard and Touzi \cite{bouchard.touzi.11}.

Let us now verify Assumption {\rm B}. First, note that -- due to boundedness and smoothness of all coefficient functions -- it follows from dominated convergence and It\^o's formula that the probabilistic representation \eqref{eq: definition of u} is a classical solution of the Second Corrector Equation \eqref{eq: 2nd corrector equation}. In  particular, \eqref{ass: sufficient hat u sol} is satisfied. Next, one readily verifies that \eqref{ass: sufficient for comparison chi} and \eqref{ass: sufficient functions in class Cc} also hold with
$\chi(t,y,x)=e^{-at}\left(e^{-y}+e^y+v^0(t,y,x)^2\right)$,
if $a$ is chosen sufficiently large. The feedback policy $\dot{\theta}^\varepsilon$ from \eqref{eq:rateBachelier} implies that the corresponding number $\theta^\varepsilon$ of risky shares solves a (random) linear ODE. It is therefore given explicitly by
$$\theta^{t,\vartheta,\varepsilon}=e^{-\int_t^\cdot \sqrt{\eta\sigma_S^2(Y_r)/2\varepsilon^4}dr}\left( \vartheta+\int_t^\cdot \Big(e^{\int_t^r \sqrt{\eta\sigma_S^2(Y_s)/2\varepsilon^4}ds}\sqrt{\eta\sigma_S^2(Y_r)/2\varepsilon^4}\theta^0(r,Y_r)\Big)dr\right).$$
Hence, $\theta^\varepsilon$ is well defined and uniformly bounded. As a result, the corresponding wealth process \eqref{eq:wealth} is well defined, too, and the corresponding utility \eqref{eq: admiss for friction control} is integrable by Novikov's condition and the boundedness of $\theta^\varepsilon$, $\theta^0$, $\mu_S$, and $\sigma_S$. Moreover, dominated convergence shows that the corresponding wealth process can be approximated by simple strategies as in \cite{biagini.cerny.11}. In summary, \eqref{ass: sufficient strong sol} is satisfied.

Now, turn to \eqref{ass: sufficient growth condition}. By \eqref{eq:v0}, \eqref{eq: definition of u}, and Lemma \ref{lem: explicit resolution 1st corrector}, we can choose 
$\gamma^\varepsilon(x)=G e^{-\eta x}$
for a suitable constant $G>0$, because the quadratic trading cost, the risky asset's volatility, the investor's absolute risk aversion, the frictionless reduced value function $w^0$, and the quadratic variation of the frictionless trading strategy $\theta^0$ are all uniformly bounded. The frictional wealth process $X^{\zeta,\vartheta,\varepsilon}_t$ is an It\^o process with bounded drift and diffusion coefficients. Hence, it follows from Novikov's condition and Doobs maximal inequality that it's running supremum has exponential moments of all orders, verifying Assumption~\eqref{ass: sufficient growth condition}.

 \eqref{ass: sufficient control of the remainder} is derived along the same lines by also taking into account that $\mathbb{E}\left[\int_t^T |\theta^0_r-\theta^\veps_r|^2/\veps^2 dr\right]$ is uniformly bounded in $\veps>0$. To see this, first notice that
$$\theta^0-\theta^\veps=e^{-\veps^{-2}\int_t^\cdot \sqrt{\eta\sigma_S^2(Y_r)/2}dr}(\theta^0(\zeta)-\vartheta)+\int_t^\cdot e^{-\veps^{-2}\int_r^\cdot \sqrt{\eta\sigma_S^2(Y_s)/2}ds}d\theta^0_r,$$
by \eqref{eq:rateBachelier} and the explicit formula for solutions of linear SDEs (cf., e.g., \cite[Theorem V.52]{protter.05}). Recall that the drift and diffusion coefficients of the frictionless optimizer $\theta^0$ are uniformly bounded by constants $M,\Sigma>0$, and that $\sqrt{\eta\sigma_S^2(\cdot)/2}$ is uniformly bounded away from zero by some constant $C>0$. Hence it follows from the algebraic inequality $(x+y)^2 \leq 2x^2+2y^2$, Jensen's inequality, the It\^o isometry, and a simple integration that
$$
\mathbb{E}\left[\int_t^T \frac{|\theta^0_r-\theta^\veps_r|^2}{\veps^2}dr\right] \leq \frac{|\theta^0(\zeta)-\vartheta|^2}{C} +\frac{2(M^2 T^2+\Sigma^2 T)}{C},
$$
establishing the claimed uniform bound in $\veps>0$. In summary, Assumption {\rm B} is satisfied and the leading-order optimality of the trading rate \eqref{eq:rateBachelier} follows from Theorem \ref{cor: optimal strat}. The representation for the leading-order correction of the corresponding value function is a consequence of Theorem \ref{theo: main result}, Proposition \ref{prop: sufficient assumption}, as well as Taylor expansion and the definition of $\mathbb{Q}$.
\end{proof}

\subsection{Random Endowments}

Similar arguments can be used to verify the regularity assumptions needed  to apply the general argument from Section \ref{sec:options} to deal with random endowments. To illustrate this, consider the Bachelier model 
$$dS_t=\mu dt+\sigma dW_t,$$
for a standard Brownian motion $W$, and a European option with payoff function $H=h(S_T)$. If the function $h: \mathbb{R}_+ \to \mathbb{R}$ is bounded and smooth, with bounded and smooth derivative of all orders,\footnote{Weakening these regularity assumptions to European call and put options, for example, is an open problem even in simpler models with proportional transaction costs \cite{bichuch.13,PoRo13}.} then it follows from the Markov property that the density process $Z^H_t$ generated by the Radon-Nikodym derivative 
$d\mathbb{P}^H/d\mathbb{P}=e^{-\eta h(S_T)}/\mathbb{E}[e^{-\eta h(S_T)}]$
is given by a smooth function $f(t,S_t)$ which solves  
$$ \partial_t f(t,s)+\mu \partial_s f(t.s) +\frac{\sigma^2}{2} \partial_{ss}f(t,s)=0, \quad f(T,s)=\frac{e^{-\eta h(s)}}{\int_{-\infty}^\infty e^{-\eta h(\mu T+\sigma \sqrt{T}s')} \phi(s') ds']},$$
where $\phi$ denotes the density function of the standard Normal distribution. Due to our assumptions on $h$, the function $f$ is smooth, bounded, and bounded away from zero; by the dominated convergence theorem, the same holds for all of its derivatives. As a result, It\^o's formula shows that the dynamics of the density process $Z^H_t$ are given by $dZ^H_t/Z^H_t= (\partial_s f(t,S_t)/f(t,S_t))\sigma dW_t$.
Girsavov's theorem in turn yields the dynamics of the risky asset $S$ under the measure $\mathbb{P}^H$: 
$$dS_t= \left(\mu+ \frac{\partial_s f(t,S_t)}{f(t,S_t)}\sigma^2 \right)dt+\sigma dW^H_t,$$
for a $\mathbb{P}^H$-Brownian motion $W^H$. Due to the regularity of $f$ and its derivatives, the regularity assumptions of Section \ref{ex:pc} are satisfied. As a consequence, the portfolio choice problem with the random endowment $H=h(S_T)$ is equivalent to the pure investment problem under the measure $\mathbb{P}^H$, whose solution is provided by Theorem \ref{thm:pc}. Utility-based prices and hedging strategies can in turn be computed using the indifference argument of Hodges and Neuberger \cite{hodges.neuberger.89}.

\bibliographystyle{abbrv}
\bibliography{mms}

\end{document}